\RequirePackage[l2tabu,orthodox]{nag}
\documentclass
[11pt,letterpaper]
{article} 

\usepackage[dvipsnames]{xcolor}

\usepackage[notes=false,later=false,camera=true,draft=false]{dtrt}
\usepackage[utf8]{inputenc}
\usepackage{etex}
\usepackage{ stmaryrd }
\usepackage{xspace,enumerate}
\usepackage[T1]{fontenc}
\usepackage[full]{textcomp}
\usepackage[american]{babel}
\usepackage{mathtools}

\usepackage{amsthm}
\usepackage{empheq}

   \usepackage{hyperref}
   
\hypersetup{
colorlinks=true,
urlcolor=Cerulean,
linkcolor=RoyalBlue,
citecolor=OliveGreen,
linktocpage=true,
}
\usepackage[capitalise,nameinlink]{cleveref}
\crefname{lemma}{Lemma}{Lemmas}
\crefname{fact}{Fact}{Facts}
\newcommand{\colorconstraints}{\text{Color Constraints}}
\crefname{colorconstraints}{(color constraints)}{Color Constraints}
\crefformat{colorconstraints}{#2\colorconstraints#3}
\crefname{indsetconstraints}{(indset constraints)}{IndSet Constraints}
\crefformat{indsetconstraints}{#2$\mathsf{IndSet\ Axioms}$#3}
\crefname{theorem}{Theorem}{Theorems}
\crefname{mtheorem}{Theorem}{Theorems}
\crefname{corollary}{Corollary}{Corollaries}
\crefname{claim}{Claim}{Claims}
\crefname{example}{Example}{Examples}
\crefname{algorithm}{Algorithm}{Algorithms}
\crefname{problem}{Problem}{Problems}
\crefname{definition}{Definition}{Definitions}
\usepackage{paralist}
\usepackage{turnstile}
\usepackage{mdframed}
\usepackage{tikz}
\usepackage{caption}
\DeclareCaptionType{Algorithm}
\usepackage{newfloat}
\newtheorem{theorem}{Theorem}[section]
\newtheorem{mtheorem}{Theorem}
\newtheorem*{theorem*}{Theorem}

\newtheorem*{proposition*}{Proposition}
\newtheorem{lemma}[theorem]{Lemma}
\newtheorem*{lemma*}{Lemma}

\newtheorem*{conjecture*}{Conjecture}
\newtheorem{fact}[theorem]{Fact}
\newtheorem*{fact*}{Fact}

\newtheorem*{hypothesis*}{Hypothesis}

\theoremstyle{definition}
\newtheorem{definition}[theorem]{Definition}
\newtheorem*{definition*}{Definition}

\newtheorem{algorithm}[theorem]{Algorithm}

\theoremstyle{remark}
\newtheorem{claim}[theorem]{Claim}
\newtheorem*{claim*}{Claim}
\newtheorem{remark}[theorem]{Remark}
\newtheorem*{remark*}{Remark}
\newtheorem{observation}[theorem]{Observation}
\newtheorem*{observation*}{Observation}

\usepackage[
letterpaper,
top=1.2in,
bottom=1.2in,
left=1in,
right=1in]{geometry}

\usepackage{mathpazo}
\usepackage{textcomp} 
\usepackage[varg,bigdelims]{newpxmath}
\usepackage{bm} 
\let\mathbb\varmathbb
\usepackage{microtype}
\definecolor{petergreen}{rgb}{0, 0.75, 0}

\usepackage{thm-restate}

\allowdisplaybreaks
\newcommand{\FormatAuthor}[3]{
\begin{tabular}{c}
#1 \\ {\small\texttt{#2}} \\ {\small #3}
\end{tabular}
}

\newcommand{\keywords}[1]{\bigskip\par\noindent{\footnotesize\textbf{Keywords\/}: #1}}

\newcommand{\R}{{\mathbb R}}
\newcommand{\N}{{\mathbb N}}
\newcommand{\norm}[1]{\lVert #1 \rVert}
\newcommand{\boolnorm}[1]{{\lVert #1 \rVert}_{\infty \to 1}}
\newcommand{\abs}[1]{\lvert #1 \rvert}
\newcommand{\Abs}[1]{\left \lvert #1 \right \rvert}

\newcommand{\eps}{\varepsilon}
\newcommand{\F}{{\mathbb F}}

\newcommand{\E}{{\mathbb E}}

\newcommand{\C}{\mathbb C}
\newcommand{\Bits}{\{0,1\}}

\newcommand{\Fits}{\{-1,1\}}

\newcommand{\cH}{\mathcal H}

\newcommand{\cB}{\mathcal B}

\newcommand{\poly}{\mathrm{poly}}
\newcommand{\val}{\mathrm{val}}

\newcommand{\mper}{\,.}
\newcommand{\mcom}{\,,}
\newcommand{\Paren}[1]{\left(#1\right)}

\newcommand{\Norm}[1]{\left\lVert#1\right\rVert}

\newcommand{\Set}[1]{\left\{#1\right\}}

\newcommand{\polylog}{\mathrm{polylog}}
\newcommand{\cD}{\mathcal D}
\newcommand{\cR}{\mathcal R}

\newcommand{\cV}{\mathcal{V}}

\newcommand{\Code}{\mathcal{L}}
\newcommand{\cJ}{\mathcal{J}}
\newcommand{\cT}{\mathcal{T}}
\newcommand{\Dec}{\text{Dec}}

\newcommand{\Deg}{\mathsf{Deg}}

\begin{document}

\title{An Exponential Lower Bound for Linear $3$-Query Locally Correctable Codes}
\author{
\begin{tabular}[h!]{ccc}
      \FormatAuthor{Pravesh K.\ Kothari}{praveshk@cs.cmu.edu}{Carnegie Mellon University}
      \FormatAuthor{Peter Manohar}{pmanohar@cs.cmu.edu}{Carnegie Mellon University}
\end{tabular}
} %
\date{\today}

\maketitle

\begin{abstract}
We prove that the blocklength $n$ of a linear $3$-query locally correctable code (LCC) $\Code \colon \F^k \to \F^n$ with distance $\delta$ must be at least $n \geq 2^{\Omega\left(\left(\frac{\delta^2 k}{(\abs{\F}-1)^2}\right)^{1/8}\right)}$. In particular, the blocklength of a linear $3$-query LCC with constant distance over any small field grows \emph{exponentially} with $k$. This improves on the best prior lower bound of $n \geq \tilde{\Omega}(k^3)$~\cite{AlrabiahGKM23}, which holds even for the weaker setting of $3$-query locally \emph{decodable} codes (LDCs), and comes close to matching the best-known construction of $3$-query LCCs based on binary Reed--Muller codes, which achieve $n \leq 2^{O(k^{1/2})}$.  Because there is a $3$-query LDC with a strictly subexponential blocklength~\cite{Yekhanin08, Efremenko09}, as a corollary we obtain the first strong separation between $q$-query LCCs and LDCs for any constant~$q \geq 3$. 

Our proof is based on a new upgrade of the method of spectral refutations via \emph{Kikuchi matrices} developed in recent works \cite{GuruswamiKM22, HsiehKM23, AlrabiahGKM23} that reduces establishing (non-)existence of combinatorial objects to proving unsatisfiability of associated XOR instances. Our key conceptual idea is to apply this method with XOR instances obtained via \emph{long-chain derivations} --- a structured variant of low-width resolution for XOR formulas from proof complexity~\cite{Grigoriev01,Schoenebeck08}.
\end{abstract}

\clearpage
 \microtypesetup{protrusion=false}
  \tableofcontents{}
  \microtypesetup{protrusion=true}

\clearpage

\pagestyle{plain}
\setcounter{page}{1}

\section{Introduction}
A \emph{locally correctable code} (LCC) is an error correcting code that admits, in addition, a \emph{local} correction (a.k.a. \emph{self correction}) algorithm that can recover any symbol of the original codeword by querying only a small number of randomly chosen symbols from the received corrupted codeword. More formally, we say that a code $\Code \colon \Bits^k \to \Bits^n$ is $q$-locally correctable if for any codeword $x$, a corruption $y$ of $x$, and input $u \in [n]$, the local correction algorithm reads at most $q$ symbols (typically a small constant such as $2$ or $3$) of $y$ and recovers the bit $x_u$ with probability $1/2 + \eps$ whenever $\Delta(x, y) \coloneqq \abs{\{v \in [n] : x_v \ne y_v\}} \leq \delta n$, where $\delta$, the ``distance'' of the code, and $\eps$, the decoding accuracy, are constants. The central question about LCCs is to determine the smallest possible blocklength $n$ as a function of the message length $k$ for a fixed number of queries $q$. 

Local correction was first introduced for \emph{program checking}~\cite{BlumK95}, and early applications utilized that Reed--Muller codes are locally correctable via polynomial interpolation. Since then, LCCs have been a mainstay in complexity and algorithmic coding theory with a long array of applications. An abridged list (the surveys~\cite{Trevisan04,Yekhanin12,Dvir12} provide details) of applications includes sublinear algorithms and property testing~\cite{RubinfeldS96,BlumLR93}, probabilistically checkable proofs~\cite{AroraLMSS98,AroraS98}, \textsc{IP}=\textsc{PSPACE}~\cite{LundFKN90,Shamir90}, worst-case to average-case reductions~\cite{BabaiFNW93}, constructions of explicit rigid matrices~\cite{Dvir10}, and $t$-private information retrieval protocols~\cite{IshaiK99,BarkolIW10}. The existence of LCCs turns out to have natural connections to incidence geometry~\cite{Dvir12}, additive combinatorics~\cite{BhowmickDL13}, and the theory of block designs~\cite{BarkolIW10}. 

For any constant $q \in \N$, Reed--Muller codes (i.e., evaluations of $(q-1)$-degree polynomials) yield binary, linear\footnote{A code is \emph{linear} over a field $\F$ if the encoding map $\Code$ is an $\F$-linear map.} $q$-LCCs with a blocklength $n \leq 2^{O(k^{\frac{1}{q-1}})}$. Given their extensive applications and connections, finding LCCs of smaller blocklength has been a major project in theoretical computer science over the past three decades with some remarkable  successes over the years. For example, \emph{multiplicity codes}~\cite{KoppartySY14} significantly beat the blocklength of Reed--Muller codes in the \emph{super-constant} query regime. In the constant-query regime, \emph{matching vector codes}~\cite{Efremenko09,Yekhanin08} use a strictly sub-exponential (i.e., $n \leq \exp (\exp (O(\sqrt{\log k} \log \log k)))$) blocklength to obtain $3$-query locally \emph{decodable} codes --- a relaxation of LCCs where the local correction property holds only for the $k$ message bits. To sidestep the difficulty of finding more efficient LCCs, the work of~\cite{Ben-SassonGHSV04} introduced \emph{relaxed LCCs} that soften the local correction property and has seen exciting recent developments~\cite{GurRR20,AsadiS21,ChiesaGS20, KumarM23, CohenY23}. These successes notwithstanding, constructing better constant-query LCCs has remained a major open question (see, e.g., Chapter 8 in~\cite{Yekhanin12}). 

\parhead{LCC lower bounds.} The lack of progress on finding better constant-query LCCs has motivated a long investigated conjecture that Reed--Muller codes might be \emph{optimal} constant query LCCs. The work of~\cite{KerendisW04, GoldreichKST06} essentially confirmed this conjecture for the ``base case'' of $q=2$ by proving that $n \geq 2^{\Omega(k)}$ for any two-query LCC, matching the construction of Hadamard codes, which are $2$-LCCs with $n=2^k$. For $q \geq 3$, however, only a polynomial lower bound is known. The works of~\cite{KerendisW04,Woodruff07} prove that $q$-LCCs must have $n \geq \tilde{\Omega}(k^{1/(1 - 1/\lceil\frac{q}{2} \rceil})$,\footnote{These lower bounds all hold for non-linear codes over small (i.e., $\polylog(k)$) size alphabets. A weaker polynomial lower bound~\cite{KatzT00,IcelandS18} is known to hold for linear codes over all fields and for the specific case of $q = 3$,~\cite{Woodruff10} shows a lower bound of $\Omega(k^2)$ for linear $3$-LDCs over all fields.} and for the specific case of $q = 3$, a recent work~\cite{AlrabiahGKM23} (which, like this work, is based on the \emph{Kikuchi} matrix method) obtained a polynomial improvement on this bound, showing that $n \geq \tilde{\Omega}(k^3)$.

\parhead{Limitations of prior lower bound techniques} Beyond the weakness in the quantitative results, all the above lower bounds suffer from an important inherent limitation --- they all hold even for the weaker setting of locally decodable codes (LDCs). As we mentioned above, there are sub-exponential length (and thus substantially beating Reed--Muller) $3$-query binary, linear codes that are locally decodable~\cite{Yekhanin08, Efremenko09}. Indeed, characterizing the limitations of prior proof techniques and finding methods that could separate LCCs and LDCs itself has been a major research goal. For example, Dvir, Gopi, Gu and Wigderson~\cite{DvirGGW19} formalize the limitations of prior lower bound techniques for LCCs by showing that the ``random restriction'' approach in~\cite{KatzT00} applies to a more general setting of ``spanoids'' where they are, in fact, tight. On the other hand, to show a strong separation between LCCs and LDCs, Barkol, Ishai and Weinreb~\cite{BarkolIW10} build an approach for stronger LCC lower bounds via connections to the well-studied Hamada conjecture (\cite{Hamada73}, see lecture notes~\cite{Tonchev11}) and its generalizations in the theory of block designs, while Dvir, Saraf and Wigderson~\cite{DvirSW14} develop new geometric techniques to prove a slightly superquadratic lower bound for an appropriate formulation of $3$-LCCs over the reals. 

To summarize: there is an exponential gap between best-known constructions and lower bounds for $q$-LCCs for $q\geq 3$. Further, the best known lower bound techniques for $q$-LCCs apply also to $q$-LDCs and thus provably cannot yield an exponential lower bound.

\parhead{Our result.} In this work, we prove an \emph{exponential} lower bound for linear $3$-query LCCs. We note that the best-known constructions of LCCs (and also LDCs) namely Reed--Muller codes and matching vector codes, are $\F_2$-linear.
\begin{mtheorem}
\label{mthm:main}
Let $\Code \colon \F^k \to \F^n$ be a linear $(3, \delta, \eps)$-LCC. Then, $n \geq 2^{\Omega((\delta^2 k/(\abs{\F} - 1)^2)^{1/8})}$.
In particular, if $\Code \colon \F_2^k \to \F_2^n$ is a $(3, \delta, \eps)$-LCC where $\delta$ is constant, then $n \geq 2^{\Omega(k^{1/8})}$.
\end{mtheorem}
\cref{mthm:main} improves on the prior best lower bound of $n \geq \tilde{\Omega}(k^3)$~\cite{AlrabiahGKM23} and comes close to matching the blocklength $n = \exp(O(\sqrt{k}))$ of $3$-query LCCs based on Reed--Muller codes; in \cref{sec:discussion}, we comment on potential strengthenings of our argument to come closer and even match (up to constants in the exponent) the $\exp(O(\sqrt{k}))$ bound. 

\cref{mthm:main} also yields the first strong \emph{separation} between $3$-LCCs and $3$-LDCs. No such separation was known for $q$-LDCs and $q$-LCCs for any constant $q \geq 3$.\footnote{The work of~\cite{BhattacharyyaGT17} shows a separation between $2$-LCCs and $2$-LDCs over $\poly(n)$-sized alphabets. For $2$-LCCs on small alphabets, a strong separation cannot exist, e.g., on $\F_2$, the Hadamard code gives both an essentially optimal $2$-LCC and $2$-LDC.} In particular, \cref{mthm:main} implies that matching vector codes that yield linear $3$-LDCs over $\F_2$ of sub-exponential blocklength, such as the codes in~\cite{Yekhanin08,Efremenko09}, \emph{cannot} admit a local correction algorithm, answering a question of Yekhanin (see Chapter 8 in~\cite{Yekhanin12}).

Our proof is based on the method of spectral refutation via Kikuchi matrices developed in prior works~\cite{GuruswamiKM22,HsiehKM23,AlrabiahGKM23}. The key idea in this method is to associate the existence of a combinatorial object (e.g., a $3$-LCC) to the satisfiability of a family of XOR formulas and find a \emph{spectral} refutation (i.e., certificate of unsatisfiability) for a randomly chosen member of the family.  

Our key new conceptual idea is to apply an appropriate version of the Kikuchi matrix method to XOR formulas obtained by \emph{long chain derivations} --- a structured variant of low-width XOR resolution refutations in proof complexity~\cite{Grigoriev01,Schoenebeck08} --- to the naive XOR instances obtained from the query sets of a purported linear $3$-LCC. These new XOR formulas allow us to utilize the additional structural in $3$-LCCs and, in particular, significantly surpass the cubic lower bound~\cite{AlrabiahGKM23} for $3$-LDCs that also used the Kikuchi matrix method. We discuss the new challenges that arise in analyzing spectral refutations of XOR instances produced by such long chain derivations and our technical ideas for handling them in \cref{sec:techniques,sec:warmup}.

\subsection{Roadmap}
The rest of the paper is organized as follows. First, in \cref{sec:prelims}, we introduce some notation and recall basic facts about LCCs that we shall use in the proof. Then, in \cref{sec:techniques}, we give a detailed overview of the proof. In \cref{sec:warmup}, we give an essentially complete proof of a new lower bound of $n \geq \tilde{\Omega}(k^{4})$ for binary linear $3$-LCCs as a warmup. Following the warmup, in~\cref{sec:lcctoxor,sec:regular-partition,sec:kikuchimethod,sec:rowpruning} we prove \cref{mthm:main} for binary $3$-LCCs, i.e., when $\F = \F_2$; we handle the case of arbitrary finite fields in~\cref{sec:largeralphabet}. Finally, in \cref{sec:discussion} we conclude with some remarks on the proof of \cref{mthm:main}, possible strengthenings, and extensions. 
\section{Preliminaries}
\label{sec:prelims}

\subsection{Basic notation}
We let $[n]$ denote the set $\{1, \dots, n\}$. For two subsets $S, T \subseteq [n]$, we let $S \oplus T$ denote the symmetric difference of $S$ and $T$, i.e., $S \oplus T \coloneqq \{i : (i \in S \wedge i \notin T) \vee (i \notin S \wedge i \in T)\}$. For a natural number $t \in \N$, we let ${[n] \choose t}$ be the collection of subsets of $[n]$ of size exactly $t$. Given variables $x_1, \dots, x_n$ and a subset $C \subseteq [n]$, we let $x_C \coloneqq \prod_{v \in C} x_v$.

For a rectangular matrix $A \in \R^{m \times n}$, we let $\norm{A}_2 = \coloneqq \max_{x \in \R^m, y \in \R^n: \norm{x}_2 = \norm{y}_2 = 1} x^{\top} A y$ denote the spectral norm of $A$, and $\boolnorm{A} \coloneqq \max_{x \in \Fits^m, y \in \Fits^n} x^{\top} A y$. We note that $\boolnorm{A} \leq \sqrt{nm} \norm{A}_2$.

\subsection{XOR formulas}
An XOR instance $\psi$ on $n$ variables $x_1, x_2,\ldots, x_n$ taking values in $\Fits$ is a collection of constraints of the form $\{x_C = b_C\}$ where $C \in \cH$ where $\cH \subseteq 2^{[n]}$ is the \emph{constraint hypergraph}. The \emph{arity} of a constraint $\{x_C = b_C\}$ equals $|C|$. The arity of $\psi$ is the maximum arity of any constraint in it. The XOR formula associated with $\psi$ is the expression $\psi(x) = \sum_{C \in \cH} b_C x_C$ seen as a polynomial over $\Fits^n$. Notice that $\psi(x) = |\cH|$ if $x$ satisfies all the constraints of $\psi$ and in general evaluates to (number of constraints satisfied by $x$) - (number of constraints violated by $x$). The \emph{value} $\val(\psi)$ of a XOR instance $\psi$ (or, of the associated formula $\psi(x)$) is the maximum of $\psi(x)$ as $x$ ranges over $\Fits^n$. More generally, for a function $f(x)$, we shall let $\val(f) \coloneqq \max_{x \in \Fits^n} f(x)$.
\subsection{Locally correctable codes}
We refer the reader to the survey~\cite{Yekhanin12} for background. 

\begin{definition}[Locally correctable code]
\label{def:LCC}
A map $\Code \colon \F^k \to \F^n$ is a $(q, \delta, \eps)$-locally correctable code if there exists a randomized decoding algorithm $\Dec(\cdot)$ that takes input an oracle access to some $y \in \F^n$ and a $u \in [n]$, \begin{inparaenum}[(1)] \item makes at most $q$ queries to the string $y$, and \item for all $b \in \F^k$, $u \in [n]$, and all $y \in \F^n$ such that $\Delta(y, \Code(b)) \leq \delta n$, $\Pr[\Dec^{y}(u) = \Code(b)_u] \geq \frac{1}{2} + \eps$. Here, $\Delta(x,y)$ denotes the Hamming distance between $x$ and $y$, i.e., the number of indices $v \in [n]$ where $x_v \ne y_v$.
\end{inparaenum}
We will use $q$-LCCs to denote $(q,\delta,\epsilon)$-LCCs where $\delta,\epsilon$ are some fixed small constants.

$\Code$ is \emph{linear} if the map $\Code$ is a linear map.
We note that for linear codes, $k = \dim(\cV)$, where $\cV$ is the image of $\F^k$ under the map $\Code$. Without the loss of generality, a linear $\Code$ is systematic, i.e., $\Code(b)_{i} = b_i$ for $i \in [k]$. By a slight abuse of notation, we will also use $\Code$ to denote the set of all codewords, i.e., elements in the range of the map $\Code$. 
\end{definition}
For the Boolean case, i.e., when $\F = \F_2$, it shall be more convenient to think of the map $\Code$ as a function from $\Fits^k$ to $\Fits^n$, defined via the mapping $0 \leftrightarrow 1$ and $1 \leftrightarrow -1$.

We next discuss a combinatorial characterization of locally correctable codes. To begin with, we recall basic notions about hypergraphs.

\begin{definition}
A $q$-uniform hypergraph $\cH$ on vertex set $[n]$ is a collection of subsets $C \subseteq [n]$ of size $q$ called hyperedges. We say that $\cH$ is a \emph{matching} if all the hyperedges in $\cH$ are disjoint. For a subset $Q \subseteq [n]$, we define the degree of $Q$ in $\cH$, denoted $\deg_{\cH}(Q)$, to be $\abs{\{C \in \cH : Q \subseteq C\}}$.
\end{definition}

LCCs admit a standard combinatorial characterization (formalized in the definition below).

\begin{definition}[Linear LCC in normal form]
\label{def:normalLCC}
A linear code $\Code \colon \F^k \to \F^n$ is $(q, \delta)$-normally correctable if for each $u \in [n]$, there is a $q$-uniform hypergraph matching $\cH_u$ with at least $\delta n$ hyperedges such that for every $C = \{v_1, \dots, v_q\} \in \cH_u$, there are coefficients $\alpha_1, \dots, \alpha_q \in \F \setminus \{0\}$ such that, for any $b \in \F^k$, $x = \Code(b)$ satisfies $x_u = \alpha_1 x_{v_1} + \dots + \alpha_q x_{v_q}$.
\end{definition}

\begin{fact}[Reduction to LCC normal form, Theorem 8.1 in~\cite{Dvir16}]\label{fact:normalform}
Let $\Code \colon \F^k \to \F^n$ be a linear code that is $(q, \delta, \eps)$-locally correctable. Then, there is a linear code $\Code' \colon \F^k \to \F^{2n}$ that is $(q, \delta')$-normally correctable, with $\delta' \geq \delta/2q$.
\end{fact}
We note that there is slight difference in \cref{fact:normalform} compared to Theorem 8.1 in~\cite{Dvir16}. In \cref{fact:normalform}, we require that the matchings are $q$-uniform and all coefficients $\alpha$ are nonzero, and we obtain $\delta' \geq \delta/2q$. On the other hand, \cite{Dvir16} allows for hyperedges of size $\leq q$, i.e., some coefficients $\alpha$ may be zero, and obtains $\delta' \geq \delta/q$. We remark that \cref{fact:normalform} immediately follows from~\cite{Dvir16} by ``padding'' the code with $n$ $0$'s. This loses an additional factor of $2$ in $\delta$, but allows one to make all hyperedges have size exactly $q$ by querying the padded $0$ entries.

Finally, we recall the lower bound for linear $2$-LDCs from~\cite{GoldreichKST06}.
\begin{fact}[Lemma 3.3, Claim 4.4 in~\cite{GoldreichKST06}]
\label{fact:2ldclb}
Let $\Code \colon \F^k \to \F^n$ be a linear map, and let $G_1, \dots, G_k$ be matchings on $n$ vertices such that for every $b \in \F^k$ and every $i \in [k]$ and every $(u,v) \in G_i$, it holds that $x_u - x_v = b_i$, where $x = \Code(b)$. Suppose that $\frac{1}{k} \sum_{i = 1}^k \abs{G_i} \geq \delta n$. Then, $\delta k \leq 2\log_2 n$.
\end{fact}
When $\F \ne \F_2$, the above lower bound, as stated, only applies to the setting where the decoder is a linear function with the added restriction that each non-zero coefficient of the linear combination is in $\{-1,1\}$. It is not hard to remove this restriction on coefficients, but, in our setting, we obtain a better dependence on $\abs{\F}$ in \cref{mthm:main} by applying this more specialized lemma.

\subsection{Concentration inequalities}
We will need the following standard concentration inequalities.

\begin{fact}[Chernoff Bound]
Let $x_1, \dots, x_n$ be i.i.d.\ Bernoulli random variables with mean $p$, and let $\mu = p n$. Then, for any $\delta \geq 0$, 
\begin{equation*}
\Pr[\sum_{i = 1}^n x_i \geq (1 + \delta)\mu] \leq \exp(-\delta^2 \mu/(2 + \delta)) \enspace.
\end{equation*}
\end{fact}

\begin{fact}[Scalar Bernstein inequality] \label{fact:bernstein}
Let $x_1, \dots, x_n$ be independent mean $0$ random variables satisfying $\abs{x_i} \leq M$ almost surely for every $i$. Let $\sigma^2 \geq \sum_{i = 1}^n \E[x_i^2]$, for every $i \in [n]$. Then, for all $t \geq 0$, it holds that
\begin{equation*}
\Pr[\sum_{i = 1}^n x_i \geq t] \leq \exp\left(-\frac{\frac{1}{2} t^2}{\sigma^2 + \frac{1}{3} M t}\right) \enspace.
\end{equation*}
\end{fact}

We will use the following non-commutative Khintchine inequality~\cite{LustPiquardG91}.
\begin{fact}[Rectangular Matrix Khintchine inequality, Theorem 4.1.1 of \cite{Tropp15}]
\label{fact:matrixkhintchine}
Let $X_1, \dots, X_k$ be fixed $d_1 \times d_2$ matrices and $b_1, \dots , b_k$ be i.i.d.\ from $\Fits$. Let $\sigma^2 \geq \max(\norm{\sum_{i = 1}^k X_i X_i^{\top}]}_2, \norm{\sum_{i = 1}^k X_i^{\top} X_i]}_2)$. Then
\begin{equation*}
\E\Bigl[\ \Norm{\sum_{i = 1}^k b_i X_i}_2\ \Bigr] \leq \sqrt{2\sigma^2 \log(d_1 + d_2)} \enspace.
\end{equation*}
\end{fact}

\parhead{Tail Bounds for $r$-partite non-negative polynomials.} We give an elementary proof of a concentration inequality for $r$-partite polynomials with non-negative coefficients. Such inequalities are the subject of the celebrated work of Kim and Vu~\cite{KimV00} (with tightenings due to Schudy and Sviridenko~\cite{SchudyS12}).  For $r$-partite polynomials, our inequality below saves a crucial $2^{O(r)}$ factor in the estimate of the typical value when compared to a blackbox application of the above results (without which, we can only obtain a quasi-polynomial lower bound for $3$-LCCs).
\begin{lemma}[Tail Bounds from bounded expected derivatives]
\label{lem:partitepolyconc}
Let $x=\{x^{(i)}_j\}_{\begin{subarray}{c}1 \leq i \leq r \\ 1 \leq j \leq n \end{subarray}}$ be $nr$ independent $p$-biased Bernoulli random variables. Let $P(x) = P(x^{(1)},x^{(2)},\ldots, x^{(r)})$ be a $r$-partite multilinear polynomial of degree $\leq r$ with nonnegative coefficients. That is, each monomial with a non-zero coefficient in $P$ has degree at most $1$ in each $x^{(i)}$ for $1 \leq i \leq r$. For $Z \in ([n] \cup \{\star\})^r$, let $\mu_Z(P)$ be the expected partial derivative of $P$ with respect to the variables $\{x^{(h)}_{Z_h} \mid 1 \leq h \leq r, \text{ } Z_h \neq \star\}$. Suppose that there exists a $\mu,\gamma >0$ such that for every $Z$, $\mu_Z(P) \leq \mu \cdot \gamma^{|Z|}$, where $|Z|$ denotes the number of non $\star$ coordinates in $Z$. 

Then, for every $\beta>0$, 
\begin{equation*}
\Pr_{y}[P(y) \geq (1 + \beta)^r \mu] \leq r (n+1)^{r} \alpha\mcom
\end{equation*} 
where $\alpha = \exp\left(-\frac{\frac{1}{2} \beta^2}{2 \gamma + \frac{1}{3} \gamma \beta} \right)$.
\end{lemma}

\begin{proof}
We will consider the random process that samples $x \in \Bits^{nr}$ by fixing $x^{(i)}$ to a random draw from their distribution one at a time. At each step, we obtain a new polynomial of smaller degree obtained by fixing one additional set of variables to a fixed value in $P$. We understand how the parameters $\mu_Z$ of the polynomials so generated evolve via the Bernstein inequality \cref{fact:bernstein}.  

Formally, fix a $0 \leq t$. Let $(Z_{t+1}, \dots, Z_r) \in [n]^{r-t}$ be a tuple of length $r - t$. We define the quantity $\mu_{y^{(1)}, \dots, y^{(t)}, Z_{t+1}, \dots, Z_r}$ to be the quantity $\mu_{Z_{t+1}, \dots, Z_r}(P_t)$ where $P_t = P(y^{(1)}, \dots, y^{(t)}, x^{(t+1)}, \dots, x^{(r)})$. Here, we use the notation $y^{(i)}$ to denote sampled values for $x^{(i)}$. Note that $P_t$ has $r - t$ ``free'' groups of variables $x^{(t+1)}, \dots, x^{(r)}$.

Let $t \in \{0, \dots, r\}$. We will show that with probability at least $1 - t (n+1)^{r} \alpha$ over the draw of $y^{(1)}, \dots, y^{(t)}$, it holds that for every $Z_{t+1}, \dots, Z_r$ with $\abs{Z_h} \in \Bits$ for all $h = t+1, \dots r$, we have $\mu_{y^{(1)}, \dots, y^{(t)}, Z_{t+1}, \dots, Z_r} \leq (1 + \beta)^t \mu \cdot \gamma^{\sum_{h = t+1}^r \abs{Z_h}}$.

We prove this by induction. The base case of $t = 0$ forms the hypothesis of the lemma. We now prove the inductive step. Let $t \geq 1$, and suppose that with probability at least $1 - (t - 1) (n+1)^{r} \alpha$ over the draw of $y^{(1)}, \dots, y^{(t-1)}$, it holds that for tuple of length $r - t + 1$ $Z_{t}, \dots, Z_r$, we have $\mu_{y^{(1)}, \dots, y^{(t-1)}, Z_{t}, \dots, Z_r} \leq (1 + \beta)^{t-1} \mu \cdot \gamma^{\sum_{h = t}^r \abs{Z_h}}$.

Fix $Z_{t+1}, \dots, Z_r$ with $\abs{Z_h} \in \Bits$ for all $h = t+1, \dots r$. We now show that with probability at least $1 - \alpha$ over the draw of $y^{(t)}$, it holds that $\mu_{y^{(1)}, \dots, y^{(t-1)}, y^{(t)}, Z_{t+1}, \dots, Z_r} \leq (1 + \beta) \mu'$, where $\mu' = (1 + \beta)^{t-1} \mu \cdot \gamma^{\sum_{h = t+1}^r \abs{Z_h}}$. The lemma then follows by union bound over the (crudely) at most $(n+1)^r$ choices for $Z_{t+1}, \dots, Z_r$.

For an assignment $y^{(t)}$, we have that $\mu_{y^{(1)}, \dots, y^{(t-1)}, y^{(t)}, Z_{t+1}, \dots, Z_r} = f(y^{(t)})$, where $f(x^{(t)}) \coloneqq \sum_{u = 1}^n c_u x^{(t)}$ is a linear polynomial with nonnegative coefficients $c_u \coloneqq \mu_{S_1, \dots, S_{t-1}, \{u\}, Z_{t+1}, \dots, Z_r}$. We note that the mean is $\E_{y^{(t)}}[f(x^{(t)})] = \mu_{S_1, \dots, S_{t-1}, \emptyset, Z_{t+1}, \dots, Z_r} \leq \mu'$, by the induction hypothesis. We also have that $c_u = \mu_{S_1, \dots, S_{t-1}, \{u\}, Z_{t+1}, \dots, Z_r} \leq (1 + \beta)^{t-1} \mu \cdot \gamma^{1 + \sum_{h = t+1}^r \abs{Z_h}} = \gamma \mu'$, again by the induction hypothesis.

We now bound the polynomial by using the Bernstein Inequality. Let $y'^{(t)}$ be the centered version of $y^{(t)}$, i.e., $y'^{(t)}_u = y^{(t)}_u - p$, so that $y'^{(t)}_u = 1 - p$ with probability $p$, and $-p$ with probability $1 - p$. Then, $\E[(y'^{(t)}_u c_u)^2] = c_u^2 ((1 - p)^2 p + p^2 (1 - p)) \leq 2 p c_u^2$. Further, we observe that $\abs{y'^{(t)}_u c_u} \leq c_u \leq \gamma \mu' =: M$ always holds. We also note that
\begin{flalign*}
&\sigma^2 \coloneqq \sum_{u = 1}^n \E[(y'^{(t)}_u c_u)^2] \leq \sum_{u = 1}^n 2p c_u^2 \leq 2 (\max_{u} c_u) (p \sum_{u = 1}^n c_u) \leq 2 M \cdot \mu_{S_1, \dots, S_{t-1}, \emptyset, Z_{t+1}, \dots, Z_r} \\
&\leq2 (\gamma \mu') \mu' = 2 \gamma {\mu'}^2 \enspace.
\end{flalign*}
Thus, by the Bernstein Inequality, we have
\begin{equation*}
\Pr[f(y'^{(t)}) \geq \lambda] \leq \exp\left(-\frac{\frac{1}{2} \lambda^2}{\sigma^2 + \frac{1}{3} M \lambda} \right) \enspace,
\end{equation*}
and therefore 
\begin{equation*}
\Pr[f(y'^{(t)}) \geq \beta \mu' ] \leq \exp\left(-\frac{\frac{1}{2} \beta^2}{2 \gamma + \frac{1}{3} \gamma \beta} \right) = \alpha \enspace.
\end{equation*}
Note that since $f$ is linear, $f(y'^{(t)}) = f(y^{(t)}) - \mu_{S_1, \dots, S_{t-1}, \emptyset, Z_{t+1}, \dots, Z_r}$, and so it follows that $\Pr[f(y^{(t)}) \geq (1 + \beta) \mu'] \leq \alpha$, which finishes the proof.
\end{proof}

\section{Proof overview}
\label{sec:techniques}
In this section, we will focus on the case of $\F=\F_2$ to give a high-level overview of the main ideas in the proof of \cref{mthm:main}. Without loss of generality, we can assume that $\Code$ is a systematic linear map $\Code \colon \Fits^k \to \Fits^n$, so that the first $k$ bits in any codeword are the message bits themselves, i.e., for any $b \in \Fits^k$, $x = \Code(b)$ satisfies $x_i = b_i$ for all $i \in [k]$. In this section and the next, we will use the notation $\gtrapprox$ and $\lessapprox$ to suppress a multiplicative $\polylog(n)$ factor.

\parhead{The Kikuchi matrix method.} Our proof uses the Kikuchi matrix method developed in prior works~\cite{GuruswamiKM22,HsiehKM23,AlrabiahGKM23} for finding extremal trade-offs for combinatorial structures in hypergraphs. This method works in two steps: (1) formulate a hypergraph possessing some relevant structure as a family of satisfiable XOR formulas, and, (2) construct a spectral refutation (i.e.,  a certificate of unsatisfiability) of a randomly chosen member of this family. The spectral refutations in the second step rely on appropriate \emph{Kikuchi} matrices --- a term that we loosely use to describe induced subgraphs of an appropriately chosen Cayley graph associated with the hypergraph. The success of the spectral refutation naturally relies on the structure of the XOR instances. The power of the method comes from the ease (at least in hindsight, given~\cite{GuruswamiKM22,HsiehKM23,AlrabiahGKM23}) in identifying the relevant combinatorial structure that is sufficient for the success of the spectral refutations. This method has been used to prove Feige's conjecture~\cite{Feige08} on the hypergraph Moore bound (extremal girth vs.\ density trade-off)~\cite{GuruswamiKM22,HsiehKM23}, improved lower bounds for $3$-LDCs~\cite{AlrabiahGKM23}, and generalizations of Szemeredi's theorems for arithmetic progressions with restricted common differences~\cite{BrietC23} (which closely follows the argument in~\cite{AlrabiahGKM23}). 

Our proof can be seen as an upgrade on a recent work~\cite{AlrabiahGKM23} that showed a lower bound of $n \geq \tilde{\Omega}(k^3)$ on the block length $n$ of a code of dimension $k$ and constant distance.\footnote{Their result extends to non-linear codes but we omit this distinction here.} The key conceptual idea that helps us move beyond the cubic to an exponential lower bound (a bound that provably cannot hold for $3$-LDCs~\cite{Efremenko09,Yekhanin08}) is a new family of XOR instances that crucially exploits the additional structure in LCCs. Our new family of XOR instances is produced by performing a certain structured variant of \emph{low-width} resolution (well-studied in proof complexity~\cite{Grigoriev01,Schoenebeck08}) on the ``basic'' family. We call this process \emph{long chain derivations}. 

In the following, we will first recall the conceptual crux of the lower bound for $q$-LDCs in~\cite{AlrabiahGKM23} and then use it to motivate our approach for $3$-LCCs. 

\subsection{The naive XOR instance and LDC lower bounds}
\label{sec:evenqLDC}
 Let's first consider the case of $3$-LDCs and start by recalling the combinatorial characterization (formalized as the \emph{normal form} in \cref{def:normalLCC}). A code $\Code \colon \Fits^k \to \Fits^n$ is a $(q,\delta)$-LDC if for every $1 \leq i \leq k$, there exists a $q$-uniform hypergraph matching $H_i$ over $[n]$ of size $\delta n$ such that for every $b \in \Fits^k$ and codeword $x = \Code(b)$, for every $i \in [k]$ and every $C \in H_i$, it holds that $x_C = b_i$. The combinatorial characterization above can be easily seen to be equivalent to the satisfiability of a family of $q$-XOR instances. 

\begin{observation}[LDCs and a Family of XOR Instances] \label{obs:overview-LDC-to-XOR}
Let $H_1, H_2,\ldots, H_k$ be $q$-uniform hypergraph matchings on $[n]$ of size $\delta n$. For every $b \in \Fits^k$, define the following $q$-XOR instance $\Phi_b$ in $n$ variables $x_1, x_2,\ldots, x_n$. 
\begin{equation} \label{eq:overview-3-ldc-xor}
\forall i \in [k], \text{ } \forall C \in H_i \text{ }, x_C = b_i \mper
\end{equation}
Then, there exists a (normal form) linear LDC $\Code:\Fits^k \rightarrow \Fits^n$ described by the collection of $q$-uniform matchings $H_1, H_2,\ldots, H_k$ on $[n]$ if and only if $\Phi_b$ is satisfiable for every $b \in \Fits^k$.
\end{observation}
If $\Code$ is a $(q,\delta)$-LDC described by matchings $H_1, H_2,\ldots, H_k$, then $x = \Code(b)$ satisfies all the constraints in $\Phi_b$. Conversely, if $\Phi_b$ is satisfiable for every $b$, then one can easily construct a linear map $\Code$ (easily seen to be a linear $(q,\delta)$-LDC) where $\Code(b)$ is some satisfying assignment to $\Phi_b$.

The main idea of~\cite{AlrabiahGKM23} is to show that for any collection of $\delta n$-size $q$-matchings $H_1, H_2,\ldots, H_k$, if $k$ is large enough as a function of $n$, then for a randomly chosen $b$, $\Phi_b$ is unsatisfiable with high probability. This implies an upper bound on $k$. Now, when $b$ is random, $\Phi_b$ is XOR formula generated via $k \ll n$ bits, i.e., much smaller than the number of variables. Thus, a naive union bound argument cannot establish unsatisfiability of $\Phi_b$. The work of~\cite{AlrabiahGKM23} establishes unsatisfiability of $\Phi_b$ for a random $b$ via a \emph{spectral refutation} using \emph{Kikuchi} matrices. 

\parhead{Spectral refutations for $\Phi_b$.} Let us now recall how the spectral refutation in~\cite{AlrabiahGKM23} works. Their main result is for the case of $q=3$ (where they obtained improvements on prior works). However, for our purpose of illustrating the conceptual idea,  we will focus on the simpler setting of even $q$ and sketch their proof that $k \leq \tilde{O}(n^{1-2/q})$ for $q$-LDCs. 

First, we observe that for the XOR instance $\Phi_b$, there is an associated ``instance polynomial'' $\Phi_b(x) \coloneqq \sum_{i = 1}^k \sum_{C \in H_i} b_i x_C$. We note that $\Phi_b(x)$ is the number of constraints satisfied by $x$ minus the number of constraints violated, and thus $\Phi_b$ is unsatisfiable if and only if $\val(\Phi_b) \coloneqq \max_{x \in \Fits^n} \Phi_b(x)$ is less than $\sum_{i = 1}^k \abs{H_i} = k \cdot \delta n$. Thus, to show that $\Phi_b$ is unsatisfiable, we will bound $\val(\Phi_b)$.

To do this, we define a matrix whose quadratic form is equal to $\Phi_b(x)$.

\begin{definition}[Kikuchi matrix and graphs]
\label{def:basickikuchi}
Let $C \in {[n] \choose q}$, let $\ell$ be a parameter, and let $N \coloneqq {n \choose \ell}$. Let $A_C \in \Bits^{N \times N}$ be the matrix indexed by sets $S \in {[n] \choose \ell}$ where $A_{C}(S, T) = 1$ if $S \oplus T = C$, and $0$ otherwise. Let $A_i \coloneqq \sum_{C \in H_i} A_C$, and let $A \coloneqq \sum_{i = 1}^k b_i A_i$. We naturally interpret (and by abuse of notation, also call) $A_C$, $A_i$ and $A$ as adjacency matrices of ``Kikuchi graphs'' on the vertex set ${{[n]} \choose \ell}$.  
\end{definition}

Observe that $A_C$ is a matching on vertex set $[n] \choose \ell$ of size $D = {n-q \choose q/2} {q \choose q/2}$. For any $x \in \Fits^n$, let $x^{\circ \ell}$ denote the $\ell$-wise monomial vector indexed by $S \in {{[n]} \choose \ell}$ with corresponding entry equal to $x_S$. Then, ${x^{\circ \ell}}^{\top} A_C x^{\circ \ell} = D x_C$. Consequently, ${x^{\circ \ell}}^{\top} A x^{\circ \ell} = D \Phi_b(x)$. Thus, if $x\in \Fits^n$ satisfies $\Phi_b$, then we have the following inequality that upper bounds $k$ in terms of $\Norm{A}_2$: 
\begin{equation} 
\label{eq:simple-LDC-spectral}
k \delta n = \Phi_b(x) \leq \frac{1}{D} \Norm{x^{\circ \ell}}_2^2 \Norm{A}_2 = \frac{{n \choose \ell}}{D} \Norm{A}_2 \leq O((n/\ell)^{q/2}) \Norm{A}_2\mper
\end{equation}  
We now choose $b\in \Fits^k$ uniformly at random and consider $A= \sum_i b_i A_i$, which is a matrix Rademacher series of the $A_i$'s. By the matrix Khintchine inequality, $\Norm{A}_2 \leq O(\sqrt{\log N}) \Norm{\sum_i A_i^2}^{1/2}_2$ with high probability. 

\parhead{A combinatorial proxy for $\Norm{A}_2$.} Let $\Delta_i$ be the maximum degree of any node in the Kikuchi graph $A_i$, and let $\Delta = \max_{1 \leq i \leq k} \Delta_i$. Then, we can naively bound $\Norm{\sum_i A_i^2}_2 \leq \sum_i \Norm{A_i}_2^2 \leq k \Delta^2$. Thus, the maximum degree of the $A_i$'s naturally controls the spectral norm of $A$ as $\Norm{A}_2 \leq \Delta \cdot O(\sqrt{k\ell \log n})$.  

Let us now investigate bounds on $\Delta$. Since for each $C \in H_i$, $A_C$ contributes $D$ edges to $A_i$, the average degree of $A_i$ is clearly $\delta n D/ N \sim n (\ell/n)^{q/2}$. Thus, $\Delta \geq O(1) \max\{1, n (\ell/n)^{q/2}\}$. If $\Delta$ happens to be equal to this minimum possible value, then plugging it in \cref{eq:simple-LDC-spectral} yields:
\begin{equation*}
k \delta n \leq O(1) \left(\frac{n}{\ell}\right)^{q/2} \sqrt{k \ell \log n} \cdot \max\{1, n (\ell/n)^{q/2}\} \mcom
\end{equation*}
which implies that $k \leq O(\ell \log n) \cdot \max\{ n^{q-2}/\ell^{q}, 1\}$. This is minimized at $\ell = n^{1 - 2/q}$ to give the lower bound of $k \leq \tilde{O}(n^{1-2/q})$, i.e., $n \geq \tilde{\Omega}(k^{q/(q-2)})$.

\parhead{Handling irregularities: row pruning via polynomial concentration.} We will now (for the first time in the argument) use that the $H_i$'s are matchings to argue that while the $A_i$'s are certainly not approximately regular (i.e., max degree $\Delta_i$ at most a $\polylog(n)$ factor larger than the average-degree), there is only a small fraction of nodes in any $A_i$ that have a large degree. Of course, a small fraction of rows can still cause $\Norm{A}_2$ to be too large. In order to circumvent this issue, we observe that the argument in \cref{eq:simple-LDC-spectral} works even if we were to replace $N \Norm{A}_2$ (maximum over arbitrary quadratic forms) by $\Norm{A}_{\infty \to 1}$ (maximum over quadratic forms on $\pm 1$-coordinate vectors). The latter quantity is insensitive to dropping a small fraction of rows since $\pm 1$-coordinate vectors when restricted to a small number of rows must have correspondingly small $\ell_2$-norm.  

To prove that only a small fraction of nodes can have a large degree in any $A_i$, we view the degree of any node $S$ as a polynomial in the corresponding indicator variables $z \in \Bits^n$ with $\sum_i z_i = \ell$ and use tail inequalities for low-degree polynomials (that generalize concentration of Lipschitz functions) of Kim and Vu and extensions~\cite{KimV00,SchudyS12} to bound the chance that it takes a value $\polylog(n)$ times the average. This relies on establishing strong bounds on the expected partial derivatives of the degree polynomial by using that the $H_i$'s are matchings. 

\parhead{The key heuristic: high density for Kikuchi graphs at low levels.} Let's summarize the crucial steps of the above argument as follows: (1) $q$-LDCs naturally yields XOR instances of arity $q$, (2) to obtain our lower bound, we need that the Kikuchi matrices $A_i$ corresponding to a matching $H_i$ are approximately regular (after dropping a negligible fraction of rows), and (3) the argument can only yield a bound of the form $k \lessapprox \ell$ where $\ell$ is the smallest level of the Kikuchi graphs $A_i$ with an average degree $\gg 1$. More precisely, if there are $m_i$ constraints of arity $q$ in $H_i$, then the threshold $\ell$ is the smallest integer satisfying $m_i (\ell/n)^{q/2} \gg 1$ for all $i \in [k]$. Note that this threshold $\ell$ increases as $q$ increases. 

We assert that even though the argument in~\cite{AlrabiahGKM23} for the case when $q=3$ requires more work (in both the design of the Kikuchi matrix itself and its analysis), the heuristic above continues to hold. Let us also note that ensuring approximate regularity is usually the trickiest aspect of the proof. In particular, while the heuristic above  makes sense for all odd $q$ (and not just $q = 3$), and the work of~\cite{AlrabiahGKM23} fails to obtain an improved lower bound for odd $q>3$ because they were unable to find an appropriate ``decomposition'' that ensures approximate regularity of the resulting Kikuchi matrices.

Thus, in order to obtain an exponential lower bound, as in \cref{mthm:main}, via the schema above, we must construct Kikuchi graphs that have constant density (i.e., average degree) at much a lower level $\ell$. Specifically, we will need to be able to take $\ell = \polylog(n)$.\footnote{We note that while our lower bounds appear to get weaker as $\ell$ grows, generic convergence results about the Kikuchi matrices imply that taking $\ell \sim n$ and bounding $\Phi_b$ in terms of $\Norm{A}_2$ yields the \emph{optimal} bound on $k$, whatever it may be! The reason the current argument (which is likely suboptimal) does not extend beyond $\ell = n^{1 - 2/q}$ is the potentially superfluous $\sqrt{\log N}$ multiplicative loss in the matrix Khintchine inequality. Investigating when this $\sqrt{\log N}$ factor (which is tight in the worst-case) can be removed is the topic of an ongoing research effort in random matrix theory~\cite{BandeiraBH23} and is naturally related to other problems such as resolving the matrix Spencer conjecture~\cite{Zouzias12,Meka14}.}

\subsection{Long chain derivations: improved spectral refutations by increasing density}
\label{sec:chainheuristic}
Given the key heuristic above, we now show how to build XOR instances from $3$-LCCs that yield constant density Kikuchi matrices at level $\ell = \polylog(n)$. Our instances will balance two opposing concerns. On the one hand, they will be of large arity (in fact, $O(\log n)$ arity) which, given the discussion above, hurts the density at lower levels. Nonetheless, we will show that the number of higher arity constraints that we produce grows fast enough to compensate for this and gives us an overall increase in density at lower $\ell$. We note (with the hope of pointing the reader to the trickiest part of the proof that motivates all our setup) that the analysis of ``row pruning'' i.e., arguing approximate regularity after removing a negligible fraction of rows, will get significantly more involved and motivates all our design choices. This includes the specific type of Kikuchi matrices that we will choose and a new decomposition for the constraints that, while a bit unnatural at the outset, helps guarantee approximate regularity. Let us see these ideas in more detail next.  

Like $3$-LDCs, $3$-LCCs can, without loss of generality, be assumed to be $(3,\delta)$-normal. Thus, for any $3$-LCC $\Code \colon \Fits^k \to \Fits^n$, there are $3$-uniform hypergraph matchings $H_1, \dots, H_n$ on $[n]$, each of size $\delta n$, such that for every $b \in \Fits^k$, $u \in [n]$, and $C \in H_u$, the encoding $x = \Code(b)$ satisfies $x_C = x_u$. Note that the key difference between LCCs and LDCs is that here we have a ``local correcting'' hypergraph $H_u$ for each $u \in [n]$, instead of only a hypergraph for each $u \in [k]$ in the case of LDCs.

\parhead{The naive XOR instances.} Similar to \cref{obs:overview-LDC-to-XOR}, the combinatorial characterization yields that the XOR instance with constraints $x_C = x_u$ for every $C \in H_u$ and $u \in [n]$ (where on the right hand side, we set $x_u=b_u$ whenever $u \in [k]$) is satisfiable for every $b \in \Fits^k$. If we focus only on the constraints corresponding to $H_u$ for $u \in [k]$ (i.e., the ``systematic'' bits in the codeword), then we recover the same XOR instance as in the case of $3$-LDCs and our method from above yields $k \leq \tilde{O}(n^{1/3})$~\cite{AlrabiahGKM23}. To improve on this significantly lossy formulation, we must make use of the additional constraints $H_u$ for $u \not \in [k]$. More specifically, if we were to only use the hypergraphs $H_u$ for $u \in [k]$, then any lower bound we could prove would hold for LDCs as well, and in particular one could not hope to prove \cref{mthm:main}, which is false for LDCs.

\parhead{Long chain derivations.} We now show how to use the additional constraints in order to build a higher arity XOR instance that is (1) approximately regular (after an appropriate decomposition), and (2) results in high-density Kikuchi graphs at $\polylog(n)$ levels. We will construct higher arity XOR instances that use the additional constraints above using a structured variant of low-width XOR resolution~\cite{Grigoriev01,Schoenebeck08} that we call \emph{long chain derivations}.

Let us start by forming extra constraints via $2$-chains. Observe that for any $u \in [n]$ and $C \in H_u$, we have that for any $b \in \Fits^k$, $x = \Code(b) \in \Fits^n$ satisfies the equation $x_u x_C = 1$. Now, let us choose $w \in C$ and $C' \in H_w$. We also have that $x_{w} x_{C'} = 1$. As $x_C = x_{C\setminus \{w\}} x_w$, it follows that the ``derivation'' $x_u x_{C\setminus \{w\}} x_{C'} = 1$ also holds, since $x_w^2 = 1$. We shall call such a constraint a ``$2$-chain'' --- it connects two constraints intersecting in one variable. We can think of such a $2$-chain as a tuple $(u, C, w, C', w')$, where $C \cup \{w\} \in H_u$ and $C' \cup \{w'\} \in H_w$, and this yields the constraint $x_{C} x_{C'} x_{w'} = x_u$ (see \cref{fig:2chain}). 

\begin{figure}[t]
    \centering
    \includegraphics[width=0.6\textwidth]{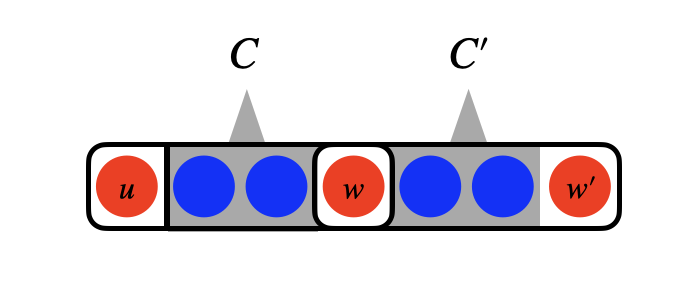}
    \caption{A $2$-chain with head $u$. Note that $C \cup \{w\} \in H_u$ and $C' \cup \{w'\} \in H_{w}$, and that $x = \Code(b)$ satisfies $x_{C} x_w = x_u$ and $x_{C'} x_{w'} = x_w$, and therefore $x_{C} x_{C'} x_{w'} = x_u$.}
\hrulefill
    \label{fig:2chain}
\end{figure}
Consider now the $2$-chains $\cup_{i \in [k]} \cH_i^{(2)}$, i.e., $2$-chains of the form $(i, C, w, C', w')$ where $i \in [k]$. Then, the constraints have the form $x_{C} x_{C'} x_{w'} = b_i$, so they decode the $i$-th independent bit $b_i$. We have thus formed a new set of constraints with ``right hand side'' $b_i$. 

\parhead{A heuristic calculation.} Let us now do a heuristic calculation (that ignores the key issue of approximate regularity) to see if we improve the density at lower Kikuchi levels by taking the XOR instances corresponding to $2$-chains. For any fixed ``head'' $i \in [k]$, there are $(3 \delta n)^2$ $2$-chains. This is because we have $\delta n$ choices for $C \cup \{w\} \in H_i$, followed by $3$ ways to choose $w$ from $C \cup \{w\}$, and then similarly $3 \delta n$ choices in total for $(C', w')$. Let $\cH_i^{(2)}$ denote the set of $2$-chains with head $i$. We have thus produced $\sim n^2$ constraints and each constraint has arity $5$,\footnote{Some constraints may have additional variable cancellations and thus have arity $< 5$. However, as the density gets worse as the arity increases, this is only ``better'' for us.} as $\abs{C} = \abs{C'} = 2$. 

The Kikuchi matrix in \cref{def:basickikuchi} only makes sense for even $q$, but let us still do a ``pretend'' calculation of the relative density for the arity $5$ constraints we have produced. This can be made precise with a slightly more sophisticated Kikuchi matrix, so this is still a meaningful heuristic. 

The density (i.e., average degree) expression for a Kikuchi matrix $A_i$ is now $n^2 (\ell/n)^{q/2} \sim n^2 (\ell/n)^{5/2} \sim \ell^{2.5}/n^{0.5}$. This density is $\gg 1$ whenever $\ell \gg n^{1/5}$, so one might expect to obtain a bound of $k \lessapprox n^{1/5}$ (beating the $n^{1/3}$ bound for the naive XOR instance~\cite{AlrabiahGKM23}) when working with $2$-chains --- a construction that crucially relies on additional structure in $3$-LCC! While there are lot of details that we have simply ignored in doing this calculation, it does suggest that we are able to achieve a constant-density Kikuchi matrix $A_i$ at a lower level $\ell$. A similar calculation (that we will omit here) for chains of larger length, say $r$, shows that the smallest level $\ell$ at which we can obtain constant density Kikuchi matrices is $\ell \sim n^{1/2r}$, and this suggests that we might be able to obtain constant density at level $\ell = \polylog(n)$ if we work with $r \sim \log n$ length chains. 

In \cref{sec:warmup}, as a warmup to our somewhat technical proof of the main theorem, we present a complete analysis of the $2$-chains (with extended commentary) to obtain a $k \leq \tilde{O}(n^{1/4})$ bound (giving a polynomial improvement on the $\sim n^{1/3}$ lower bound on $3$-LDCs already!) in order to illustrate (a simplified version of) the set of new  tools that go into the analysis. 

\subsection{From the heuristic to a proof} 
In the remaining part of this overview, we briefly discuss the technical tools we develop to turn the above heuristic calculation into a full proof. We note that the actual parameters become rather delicate. For readers familiar with the literature on random CSP refutation (our setting resembles semirandom XOR refutation with complicated correlations in the right hand sides), this is similar to the analysis getting rather delicate when dealing with XOR instances with super-constant arity.

\parhead{Setting up the Kikuchi matrix.} The instances produced by forming $r$-chains yield XOR instances of (odd) arity $2r+1$. We build a different Kikuchi matrix by first applying the ``Cauchy--Schwarz'' trick --- a standard idea in CSP refutation also utilized in~\cite{AlrabiahGKM23}. In our case, the XOR instance produced after this trick corresponds to constraints formed by joining two $r$-chains at their ``tails'' whenever the tails match. We choose a variant of the Kikuchi matrix for the ``Cauchy--Schwarzed instance'' except for the key difference that it is indexed by $2r$-tuples of sets of size $\ell$ (instead of a single set of size $\ell$) in the sketch above. This choice is crucial in the analysis of row pruning, in particular, as we discuss below, in obtaining bounds that significantly beat those obtained by a blackbox application of low-degree polynomial concentration~\cite{KimV00}, see below. 

\parhead{Regularity decomposition.} If $H_1, H_2, \ldots H_n$ are such that no pair of variables appears in more than one hyperedge (``no heavy pairs'') across all the $H_i$'s, then it turns out that the resulting Kikuchi matrices satisfy approximate regularity after pruning a negligible fraction of rows. This no-heavy-pair property holds, e.g., if $H_i$'s are uniformly random and independent hypergraph matchings of size $\delta n$. 

However, when the $H_i$'s are arbitrary, and in particular when there are ``heavy pairs'' (i.e. pairs of variables that appear in $\gg \log n$ hyperedges across the $H_i$'s), the resulting Kikuchi matrices are \emph{far} from being approximately regular. Our key technical idea is a new \emph{decomposition} procedure that operates directly on the chains. Such a decomposition procedure  partitions the chains into $\sim r$ different groups such that each group admits a (different, appropriately defined) Kikuchi matrix that satisfies approximate regularity. Regularity decompositions were already used in early applications of the Kikuchi matrix method for proving hypergraph Moore bound and smoothed CSP refutation~\cite{GuruswamiKM22,HsiehKM23}. However, our notion of regularity is (necessarily) significantly weaker (we call it ``contiguously regular'' partitioning) that, unlike~\cite{GuruswamiKM22}, does not ``by design'' ensure approximate regularity of the Kikuchi matrices after removing only a negligible fraction of rows. Instead, our argument for approximate regularity relies on combining the guarantees of the decomposition with (1) an appropriate choice of Kikuchi matrix for each piece in the partition, and (2) the structure in the chains arising by virtue of $H_i$'s being matchings.

\parhead{Polynomial concentration: bounding expected derivatives} Our main technical step (the subject of~\cref{sec:rowpruning}) is proving that our weak notion of regularity combined with the fact that $H_i$'s are matchings is enough to control expected partial derivatives of the ``degree-polynomial'' that computes the degrees of nodes in the Kikuchi graph. 

We note that off-the-shelf low-degree polynomial concentration inequalities (e.g., the Kim--Vu inequality~\cite{KimV00} or the related inequality of Schudy and Sviridenko~\cite{SchudyS12}) lose an exponential factor in the degree of the polynomial in the tail bound. This exponential factor is too costly for us as the arity of our constraints, and thus the degree of the polynomial, is $O(\log n)$ that eventually restricts us to only a quasi-polynomial instead of an exponential lower bound on $3$-LCCs. Instead, we induce a special ``partite'' structure (i.e., there exists a partition of the variables so that the degree of the polynomial is $1$ when restricted to any single piece in the partition) in the polynomial by setting up our Kikuchi matrix to be indexed by tuples of sets (instead of a single set). For such partite polynomials, we prove an analog\footnote{We did not find a reference to a known result so we include a proof in \cref{lem:partitepolyconc}.} of the Kim--Vu inequality for partite that gives sharper bounds when its expected partial derivatives decay appropriately. 

We note that the analysis of the expected partial derivatives of the ``degree polynomial'' (which we use to prove approximate regularity) and the interplay of these bounds with our decomposition of chains is the key technical part (and the focus of \cref{sec:rowpruning}) of our proof. In order to illustrate this technical part in a ``base'' case that still captures some of the complications, we present the case of $2$-chains as a warmup in the next section.

\section{Warmup: An $n \geq \tilde{\Omega}(k^4)$ Lower Bound via $2$-Chains} \label{sec:warmup}

In this section, we give a detailed sketch of the proof of the following theorem, which is a weaker version of our main result. Notice that this theorem already improves the best known $3$-LCC lower bound~\cite{AlrabiahGKM23} by a polynomial factor in $k$.

\begin{theorem}[Weak version of \cref{mthm:main}] \label{thm:overview-2-chains}
Let $\Code:\Fits^k \rightarrow \Fits^n$ be a $(3,\delta)$-LCC in normal form with $\delta = O(1)$. Then, $n \geq \tilde{\Omega}(k^4)$. 
\end{theorem}
The theorem above obtains a lower bound of $n \gtrapprox k^4$ --- worse than the bound of $n \gtrapprox k^5$ predicted by the heuristic but still beating $n \gtrapprox k^3$ from \cite{AlrabiahGKM23}; we discuss the reason that we do not match the heuristic in~\cref{rem:losing-on-the-heuristic}.

\begin{proof}
As before, we have $3$-uniform hypergraph matchings $H_1, \dots, H_n$, where for any $u \in [n]$ and $C \in H_u$, we have that for any $b \in \Fits^k$, $x = \Code(b)$ satisfies $x_C = x_u$. Following \cref{sec:chainheuristic}, we shall let $\cH_i^{(2)}$ denote the set of $2$-chains with head $i$. We define the $5$-XOR instance $\Phi_b(x)$ as
\begin{flalign*}
\Phi_b(x) \coloneqq \sum_{i = 1}^k b_i \sum_{\vec{C} = (i, C_0, w_0, C_1, w_1) \in \cH_i^{(2)}} x_{C_0} x_{C_1} x_{w_1} \enspace.
\end{flalign*}
We note that $\val(\Phi_b) = k (3 \delta n)^2$ for any $b \in \Fits^k$, as the instance is satisfiable and has $k (3 \delta n)^2$ constraints in total.
Following the strategy in \cref{sec:evenqLDC}, we shall use spectral refutation via Kikuchi matrices to bound $\val(\Phi_b)$ with high probability for a random $b \in \Fits^k$.

\subsection{Step 1: the Cauchy--Schwarz trick}
\label{sec:warmupstep1}
As we have observed, the basic Kikuchi matrices in \cref{def:basickikuchi} are only defined for constraints of even arity, but the constraints in $\cH_i^{(2)}$ have arity $5$, i.e., odd arity.
The standard way to handle odd arity XOR instances is to use the ``Cauchy--Schwarz trick'', which produces even arity instances as follows. 
Let $\vec{C} \in \cH^{(2)}_i$ and $\vec{C'} \in \cH^{(2)}_j$ for $i \ne j \in [k]$ be two constraints in our initial $5$-XOR instance, where $\vec{C} = (i, C_0, w_0, C_1, w_1)$ and $\vec{C'} = (j, C'_0, w'_0, C'_1, w'_1)$ where $w_1 = w'_1$, i.e., the last element of both chains is the same. From this pair, we can ``cancel'' $w_1 = w'_1$, producing the derived constraint $x_{C_0} x_{C_1} x_{C'_0} x_{C'_1} = b_i b_j$, which has arity $8$. We do this for all pairs of chains with the same ``tail'' vertex $w$. We note that this process produces at least $(k (3 \delta n)^2)^2/n \sim k^2 n^3$ constraints.

We now define the following ``Cauchy--Schwarzed instance'' polynomial:
\begin{flalign*}
f_b(x) = \sum_{i \ne j \in [k]} b_i b_j \sum_{w \in [n]} \sum_{\vec{C} \in \cH^{(2)}_i, \vec{C'} \in \cH^{(2)}_j : w_1 = w'_1 = w} x_{C_0} x_{C_1} x_{C'_0} x_{C'_1} \enspace.
\end{flalign*}
The phrase ``Cauchy--Schwarz trick'' refers to the fact that one can show $k^2 n^4 \sim \Phi_b(x)^2 \leq n \cdot f_b(x) + o(k^2 n^4)$ via a simple application of the Cauchy--Schwarz inequality and a bound on the ``diagonal terms'' where $i = j$. This reduces the task to bounding the cross-term polynomial $f_b$. 

We now observe that the ``right-hand sides'' of the constraints in $f_b$ are no longer independent, as they are of the form $b_i b_j$ for $i \ne j \in [k]$, and this will cause an issue ``downstream'' when we apply matrix concentration bounds, as the matrices will not be independent. To recover independence, we consider the polynomial $f_{M, b}(x)$ defined for a (directed) matching $M$ on $[k]$:
\begin{flalign*}
f_{M,b}(x) = \sum_{(i,j) \in M} b_i b_j \sum_{w \in [n]} \sum_{\vec{C} \in \cH^{(2)}_i, \vec{C'} \in \cH^{(2)}_j : w_1 = w'_1 = w} x_{C_0} x_{C_1} x_{C'_0} x_{C'_1} \enspace.
\end{flalign*}

Because we now sum over a matching, we have that $b_i b_j$ and $b_{i'} b_{j'}$ are independent for different directed edges $(i,j)$ and $(i', j')$ in $M$. And, we can easily relate $f_b$ and $f_{M,b}$, as $f_b(x) = 2(k-1) \E_{M} f_{M,b}(x)$ when $k$ is even, and $f_b(x) = 2k \E_{M} f_{M,b}(x)$ when $k$ is odd, where the expectation is over a maximum matching $M$. This is because the chance that $M$ contains a directed edge $(i,j)$ is $\frac{1}{2(k-1)}$ if $k$ is even and $\frac{1}{2k}$ if $k$ is odd. In particular, there exists a maximum matching $M$ such that $\val(f_{M,b}) \geq \frac{2}{k} \val(f_b) \sim k n^3$.
\begin{remark} \label{rem:losing-on-the-heuristic}
Restricting to a matching $M$ loses a factor of $k$ in the number of constraints. This leads to a factor $k$ ``loss'' in the density of the corresponding Kikuchi matrix and is the main reason why we obtain weaker bound of $n \geq \tilde{O}(k^4)$ instead of $k^5$ suggested by our heuristic calculation in \cref{sec:chainheuristic}. A better bound could be obtained by instead following the setup in~\cite{AlrabiahGKM23}, where they split $[k]$ randomly into a left and right set $L$ and $R$ and only consider constraints where $i \in L$ and $j \in R$ (thereby losing only $\sim 1/2$ of the constraints instead of a factor $k$). This careful setup is necessary in~\cite{AlrabiahGKM23} for their goal of obtaining a cubic (as opposed to the known quadratic) bound, but this makes the ``row pruning'' step (i.e., arguing approximate regularity of Kikuchi graphs after removing a negligible fraction of constraints) significantly more challenging. In our case, the effect of this loss on the final lower bound diminishes as the length of the chain $r$ grows and when $r \sim \log n$, disappears asymptotically, and so we pick a matching $M$ to make the row pruning easier.
\end{remark}

\subsection{Step 2: spectral refutation via Kikuchi matrices}
\label{sec:warmupstep2}
Let us now bound $\val(f_{M,b})$ (with high probability over $b \in \Fits^k$) for any maximum matching $M$. We introduce our Kikuchi matrices:
\begin{definition}
\label{def:2chainkikuchi}
For $i \ne j \in [k]$ and $\vec{C} = (i, C_0, w_0, C_1, w_1)$ and $\vec{C'} = (j, C'_0, w'_0, C'_1, w'_1)$ with $w_1 = w'_1$, we define the matrix $A_{i,j}^{(\vec{C}, \vec{C'})}$ as follows. The rows/columns of the matrix $A_{i,j}^{(\vec{C}, \vec{C'})}$ are indexed by a $4$-tuple of sets $(S_0, S_1, S'_0, S'_1)$, each in ${[n] \choose \ell}$, and the $((S_0, S_1, S'_0, S'_1), (T_0, T_1, T'_0, T'_1))$-th entry is $1$ if $S_0 \oplus T_0 = C_0$, $S_1 \oplus T_1 = C_1$, $S'_0 \oplus T'_0 = C'_0$, $S'_1 \oplus T'_1 = C'_1$, and is $0$ otherwise.

We let $A_{i,j} = \sum_{\vec{C} \in \cH^{(2)}_i, \vec{C'} \in \cH^{(2)}_j : w_1 = w'_1} A_{i,j}^{(\vec{C}, \vec{C'})}$ and $A = \sum_{(i,j) \in M} b_i b_j A_{i,j}$.
\end{definition}
We now observe that each matrix $A_{i,j}^{(\vec{C}, \vec{C'})}$ has exactly $D^{4}$ nonzero entries, where $D = 2 \cdot {n - 2 \choose \ell - 1}$, and the matrix has $N^4$ rows/columns, where $N = {n \choose \ell}$. We note that $D/N \sim \ell/n$, and so the average number of nonzero entries per row (or column), i.e., the density, is $(D/N)^4 \sim (\ell/n)^4 = (\ell/n)^{q/2}$, as the arity of the constraints is $8$.

We also observe that for any $x \in \Fits^n$, $D^4 f_{M,b}(x) = {x'}^{\top} A x'$, where $x'$ is the vector with $(S_0, S_1, S'_0, S'_1)$-th entry equal to $\prod_{v \in S_0} x_v \prod_{v \in S_1} x_v \prod_{v \in S'_0} x_v \prod_{v \in S'_1} x_v$. We thus have that
\begin{flalign*}
k n^3 \cdot D^4 \leq D^4 \cdot \val(f_{M,b}) \leq \boolnorm{A} \leq N^4 \norm{A}_2 \enspace.
\end{flalign*}
For any $i \ne j$, the matrix $A_{i,j}$ has density $\sim m_{i,j} (D/N)^4 \sim (\ell/n)^4$, where $m_{i,j}$ is the number of the constraints in $f_{b}$ with right-hand side $b_i b_j$.
Let us now argue that each $m_{i,j}$ is at most $O(n^3)$. Indeed, $m_{i,j}$ is the number of pairs of $2$-chains $(i, C_0, w_0, C_1, w_1) \in \cH_i^{(2)}$ and $(j, C'_0, w'_0, C'_1, w'_1) \in \cH_j^{(2)}$ where $w_1 = w'_1$. To show that $m_{i,j} \leq O(n^3)$, we pick $w_0, w_1$ and $w'_0$, for a total of $n^3$ choices, and observe that this completely determines both chains. Indeed, because $H_i$ is a matching, there is at most one constraint $C$ in $H_i$ that contains $w_0$, and then $C_0$ must be $C \setminus \{w\}$. This similarly shows that we have at most one choice of $C_1$ and also $C'_0$. Finally, because $w'_1 = w_1$, and we know $w_1$, we thus know $w'_1$ as well, which by similar reasoning gives us at most one choice for $C'_1$, and we have determined the entire chain. We note that we have a lower bound of $\sim k n^3$ on the total number of constraints $\sum_{(i, j) \in M} m_{i,j}$, so this calculation also shows that no $m_{i,j}$ can be much larger than the average.

Returning to the density calculation, we have shown that $A_{i,j}$ has density at most $n^3 (\ell/n)^4 = \ell^4/n$. Again, following the blueprint in \cref{sec:evenqLDC}, we will set $\ell = n^{1/4} \cdot \polylog(n)$, and we want to show that the matrices $A_{i,j}$ satisfy the approximate regularity condition, i.e., the number of rows/columns with more than $\Delta = \ell^4 \cdot \polylog(n)/n$ nonzero entries is at most $N^4/\poly(n)$. Let us finish the proof, assuming that this holds.

\parhead{Proof assuming approximate regularity.} Let $\cB$ denote the set of rows/columns that are ``bad'' for some pair $(i,j)$, i.e., the matrix $A_{i,j}$ has more than $\Delta$ nonzero entries in that row. Let $B_{i,j}$ be the matrix where the rows and columns in $\cB$ have been all set to $0$. Let $B = \sum_{(i,j) \in M} b_i b_j B_{i,j}$. We have that $B$ is the sum of mean $0$ independent matrices, each with spectral norm $\norm{B_{i,j}}_2 \leq \Delta$. Therefore, by matrix Khintchine (\cref{fact:matrixkhintchine}), we have that with high probability over $b$, $\norm{B}_2 \leq O(\Delta \sqrt{k \log (N^4)}) = O(\Delta \sqrt{k \ell \log n})$.

Now, we observe that $\boolnorm{A - B} \leq o(N)$. This is because the number of nonzero entries that we have removed from $A$ to produce $B$ is at most $k \cdot n^3 \cdot N^4/\poly(n) = o(N^4)$ (there are $k$ edges $(i,j)$ in the matching $M$, each has $m_{i,j} \leq n^3$ constraints, and each row of $A_{i,j}$ has at most $m_{i,j} \leq n^3$ nonzero entries) provided that the $\poly(n)$ factor is large enough. We thus conclude that
\begin{flalign*}
k n^3 \cdot D^4 \leq D^4 \cdot \val(f_{M,b}) \leq \boolnorm{A - B} + N^4 \norm{B}_2 \leq o(N^4) + N^4 O(\Delta \sqrt{k \ell \log n}) \enspace.
\end{flalign*}
Substituting the value for $\Delta$ and rearranging, we conclude that $k \leq \ell \cdot \polylog(n) \leq \tilde{O}(n^{1/4})$.

We remark that \cref{sec:warmupstep1,sec:warmupstep2} are fairly mechanical, and they justify the use of the heuristic calculation. The place where we had ``freedom'' is in the choice of constraints to use in the initial XOR instance, which we chose to be the $2$-chains $\cH^{(2)}_i$. It thus remains to bound the number of bad rows $\cB$. This ``row pruning'' step is key to converting the heuristic into a full proof.

\subsection{Step 3: row pruning, the key technical step}
\label{sec:warmupstep3}
We want to understand if, after dropping a $1/\poly(n)$ fraction of the rows, every Kikuchi graph $A_{i,j}$ satisfies approximate regularity. This is equivalent to showing that for every matrix $A_{i,j}$, with probability at least $1-1/\poly(n)$ a uniformly random row $(S_0, S_1, S'_0, S'_1)$, has at most $\Delta$ nonzero entries in $A_{i,j}$ for $\Delta = \ell^4 \cdot \polylog(n)/n = \Delta_{avg} \polylog(n)$. 

\parhead{The heavy pair degree.} We now make a key observation. Whether the above approximate regularity property holds for a given collection of matchings $H_1, H_2,\ldots, H_n$ is governed by a single parameter that we call the \emph{heavy pair degree} $d$. This is the maximum, over all pairs $\{v,v'\} \subseteq [n]$, of the number of hyperedges across the $H_i$'s that contain $\{v,v'\}$. We will prove that if $d$ is small enough then approximate regularity holds for every $A_{i,j}$ after dropping a $1/\poly(n)$-fraction of rows. When $d$ is large, this property will not hold for the $A_{i,j}$'s from \cref{def:2chainkikuchi}. Instead, we will define a \emph{different} collection of Kikuchi matrices that have high density and for which row pruning succeeds. 

\begin{lemma}[Row pruning for $2$-chains with no heavy pairs]
\label{lem:2chainrowpruning}
Let $H_1, \dots, H_n$ be $3$-uniform hypergraph matchings of size $\delta n$, and let $d$ be the maximum, over all pairs $\{v,v'\}$ of vertices, of the number of pairs $(u,C)$ with $u \in [n]$ and $C \in H_u$ where $\{v,v'\} \subseteq C$. Fix $i \ne j \in [k]$, and let $A_{i,j}$ be the matrix defined in \cref{def:2chainkikuchi} at level $\ell \in \N$.

Suppose that $d \leq \ell^2$. Then, the number of rows $(S_0, S_1, S'_0, S'_1)$ of $A_{i,j}$ with more than $\Delta = \ell^4 \cdot \polylog(n)/n$ nonzero entries is at most $N^4/\poly(n)$.
\end{lemma}

We note that if the matchings $H_1, \dots, H_n$ are \emph{random}, then we have $d \leq \polylog(n)$ with high probability, and so random matchings satisfy the ``small heavy-pair degree'' assumption with high probability. We can thus think of $d \leq \polylog(n)$ as a pseudorandom property of a collection $H_1, \dots, H_n$ of matchings. We now sketch a proof of \cref{lem:2chainrowpruning}. 

\parhead{The degree polynomial and its partial derivatives.} As the first step in the proof of \cref{lem:2chainrowpruning}, we define a degree $4$ polynomial $\Deg_{i,j}\colon\Bits^{4n} \to \N$, where we think of the $4n$ variables as split into $4$ groups of $n$ variables $s^{(0)}, s^{(1)}, s'^{(0)}, s'^{(1)}$, which are indicator variables of the $4$ sets $S_0, S_1, S'_0, S'_1$, respectively. This polynomial $\Deg_{i,j}(s^{(0)}, s^{(1)}, s'^{(0)}, s'^{(1)})$ upper bounds the number of nonzero entries in the $(S_0, S_1, S'_0, S'_1)$-th row in the matrix $A_{i,j}$ in \cref{def:2chainkikuchi}.

Formally, let $\cT_{i,j}$ denote the (multi)-set of $4$-tuples $(u_0, u_1, v_0, v_1)$ such that there exists $\vec{C} = (i, C_0, w_0, C_1, w_1) \in \cH^{(2)}_{i}$ and $\vec{C'} = (j, C'_0, w'_0, C'_1, w'_1) \in \cH^{(2)}_{j}$ with $w_1 = w'_1$ such that $u_0 \in C_0, u_1 \in C_1, v_0 \in C'_0, v_1 \in C'_1$; if there are multiple such pairs $(\vec{C}, \vec{C'})$ that produce the same $(u_0, u_1, v_0, v_1)$, then we add this tuple multiple times. 
Then, we set
\begin{equation*}
\Deg_{i,j}(s^{(0)}, s^{(1)}, s'^{(0)}, s'^{(1)}) \coloneqq \sum_{(u_0, u_1, v_0, v_1) \in \cT_{i,j}} s^{(0)}_{u_0} s^{(1)}_{u_1} s'^{(0)}_{v_0} s'^{(1)}_{v_1}\enspace.
\end{equation*}
Note that $\Deg_{i,j}$ is a polynomial with non-negative coefficients. We are interested in the probability that $\Deg_{i,j}$, on uniform draws of $4$-tuples of $\ell$-size sets, takes a value that deviates from its expectation $\mu$ by some multiplicative factor. It turns out (see \cref{lem:coupling}) that we can pass on to independent $p$-biased product distribution on $\Bits^{4n}$ for $p \sim \ell/n$ without much loss. This is helpful because the tail behavior of low-degree polynomials with non-negative coefficients on product distributions is determined by a bound on its expected partial derivatives. Namely, variants of the Kim-Vu inequality (see~\cref{lem:partitepolyconc}) show the following: \emph{if the expectation of every partial derivative of $\Deg_{i,j}$ is at most $\mu$, then $\Deg_{i,j}(S_0, S_1, S'_0, S'_1) \leq O(\mu \log n)$ with probability at least $1-1/\poly(n)$}. 

Let us now examine the expected partial derivatives of $\Deg_{i,j}(s)$. We start by introducing notation to refer to them. Let $Z = (z_0, z_1, z'_0, z'_1) \in ([n] \cup \{\star\})^{4}$ be an ordered tuple of length $4$, with entries either in $n$ or set to $\star$, which we think of as an ``unfixed'' value. Then, $Z$ encodes partial derivatives with respect to any subset of variables that use at most one variable in each of the groups $s^{(0)}$, $s^{(1)}$, $s'^{(0)}$, $s'^{(1)}$. All other partial derivatives of $\Deg_{i,j}$ are $0$ since $\Deg_{i,j}$ has degree $1$ in each of the $4$ groups of variables (i.e., $\Deg_{i,j}$ is $4$-partite). We know that $\E[\Deg_{i,j}(s)] = \mu_{(\star, \star, \star, \star)} \leq 2^4 (\ell/n)^4 \cdot n^3 = O(1) \cdot \ell^4/n$; the factor of $2^4$ comes from the fact that each pair $(\vec{C}, \vec{C'})$ adds $2^4$ different tuples to $\cT_{i,j}$. Now, \cref{lem:partitepolyconc} implies that the chance that $\Deg_{i,j}$ takes a value larger than $\mu \cdot \polylog(n)$ is at most $1/\poly(n)$ if $\mu_Z \leq \mu$ for all $Z$.

\parhead{Computing expected partial derivatives} To help bound the expected partial derivatives $\mu_Z$, let us relate these parameters to combinatorial quantities of the hypergraphs $H_1, H_2,\ldots, H_n$. Notice that when we take partial derivatives with respect to some $Z$, the only monomials that ``survive'' are ones that ``contain'' $Z$, and furthermore the expectation of the partial derivative is simply $(\ell/n)^{\text{\# of $\star$ entries in $Z$}}$ times the number of such monomials. Formally, let $\deg_{i,j}(Z)$ be the number of pairs $(\vec{C}, \vec{C'}) \in \cH_i^{(2)} \times \cH_j^{(2)}$ where $w_1 = w'_1$ and $z_0 \in C_0, z_1 \in C_1, z'_0 \in C'_0, z'_1 \in C'_1$, where for the symbol $\star$, we say that $\star \in C$ always holds --- we say that such a pair $(\vec{C}, \vec{C'})$ \emph{contains} $Z$. Then, the expected partial derivative at $Z$ is $\mu_Z = 2^{4 - \abs{Z}}(\ell/n)^{4 - \abs{Z}} \deg_{i,j}(Z)$, where $\abs{Z}$ is the number of non-$\star$ entries in $Z$.\footnote{The extra factor of $2^{4 - \abs{Z}}$ comes from the fact that for every $Z$ and pair $(\vec{C}, \vec{C'})$ containing $Z$, the pair $(\vec{C}, \vec{C'})$ produces  $2^{4 - \abs{Z}}$ tuples $(u_0, u_1, v_0, v_1)$ in $\cT_{i,j}$ that contain $Z$. In this case, this is just a constant factor, so we can ignore it.} For example, $Z = (\star, \star, \star, \star)$ is contained in all such pairs of $2$-chains, and so $\deg_{i,j}(\star, \star, \star, \star) = m_{i,j} \leq O(n^3)$ and $\mu_Z = \mu = 16 (\ell/n)^4 m_{i,j}$. Let us use the shorthand $\mu_{t} = \max_{Z : \abs{Z} = t} \mu_Z$.

Let $Z$ be an arbitrary $4$-tuple with at least one non-$\star$ entry. As explained above, estimating $\mu_Z$ is, up to scaling, equivalent to counting $\deg_{i,j}(Z)$, the number of pairs $(\vec{C},\vec{C'})$ that contain $Z$. We next observe that if $Z$ has no $\star$ entries, then the number of $2$-chains $(\vec{C},\vec{C'})$ containing $Z$ is an absolute constant. This is because there is at most one constraint $C_0 \cup \{w_0\}$ that contains $z_0$ in $H_i$. Given this constraint, there are $2$ choices for $w_0$, as $w_0 \in C_0 \cup \{w_0\} \setminus \{z_0\}$. Given $w_0$, there is at most one constraint $C_1 \cup \{w_1\}$ in $H_1$ that contains $z_1$, and then at most $2$ choices for $w_1$. We can similarly use the knowledge of $(z_0',z_1')$ to bound the number of choices for $C_0',C_1'$. All in all, we have at most $16 = O(1)$ choices for the pair $(\vec{C},\vec{C'})$ given $Z$ with no $\star$ entries. This immediately shows that for $Z$ such that $\abs{Z} = 4$, $\mu_Z \leq O(1) \leq \mu$. 

Let us now deal with $Z$'s with at least one $\star$ entry by breaking up into cases depending on $\abs{Z}$. We will view the counting of $\deg_{i,j}(Z)$ as a procedure that makes a bounded number of choices to decode the pair $(\vec{C}, \vec{C'})$.

Let us deal with the case when $\abs{Z} = 1$. By swapping the roles of $i$ and $j$ if needed, without loss of generality we can assume that one of $z_0$ or $z_1$ is non-$\star$, and all other entries in $Z$ are $\star$. There are at most $n$ choices for $z_0$ (if $z_1 \ne \star$) or $z_1$ (if $z_0 \ne \star$). We now have $n$ choices for $z'_0$, which again determines $C'_0$ and $w'_0$ up to $2$ choices. We now observe that $(C'_1, w'_1)$ is uniquely determined. Indeed, this is because we know $w'_1$, as it equals $w_1$ (the two $2$-chains must have matching tails), and therefore this determines the hyperedge $C'_1 \cup \{w'_1\} \in H_{w'_0}$ uniquely. We have thus shown that for $Z$ with $\abs{Z} = 1$, we have $\deg_{i,j}(Z) \leq O(n^2)$, and so $\mu_Z \leq (\ell/n)^3 \cdot O(n^2) \leq O(\ell^3/n) \leq O(\ell^4/n)$.

Let us now handle the case when $\abs{Z} = 2$. Similar arguments as above show that $\Deg_{i,j}(Z) \leq O(n)$ holds for all $Z$ except when the non-$\star$ entries of $Z$ look like $Z = (\star, z_1, \star, z'_1)$ where $z_1, z'_1 \ne \star$, and thus $\mu_Z \leq (\ell/n)^2 \cdot O(n) \leq O(\ell^4/n)$ for these $Z$'s. To count $\deg_{i,j}(Z)$ for $Z = (\star, z_1, \star, z'_1)$ where $z_1, z'_1 \ne \star$, we pay a factor of $n$ to determine $z_0$, and then this determines (up to an $O(1)$ factor) $C_0$ and $C_1$ as well. Now, we know $w'_1$ (because it is equal to $w_1$) and $z'_1$ which is in $C'_1$. Thus, the hyperedge $C'_1 \cup \{w'_1\}$ must contain the pair $\{z'_1, w'_1\}$. Using the heavy pair degree, there are at most $d$ choices for the pair $(w'_0, C'_1 \cup \{w'_1\})$, and after learning $w'_0$ we also know $C'_0$. Hence, we have paid a total of $O(n d)$ choices, which implies that $\mu_2 \leq (\ell/n)^2 \cdot O(nd) = O(\ell^2 d/n)$. For $\abs{Z} = 3$, a similar issue arises and gives a bound of $\mu_3 \leq O(\ell d/n)$. 

We can now finish the proof of \cref{lem:2chainrowpruning}.
\begin{proof}[Proof of \cref{lem:2chainrowpruning}]
Notice that if $d \leq \ell^2$ then $\mu_t \leq \mu$ for every $t$. Applying \cref{lem:partitepolyconc} now yields that the probability that $\Deg_{i,j} > \mu \cdot \polylog(n)$ is at most $1/\poly(n)$. Taking a union bound on $k < n$ yields that the fraction of bad rows $|\cB|/N$ is at most $1/\poly(n)$, as desired. 
\end{proof}

\subsection{Step 4: hypergraph decomposition to handle large heavy pair degree}
\label{sec:2chaindecomp}
We will handle the case when the heavy pair degree is high by designing a \emph{different} Kikuchi matrix. To do this, we we will construct the cross term polynomial (obtained by applying the Cauchy--Schwarz inequality) slightly differently. Our current Kikuchi matrix is built from the XOR instance obtained by pairing up chains that agree on their tails and thus ``cancel'' (i.e., square out) one variable. When the heavy pair degree is large, we will build chains by cancelling a pair of variables instead. The number of pairs of chains that agree in a pair of variables instead of just their tails, i.e., the new number of ``Cauchy--Schwarzed'' constraints, will of course be smaller than before. On the other hand, since we cancel a pair of variables instead of just the tail, the arity of the resulting XOR instance will be smaller: $6$ instead of $8$. The punchline is that the density vs.\ arity trade-off (i.e., our key heuristic discussed in \cref{sec:chainheuristic}) breaks in our favor, \emph{provided that there are many ``heavy pairs''.}

To formally implement this argument, we \emph{decompose} the set of chains by ``labeling'' each chain by the heavy pair contained within, if one exists. Intuitively, this is the pair of variables in the chain that we intend to cancel in the Cauchy--Schwarz trick. If the chain does not contain any heavy pair, then we label it by its tail variable $w$, which we will cancel in the Cauchy--Schwarz trick as done before in \cref{sec:warmupstep1}. We let $\cH_{Q}$ denote the set of chains labeled by the heavy pair $Q$, and $\cH_w$ denote the set of chains labeled by the tail variable $w$. For technical reasons (that will become relevant when we do the row pruning argument for the different, yet-to-be-defined Kikuchi matrices), our decomposition will produce multiple pieces labeled by the \emph{same} heavy pair $Q$, i.e., $\cH_{Q,1}$, $\cH_{Q,2}$, etc., and for two chains labeled by the same $Q$, we shall only  cancel the pair $Q$ if these two chains lie within the \emph{same} piece $\cH_{Q,p}$.

Formally, our hypergraph decomposition is as follows. Given the collection $\cH^{(1)} = \{(u,C, w) : u \in [n], C \cup \{w\} \in H_u\}$ of $1$-chains, we perform the following greedy algorithm: if there exists an ordered pair $Q = (Q_1, Q_2)$ such that there are more than $d \coloneqq \ell^2$ $1$-chains $(u,C, w)$ in $\cH^{(1)}$ with $Q_1 \in C$ and $Q_2 = w$, i.e., $Q$ is a heavy pair contained in the chain $(u, C, w)$, then we choose an arbitrary set of \emph{exactly} $d$ such $1$-chains, remove them from $\cH^{(1)}$, and place them in a new ``partition'' $\cH_{Q, p}$; here, $p \in \N$ denotes the ``label'' of the partition, as we may be producing multiple partitions with the same $Q$, and so we will denote these different pieces of the partition by $\cH_{Q,1}$, $\cH_{Q,2}$, etc. Finally, if there is no such heavy pair $Q$, then we create partitions $\cH_{w}$ for each $w \in [n]$, and add all remaining $1$-chains with ``tail $w$'', i.e., $1$-chains of the form $(u,C,w)$, to $\cH_{w}$.

This decomposition has the following properties:
\begin{enumerate}[(1)]
\item $\cH^{(1)} = (\cup_{w} \cH_w) \bigcup (\cup_{(Q,p)} \cH_{Q,p})$ is a disjoint partition of $\cH^{(1)}$;
\item For each $Q = (Q_1, Q_2)$ and $p \in \N$, $\cH_{Q,p}$ is a set of $1$-chains that ``contain'' the tuple $Q$, i.e., each $(u,C,w)$ in $\cH_{Q,p}$ has $w = Q_2$ and $C \ni Q_1$;
\item For each $Q$ and $p \in \N$, $\abs{\cH_{Q, p}} = d$;
\item For each $w \in [n]$, there is only one partition $\cH_{w}$;
\item The total number of partitions $\cH_{Q,p}$ is at most $O(n^2/d)$, as there are at most $O(n^2)$ $1$-chains, and each $\cH_{Q,p}$ has exactly $d$ $1$-chains.
\end{enumerate}
We stress that the decomposition is only on $1$-chains, \emph{not} the set of $2$-chains $\cup_{i \in [k]} \cH^{(2)}_i$ that are the constraints in the XOR instance! At a high level, this is because, e.g., the $2$-chains in $\cH_{i}^{(2)}$ (or $\cH_j^{(2)}$) are formed by taking a $1$-chain and \emph{prepending} it with a hyperedge in $H_i$ (or $H_j$), and so ``first link'' in each $2$-chain is specific to the choice of $i \in [k]$, but the ``second link'' is an arbitrary $1$-chain, and so it is ``shared'' across the $\cH_{i}^{(2)}$'s in some informal sense.\footnote{For this reason, in \cref{sec:regular-partition}, the length of the chains defining the XOR constraints is $r+1$, but we only decompose length $r$ chains.} This property turns out to be important when it comes time to bound the expected partial derivatives.

Now, we define $\cH_{i, Q, p}^{(2)}$ to be the set of $2$-chains $(i, C_0, w_0, C_1, w_1)$ where the ``second link'' $(w_0, C_1, w_1)$ is in $\cH_{Q,p}$. Using the decomposition, we now define the following polynomials:
\begin{flalign*}
&\Phi_b(x) \coloneqq \sum_{i = 1}^k b_i \sum_{\vec{C} = (i, C_0, w_0, C_1, w_1) \in \cH_i^{(2)}} x_{C_0} x_{C_1} x_{w_1} \enspace,\\
&\Psi_{i, w}(x) \coloneqq \sum_{C_0, w_0 : C_0 \cup \{w_0\} \in H_i}  \sum_{(w_0, C_1, w_1) \in \cH_{w}} x_{C_0} x_{C_1} \enspace,\\
&\Psi_{i, Q, p}(x) \coloneqq \sum_{(i, C_0, w_0, C_1, w_1) \in \cH^{(2)}_{i, Q,p}} x_{C_0} x_{C_1 \setminus Q_1} \enspace, \\
&\Psi^{(0)}_{b}(x,y) \coloneqq \sum_{i = 1}^k \sum_{w \in [n]} b_i y_{w}\Psi_{i,w}(x) \enspace, \\
&\Psi^{(1)}_{b}(x,y) \coloneqq \sum_{i = 1}^k \sum_{(Q,p)} b_i y_{Q,p} \Psi_{i,Q, p}(x) \enspace,
\end{flalign*}
where above $y_{Q,p}$ and $y_w$ are new variables.
By definition, if we set $y_w = x_w$ and $y_{Q} = x_{Q_1} x_{Q_2}$, then we have that $\Phi_b(x) = \Psi^{(0)}(x,y) + \Psi^{(1)}(x,y)$. Indeed, all  we have done is partitioned the constraints into these two polynomials and removed the ``$x_{Q_1} x_{Q_2}$ term'' from each monomial, replacing it with the new variable $y_{Q,p}$.

We now refute the two polynomials $\Psi^{(0)}(x,y)$ and $\Psi^{(1)}(x,y)$ separately using the machinery in \cref{sec:warmupstep1,sec:warmupstep2,sec:warmupstep3}. In fact, \cref{sec:warmupstep1,sec:warmupstep2,sec:warmupstep3} immediately show that we can successfully refute the polynomial $\Psi^{(0)}(x,y)$. Indeed, the only issue that we encountered was in \cref{sec:warmupstep3}, where the row pruning failed if there was a pair $\{v,v'\}$ that appeared in more than $\ell^2$ $1$-chains in $\cH^{(1)}$. However, this cannot happen, as otherwise our decomposition algorithm would not have terminated.

It thus remains to handle the second polynomial, $\Psi^{(1)}(x,y)$. Applying the ``Cauchy--Schwarz trick'' of \cref{sec:warmupstep1}, we can reduce this to the case of bounding the polynomial:
\begin{flalign*}
f_{M,b}(x) = \sum_{(i,j) \in M} b_i b_j \sum_{(Q,p)} \Psi_{i,Q,p}(x) \Psi_{j,Q,p}(x) \enspace,
\end{flalign*}
where $M$ is a maximum matching, as before. Notice that the constraints in $f_{M,b}$ have arity $6$ (see \cref{fig:2chainderived}). Following the blueprint of~\cref{sec:warmupstep2}, we define the following Kikuchi matrices.
\begin{definition}
\label{def:2chainkikuchiderived}
For $i \ne j \in [k]$, $(Q,p)$, and $\vec{C} = (i, C_0, w_0, C_1, w_1) \in \cH^{(2)}_{i, Q, p}$, $\vec{C'} = (j, C'_0, w'_0, C'_1, w'_1) \in \cH^{(2)}_{j, Q,p}$, we define the matrix $A_{i,j, Q, p}^{(\vec{C}, \vec{C'})}$ as follows. The matrix $A_{i,j, Q,p}^{(\vec{C}, \vec{C'})}$ is indexed by a $3$-tuple of sets $(S_0, R, S'_0)$, each in ${[n] \choose \ell}$, and the $(S_0, R, S'_0), (T_0, W, T'_0)$-th entry is $1$ if $S_0 \oplus T_0 = C_0$, $S'_0 \oplus T'_0 = C'_0$, and $R = \{u\} \cup U$, $W = \{v\} \cup V$, where $C_1 = \{u, Q_1\}$, $C'_1 = \{v, Q_1\}$, and $U \subseteq [n]$ is a set of size $\ell - 1$ where $u,v \notin U$.

We let $A_{i,j} = \sum_{Q, p} \sum_{\vec{C} \in \cH^{(2)}_{i, Q, p}, \vec{C'} \in \cH^{(2)}_{j, Q, p}} A_{i,j, Q,p}^{(\vec{C}, \vec{C'})}$ and $A = \sum_{(i,j) \in M} b_i b_j A_{i,j}$.
\end{definition}

\begin{figure}[t]
    \centering
    \includegraphics[width=0.9\textwidth]{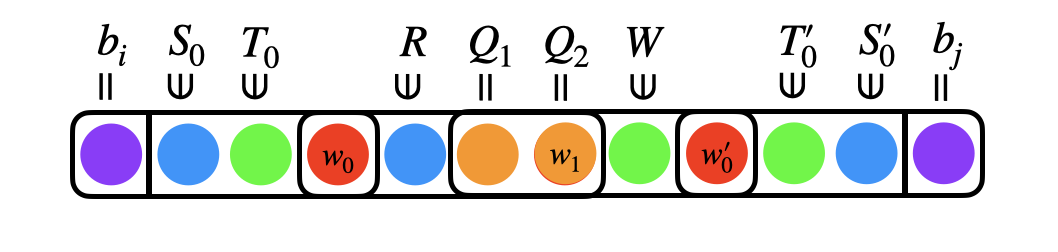}
    \caption{A pair of $2$-chains $\vec{C} = (i, C_0, w_0, C_1, w_1) \in \cH^{(2)}_{i, Q, p}$, $\vec{C'} = (j, C'_0, w'_0, C'_1, w'_1) \in \cH^{(2)}_{j, Q,p}$. The blue vertices appear in the sets $(S_0, R, S'_0)$ for the rows of the matrix $A_{i,j,Q,p}^{(\vec{C}, \vec{C'})}$, and the green vertices appear in the columns. The orange elements are the elements of $Q$ that are canceled via the Cauchy--Schwarz operation.}
    \hrulefill
        \label{fig:2chainderived}
\end{figure}

Notice that for $\vec{C} = (i, C_0, w_0, C_1, w_1) \in \cH^{(2)}_{i, Q, p}$ and $\vec{C'} = (j, C'_0, w'_0, C'_1, w'_1) \in \cH^{(2)}_{j, Q,p}$, the split of the elements in the constraint across the row $(S_0, R, S'_0)$ and the column $(T_0, W, T'_0)$ is asymmetric: see \cref{fig:2chainderived}.

Applying the same machinery in \cref{sec:warmupstep2} to the matrices in \cref{def:2chainkikuchiderived} will yield the correct lower bound provided that the row pruning step succeeds. It thus remains to bound the number of rows in $A_{i,j}$ for a fixed pair $(i,j)$ with a number of nonzero entries exceeding the average by a $\polylog(n)$ factor.

We now apply \cref{lem:partitepolyconc}. As before, we define a similar degree polynomial $\Deg_{i,j}$, and the tail bound boils down to computing the expected partial derivatives $\mu_Z$, where $Z = (z_0, r, z'_0) \in ([n] \cup \{\star\})^3$ is now a tuple of length $3$, and $\mu_Z = (\ell/n)^{3 - \abs{Z}} \deg_{i,j}(Z)$, as the constraints have arity $3$. We observe that $\deg_{i,j}(\star, \star, \star) \leq O(n^2 d)$, as we have $O(n^2)$ choices for $\vec{C} = (i, C_0, w_0, C_1, w_1) \in \cH^{(2)}_{i}$ (which then determines $(Q,p)$), followed by $O(d)$ choices for $(w'_0, C'_1, w'_1)$ (because this must be in $\cH^{(1)}_{Q,p}$, which has size $d$), and then a unique choice for $C_0$.
Therefore, $\mu_0 \leq (\ell/n)^3 \cdot O(n^2 d) = O(\ell^3 d/n)$.

Bounding $\mu_1$ is straightforward, and we omit the calculations. We obtain a bound of $\mu_1 \leq (\ell/n)^2 \cdot O(n d) = O(\ell^2 d/n)$. Bounding $\mu_2$ can be done with a trivial bound of $\deg_{i,j}(Z) \leq O(n)$, yielding $\mu_2 \leq (\ell/n) \cdot O(n) = O(\ell)$.
Finally, it is simple to bound $\deg_{i,j}(Z) \leq O(1)$ when $\abs{Z} = 3$, and so we obtain $\mu_3 \leq O(1)$.

We notice that $\mu_0 \geq \mu_1$ and $\mu_2 \geq \mu_3$ always hold. So, either $\mu_0$ or $\mu_2$ must be the maximum. Because $d = \ell^2$, we have $\mu_0 = O(\ell^3 d/n) \sim \ell^5/n \gg \ell \sim \mu_2$ because $\ell^4 \gg n$, by choice of $\ell$. Thus, $\mu_0 \gg \mu_2$, and so the row pruning argument, etc., will all succeed. This, combined with the refutation argument for $\Psi^{(0)}_{b}(x)$, implies that our heuristic calculation succeeds and we get a bound of $k \leq \tilde{O}(\ell)$, where $\ell$ is chosen to be $\tilde{O}(n^{1/4})$. Thus, we obtain a lower bound of $k \leq \tilde{O}(n^{1/4})$.
\end{proof}

\subsection{Preview: extending the warmup to a proof of \cref{mthm:main}}
We now give a brief overview of how we shall extend the ideas used in this warmup to prove \cref{mthm:main}. First, we observe that in the argument we presented in \cref{sec:warmupstep1,sec:warmupstep2,sec:warmupstep3,sec:2chaindecomp}, there were only two crucial moments in the proof where we had a lot of freedom: (1) the choice of the constraints in the initial XOR instance (in this warmup, we chose the set of $2$-chains with head $i \in [k]$), and (2) the choice of the hypergraph decomposition in \cref{sec:2chaindecomp} --- the rest of the proof was fairly mechanical, and boiled down to computing the expected partial derivatives $\mu_Z$. Namely, if we can choose the constraints and the decomposition so that the row pruning succeeds for all the resulting Kikuchi matrices, i.e., the expected partial derivatives of the degree polynomials are appropriately bounded, then the general machinery in \cref{sec:warmupstep1,sec:warmupstep2,sec:warmupstep3} succeeds in proving the lower bound predicted by the heuristic calculation in \cref{sec:chainheuristic} (up to a small loss, see \cref{rem:losing-on-the-heuristic}).

As discussed in \cref{sec:chainheuristic}, we shall define the XOR instance using $(r+1)$-chains for a parameter $r = O(\log n)$, and the heuristic calculation predicts that this will yield an exponential lower bound.
Thus, the key technical component of the proof is to choose the decomposition of the $(r+1)$-chains so that the degree polynomials of the resulting Kikuchi matrices all satisfy the bounded expected partial derivatives condition. In \cref{sec:2chaindecomp}, we showed how to do this for the case when $r = 1$.

We now wish to point out the following crucial observation: the decomposition in \cref{sec:2chaindecomp} is ``informed'' by the row pruning calculation for the \emph{undecomposed chains} done in \cref{sec:warmupstep3}. Specifically, in \cref{sec:warmupstep3}, we argued that if there is a violating partial derivative for the undecomposed chains, then there is some combinatorial structure in the chains (namely, a heavy pair) that is the ``cause'' of the large expected partial derivative, and this combinatorial structure is exactly the criteria that we use to decompose the hypergraph. In some sense, the hypergraph decomposition (along with the modified Cauchy--Schwarz trick and Kikuchi matrices) can be thought of as a precise way to ``fix'' this high expected partial derivative. For longer chains, there is once again an intimate relationship between the existence of a violating expected partial derivative and a certain ``denser-than-anticipated'' combinatorial structure (analogous to heavy pairs) being present in the chains we construct. For larger chains, this structure is a more complicated to describe, but an analogous chain decomposition for this structure accomplishes the same job.

More precisely, we generalize the decomposition of \cref{sec:2chaindecomp} as follows. As done in \cref{sec:2chaindecomp}, we shall think of an $(r+1)$-chain in $\cH^{(r+1)}_i$ as being split into two subchains, the ``first link'' in $H_i$ and then the rest of the chain, which is an $r$-chain. As before, our decomposition shall decompose the $r$-chain part only, and this induces a decomposition of the $(r+1)$-chains in $\cH^{(r+1)}_i$. Recall that in \cref{sec:2chaindecomp}, we decomposed a $1$-chain $(u, C, w)$ by picking a $Q$ where $Q_1 \in C$ and $Q_2 = w$. Notice that $Q$ only contains one element of the hyperedge $C$; there was no need to do a further decomposition to handle, e.g., heavy triples $Q = (Q_1, Q'_1, Q_2)$ where $\{Q_1, Q'_1\} = C$ and $Q_2 = w$.

Now, we have $r$-chains $(u, C_1, w_1, \dots, C_r, w_r)$, and we shall decompose if there is a $Q = (Q_1, \dots, Q_{r+1}) \in ([n] \cup \{\star\})^{r} \times [n]$ such that (1) $Q$ is heavy, i.e., is contained in many $r$-chains, meaning that (a) $Q_{h+1} = w_r$, and so in particular $Q_{h+1} \ne \star$, and (b) $Q_h \in C_h$ for $h = 1, \dots, r$; and (2) $Q$ is \emph{contiguous}, meaning that if $h \in [r+1]$ is the minimal $h$ such that $Q_h \ne \star$, then $Q_{h'} \ne \star$ for all $h' \geq h$, i.e., $Q$ has $\star$'s followed by only non-$\star$ entries.

Condition (1) above is a somewhat natural extension of the decomposition method in \cref{sec:2chaindecomp}, but condition (2) is trickier. It turns out (in a somewhat subtle way) that because the $H_i$'s are matchings, if there is a violating expected partial derivative, then not only is there a heavy $Q$, but there must be a heavy \emph{contiguous} $Q$. In a sense (that can be made precise), the contiguous $Q$'s are \emph{irreducible} violations and thus it is enough to only handle them.

\section{Proof of \cref{mthm:main}: From LCCs to XOR Formulas}
\label{sec:lcctoxor}

We now present the proof of \cref{mthm:main} for the case of $\F = \F_2$. The proof is spread over \cref{sec:lcctoxor,sec:regular-partition,sec:kikuchimethod,sec:rowpruning} and follows the steps in the warmup. In the current section, we define $r$-chains and the family of XOR instances associated to the LCC that we wish to refute. Then, in \cref{sec:regular-partition}, we decompose the $r$-chains, and thereby decompose the $(r+1)$-chains forming the constraints in the XOR instance. Then, in \cref{sec:kikuchimethod}, we define the Kikuchi matrices and finish the argument up to the proof of the row pruning lemma, \cref{lem:rowpruning}, an analogue of \cref{lem:2chainrowpruning} that is the key technical lemma. Finally, in \cref{sec:rowpruning}, we prove \cref{lem:rowpruning}.

Let $\Code \colon \F_2^k \to \F_2^n$ be $(3, \delta, \eps)$-locally correctable.
Without loss of generality, by \cref{fact:normalform} we can assume that $\Code$ is $(3, \delta')$-normally decodable, where $\delta' \geq \delta/6$ and $n' = 2n$. For the remainder of the proof, we will redefine $\delta$ to be $\delta'$, and $n$ to be $2n$.
We shall also think of the code $\Code \colon \F_2^k \to \F_2^n$ as a map $\Code \colon \Fits^k \to \Fits^n$.

We will now define satisfiable XOR formulas $\Phi$ associated with the linear code $\Code$.
Let $\Code:\Fits^k \to  \Fits^n$ be a linear $(3, \delta)$-normally correctable code. Recall that without loss of generality, $\Code$ is systematic, meaning that the first $k$ bits of $\Code$ are the message bits. In particular, for every $b \in \Fits^k$, there is a unique $x \in \Code$ such that $x \vert_{[k]} = b$. We can thus generate $x \gets \Code$ uniformly at random by first choosing $b \gets \Fits^k$ uniformly at random, and then setting $x$ to be the unique extension of $b$. 

Since $\Code$ is a linear $(3, \delta)$-normally correctable code, there exist $3$-uniform hypergraph matchings $H_1, \dots, H_n$, each of size exactly $\delta n$, such that every $x \in \Code$ satisfies the following system of $4$-XOR constraints, i.e., each constraint has arity $4$: 
\begin{equation} \label{eq:naive-xor}
\forall u \in [n], C \in H_u, \text{  } x_C x_u = 1\mper
\end{equation}

We will construct an XOR formula by \emph{long chain} derivations. Intuitively, a long chain derivation starts from the natural XOR constraints \eqref{eq:naive-xor} and derives new ones by chaining together $t$ constraints with an appropriate combinatorial structure. Below, we formalize the set of constraints in this formula as a family of hypergraphs built from the $H_u$'s. 

\begin{definition}[$t$-chain hypergraph $\cH^{(t)}$]
Let $t \geq 1$ be an integer. For any $u \in [n]$, let $\cH_u^{(t)}$ denote the set of tuples of the form $(u,C_1, w_1, C_2, w_2, \dots, C_t, w_t)$, where each $C_h \in {[n] \choose 2}$, $w_h \in [n]$, and it holds that for all $1 \leq h \leq t$, $C_h \cup \{w_h\} \in H_{w_{h - 1}}$ where we set $w_0 \coloneqq u$. We call $u$ the head, $w_h$'s the \emph{pivots} for $1 \leq h \leq t - 1$ and $w_t$ the \emph{tail} in such a chain. We let $\cH^{(t)} = \cup_{u \in [n]} \cH_u^{(t)}$ denote the set of all $t$-chains, where $\cH^{(t)}_u$ is the set of $t$-chains with head $u$.
\end{definition}

The following simple observation helps us understand the combinatorial structure in the chains. 

\begin{observation}
\label{obs:monomial-associated-with-chain}
Let $x = \Code(b)$ for a linear LCC over $\F_2$ with $\{H_u\}_{u \in [n]}$ being the associated matchings. Then, for any $t$-chain $(u,C_1, w_1, C_2, w_2, \ldots, C_t, w_t)$, $x$ satisfies $x_u x_{w_t} \prod_{h= 1}^t x_{C_h} = 1$. 
\end{observation}

\begin{proof}
We know that $x$ satisfies $x_{w_h} x_{C_{h+1}} x_{w_{h+1}} =1$ for every $0 \leq h \leq t$ where we define $w_0 = u$. Taking products of the left-hand sides of each of these $t$ equations, we observe that for every $1 \leq h \leq t-1$, $x_{w_h}$ is ``squared out'' (since $x_{v}^2 = 1$ for every $v \in [n]$), and this finishes the proof.
\end{proof}

\begin{figure}[t]
    \centering
    \includegraphics[width=0.9\textwidth]{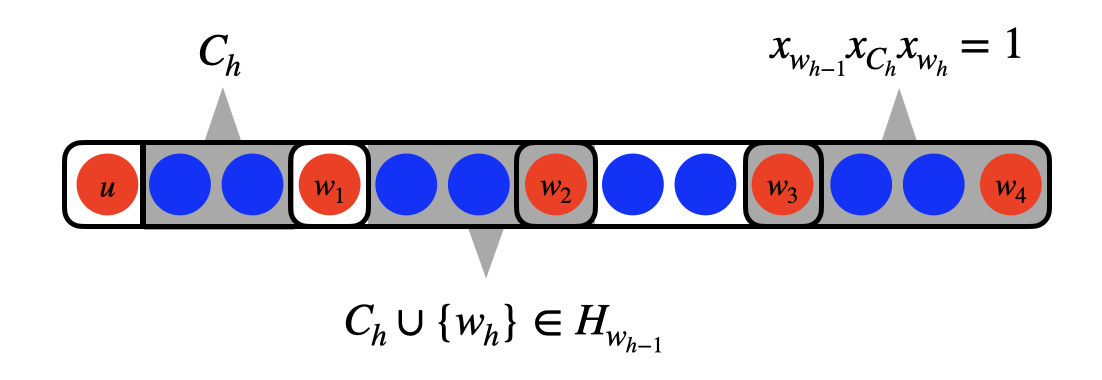}
    \caption{A $4$-chain. The pairs of blue vertices are the $C_{h}$'s, and the red vertices are the $w_{h}$'s. Note that for any $x \in \Code$, we have $x_{w_{h-1}} x_{C_h} x_{w_{h}} = 1$.}
    \hrulefill
        \label{fig:basic-chain}
\end{figure}

\parhead{Building chains iteratively.} It is useful to think of $t$-chains as being built by extending smaller chains by iteratively adding hyperedges to the head (i.e. to the left). The following notation and observation formalizes this.
\begin{definition}[Extending Chains]
For the $t$-chain hypergraph $\cH^{(t)}$ built from $3$-matchings $H_1, H_2, \ldots, H_n$ on $[n]$, we define $H_u \circ \cH^{(t+1)}$ as:
\begin{equation*}
H_u \circ \cH^{(t)} = \cup_{w_0 \in [n]} \Set{ (u,C_0, \vec{C} \mid \vec{C} \in \cH_{w_0}^{(t)}, \{C_0 \cup \{w_0\} \in H_{u}} \enspace.
\end{equation*}
\label{def:composition-of-chains}
\end{definition}

\begin{observation} \label{obs:iterative-extending-chains}
For $t \geq 1$, let $\cH^{(t)}$ be the $t$-chain hypergraph built from $3$-matchings $H_1, H_2, \ldots, H_n$ on $[n]$. Then, $\cH^{(t+1)} = \cup_{u \in [n]} H_u \circ \cH^{(t)} =\cup_{u \in [n]} \cH^{(t')}_u \circ \cH^{(t-t')}$ for any $0<t' < t$. 
\end{observation}

\parhead{Chains that fix some positions.} We will often refer to the set of chains where some of the $C_h$'s are forced to contain some $v_h \in [n]$. Towards this, we introduce the following terminology. 
\begin{definition}[Chains containing $Q$]
For any $Q = (Q_1, \dots, Q_t, Q_{t+1}) \in \{[n] \cup \star\}^{t} \times [n]$, we say that a chain $(u, C_1, w_1, \ldots, C_t, w_t) \in \cH^{(t)}$ \emph{contains} $Q$ if $Q_{t+1} = w_t$ and for $1 \leq h \leq t$, if $Q_h \neq \star$, then $Q_h \in C_h$. We say that a $Q$ is \emph{contiguous} if there exists $s \leq t$ such that $Q_{h} \neq \star$ for every $h \geq s+1$ and $Q_{h} = \star$ for every $1 \leq h \leq s$, i.e., the first $s$ entries are $\star$, and the remaining entries are non-$\star$. We note that by definition, $Q_{t+1} \ne \star$ always.

We say that $Q$ is \emph{complete} if $Q$ does not contain any $\star$. We say that $Q' \supseteq Q$ if whenever $Q_h \neq \star$, $Q'_h = Q_h$. We define the size $\abs{Q}$ to be the number of coordinates in $Q$ that do not equal $\star$.

We write $\cH^{(t)}_Q$ to denote the set of all $t$-chains that contain $Q$, and for $u \in [n]$, we write $\cH^{(t)}_{u,Q}$ to denote the set of $t$-chains with head $u$ that contain $Q$.
\end{definition}
We caution the reader that $\cH^{(t)}_u$ and $\cH^{(t)}_Q$,  are different sets of chains. In context, it shall be easy to distinguish between the two cases as the type of $u$ and $Q$ are different: namely, we have $u \in [n]$ and $Q \in \{[n] \cup \star\}^{t+1}$.

\parhead{XOR Formulas from $r$-chains.} Next, we define XOR formulas associated with $\cH^{(r+1)}$ that are guaranteed to be satisfiable. The length of the chain depends on a parameter $r$, which we shall set later.
\begin{definition}[The XOR Formula $\Phi$]
Fix $r \in \N$.

For any $b = (b_1, \dots, b_k) \in \Fits^k$, define the polynomial $\Phi_b$:
\begin{flalign*}
&\Phi_b(x) = \sum_{i = 1}^k b_i \sum_{(i,C_0,w_0,C_1, w_1\ldots, C_r,w_r) \in \cH^{(r+1)}} x_{w_r} \prod_{h=0}^r x_{C_h} = \sum_{i = 1}^k b_i \sum_{C_0, w_0 : C_0 \cup \{w_0\} \in H_i} x_{C_0} \sum_{(w_0, C_1,w_1,\ldots, C_r,w_r) \in \cH^{(r)}} x_{w_r}\prod_{h=1}^r x_{C_h}
\enspace.
\end{flalign*}
We will drop the subscript $b$ when it is clear from the context. 
\end{definition}
We note the equality holds above as we are simply thinking of the chain $(i,C_0,w_0,\ldots, C_r,w_r)$ as being split into two parts, the $1$-chain $(i, C_0, w_0)$, followed by the $r$-chain $(w_0,C_1, w_1, \ldots, C_r,w_r)$. We write the polynomial in this form because for much of the proof, we shall wish to think of the $r$-chain as separate from the $1$-chain $(i, C_0, w_0)$.

We now observe that $\Phi_b(x)$ is satisfiable and thus has a high value. 
\begin{lemma}
\label{lem:val-lb}
For every $b \in \Fits^k$, $\Phi_b$ is satisfied by $x = \Code(b)$ and thus, $\val(\Phi_b) = k (3 \delta n)^{r+1}$.
\end{lemma}

\begin{proof}
Observe that $\Phi_b$ is a sum of monomials corresponding to a $r$-chain each of which is satisfied by $x = \Code(b)$ by \cref{obs:monomial-associated-with-chain}. Thus, $\val(\Phi_b)$ equals the total number of chains of length $r$ with head in $[k]$, which we next count. 

Define $w_{-1} \coloneqq i$. Given $w_{h-1}$ for $h \geq 0$, there are $\delta n$ choices for the set $C_h \cup \{w_h\} \in H_h$ and for each such choice, there are $3$ choices for the next pivot $w_h$. Thus, the number of $(r+1)$-chains with head $i$ is $(3\delta n)^{r+1}$. Summing over the $k$ possible heads completes the proof.
\end{proof}

\section{Contiguously Regular Partition of Chains}
\label{sec:regular-partition}
In this section, we partition the $r$-chain hypergraph $\cH^{(r)}$ into buckets that satisfy a useful regularity property. We first abstract out the relevant properties of the partition below and then show how it can be using a simple greedy partitioning algorithm. This partitioning will be key to setting up and analyzing our spectral refutation in the next section. 

\begin{definition}[Contiguously regular partition]
\label{def:decomp}
For $\delta>0$ and $r \in \N$, let $\cH^{(r)}$ be the $r$-chain hypergraph built form $3$-matchings $H_1, H_2, \ldots, H_n$ on $[n]$ of size $\delta n$ each. Let $\cH^{(r)} = \cup_{Q,p} \cH^{(r)}_{Q,p}$ be a disjoint partition of $\cH^{(r)}$ indexed by $Q \in ([n] \cup \{\star\})^r \times [n]$ and $p \in [m]$ for some large enough $m \in \N$. We say that such a partition is $d$-contiguously regular if the following conditions hold:

\begin{enumerate}[(1)]
\item for every $Q \in ([n] \cup \{\star\})^r \times [n]$ and $p \in [m]$, $\cH^{(r)}_{Q,p} \subseteq \cH^{(r)}_Q$, 
\item for every $(Q,p)$ such that $\cH^{(r)}_{Q,p} \neq \emptyset$, $Q$ is contiguous, 

\item if $|Q|=1$, then $\cH^{(r)}_{Q,p} = \emptyset$ whenever $p>1$,
\item for every contiguous $Q,Q'$ such that $Q' \supseteq Q$, 
\[
\Abs{ \Set{\vec{C'} \in \cH^{|Q'|-1} \mid \vec{C'} \text{ contains } Q', \text{ and } \exists \vec{C} \text{ extending } \vec{C'}, \vec{C} \in \cH^{(r)}_{Q,p} }} \leq d^{|Q'|-1}\mper
\]
\item For every $t$, the set $P_t$ of all $(Q,p)$ such that $\cH^{(r)}_{Q,p} \neq \emptyset$ and $|Q|=t+1$ satisfies $|P_t| d^t \leq n (3 \delta n)^{t}$. Observe that $|P_0|\leq n$ is forced by (3). 
\end{enumerate}
\end{definition}
We now give a bit of intuition for the definition.
A contiguously regular partition takes the set of $r$-chains $\cH^{(r)}$ and decomposes it into pieces, where the pieces are intuitively indexed by $Q$; however, for technical reasons, we will want to have multiple pieces assigned to the \emph{same} $Q$, and so we disambiguate these pieces using the label $p$, i.e., we can have pieces $\cH^{(r)}_{Q,1}$, $\cH^{(r)}_{Q,2}$, $\cH^{(r)}_{Q,3}$, etc.

Condition (1) says that the chains in the piece $\cH^{(r)}_{Q,p}$ in the decomposition are all chains that contain $Q$, hence why we view them as indexed by $Q$. Condition (2) says that the only nonempty pieces have a \emph{contiguous} $Q$, hence the name ``contiguously regular partition''. Condition (3) says that if $\abs{Q} = 1$, then there is only one piece with this $Q$. Recall that when $\abs{Q} \geq 2$, we can have pieces $\cH^{(r)}_{Q,1}$, $\cH^{(r)}_{Q,2}$, $\cH^{(r)}_{Q,3}$, etc.; we have asserted that when $\abs{Q} = 1$, this does not happen. Condition (4) is a \emph{regularity} condition saying that chains in $\cH^{(r)}_{Q,p}$ appear in this piece because the tuple $Q$ is ``maximal''. Condition (5) asserts that the number of pieces with $\abs{Q}$ of a given size is not too large.

We now make the following observation.
\begin{observation}
\label{obs:pieceupperbound}
 Items (1), (2), and (4) imply that $\abs{\cH^{(r)}_{Q,p}}\leq n (3\delta n)^{r-|Q|} d^{|Q|-1}$ for all $(Q,p)$.

Moreover, let $\cH^{(r)}_{u,Q,p}$ denote the set of chains in $\cH^{(r)}_{Q,p}$ with head $u$. Then, Items (1), (2) and (4) imply that $\abs{\cH^{(r)}_{u, Q,p}} \leq (3\delta n)^{r-|Q|} d^{|Q|-1}$ when $\abs{Q} \leq r$.
 \end{observation}
 \begin{proof}
 To see this, we apply item (4) with $Q' = Q$, and we now count the chains in $\cH^{(r)}_{Q,p}$ by (1) first choosing a suffix $\vec{C'} \in \cH^{\abs{Q} - 1}_{\abs{Q}}$, and then (2) completing the chain. By Item (4), we have at most $d^{\abs{Q} - 1}$ choices for the suffix. Once the suffix is fixed, we now complete the chain as follows. If $\abs{Q} = r + 1$, then we have chosen the entire chain and are done. Otherwise, we do the following. First, we choose $w_0$, which has $n$ choices. Then, we choose $C_1 \cup \{w_1\} \in H_{w_0}$, which has $\delta n$ choices, followed by $w_1 \in C_1 \cup \{w_1\}$, which has $3$ choices. We repeat this until we reach the point in the chain where we have determined $w_{r - \abs{Q}}$. Because we also know the suffix $\vec{C'}$, we have already determined $w_{r + 1 - \abs{Q}}$. Because $H_{w_{r - \abs{Q}}}$ is matching, there is at most one $C_{r + 1 - \abs{Q}}$ such that $C_{r + 1 - \abs{Q}} \cup \{w_{r + 1 - \abs{Q}}\} \in H_{w_{r - \abs{Q}}}$, and so we have determined the entire chain. We have thus made at most $n (3 \delta n)^{r - \abs{Q}} d^{\abs{Q} - 1}$ choices when $\abs{Q} \leq r$, and $d^{r}$ choices if $\abs{Q} = r + 1$. In both cases, this is at most $n (3\delta n)^{r-|Q|} d^{|Q|-1}$.

 Finally, we note that for $\abs{Q} \leq r$, the above argument also bounds $\abs{\cH^{(r)}_{u, Q,p}}$. We simply save a factor of $n$ because $w_0$ must be equal to $u$.
 \end{proof}

We now give an algorithm that, given $\cH^{(r)}$ and $d$, outputs a $d$-contiguously regular partition of $\cH^{(r)}$ using a simple iterative greedy scheme. 

\begin{lemma}[Contiguously regular partition of chains]
\label{lem:decomp}
For $\delta>0$, let $H_1, H_2, \ldots, H_n$ be arbitrary $\delta n$ size $3$-matchings on $[n]$. For $r \in \N$, let $\cH^{(r)}$ be the $r$-chain hypergraph built from $H_1,H_2,\ldots, H_n$. Then, for every $d \in \N$, there exists a $d$-contiguously regular partition $\cH^{(r)} =\cup_{Q,p} \cH^{(r)}_{Q,p}$. 
\end{lemma}

\begin{proof}
The greedy algorithm that computes the decomposition is given below.
\begin{mdframed}
  \begin{algorithm}
    \label{alg:decomp}\mbox{}
    \begin{description}
    \item[Given:]
       An $r$-chain hypergraph $\cH^{(r)}$. 
    \item[Output:]
       A contiguously $d$-regular partition $\cH^{(r)} = \cup_{Q,p} \cH^{(r)}_{Q,p}$. 
     \item[Operation:]\mbox{}
    \begin{enumerate}
    \item \textbf{Initialize: } For $Q = (w)$ for each $w\in [n]$, let $\cH^{(0)}_{Q,1} = \{(w)\}$, i.e., the set of $0$-chains with tail $w$. 
    \item \textbf{Iterative Greedy Fixing:} For $t = 1,\ldots,r$, do: 
        \begin{enumerate}
            \item Initialize $\cH^{(t)}_{(\star,Q),p} = \cup_{u \in [n]} H_u \circ \cH^{(t-1)}_{Q,p}$ for every $Q \in ([n] \cup \{\star\})^{t-1} \times [n]$.
            \item For every $Q' =(u,Q)$ for $Q \in [n]^{t-1} \times [n]$, initialize $p'=1$ and do:
            \begin{enumerate}
                \item Let $\cR = \cH^{(t)}_{(\star,Q),p} \cap \{ \vec{C} \in \cH^{(t)} \mid \vec{C} \text{ contains } Q' \}$. If $|\cR| \leq d^{|Q'|-1}$, end.
                \item Otherwise, select exactly $d^{|Q'|-1}$ $t$-chains from $\cR$, remove them from $\cH^{(t)}_{(\star, Q), p}$, and put them in a new piece $\cH^{(t)}_{Q',p'}$.
                \item Set $p' = p'+1$.
            \end{enumerate}
        \end{enumerate}
      \end{enumerate}
    \end{description}
  \end{algorithm}
  \end{mdframed}
We now verify that our decomposition satisfies the properties required of a contiguously $d$-regular partition. The key observation is that the algorithm iterates over $t = 1, \dots, r$, and computes, after the $t$-th iteration, a $d$-contiguously regular partition of $\cH^{(t)}$. Indeed, we prove this by induction. For $t = 0$ this trivially holds. 

We now show the induction step. We observe that properties (1) and (2) are trivial. Property (3) holds because of the following. We observe that the pieces in the decomposition of $\cH^{(t)}$ are either obtained by extending ``old'' pieces to get $\cH^{(t)}_{(\star, Q), p} = \cup_{u \in [n]} H_u \circ \cH^{t-1}_{Q, p}$, or by adding ``new'' pieces produced in step (2b). We note that we only produce new pieces for $Q$ with $\abs{Q} \geq 2$, so we cannot violate property (3). Property (4) follows because the loop in step (2bi) finished. 

Finally, to check property (5), we need to bound $\abs{P_{t'}}$ for $0 \leq t' \leq t$. As $\abs{P_{0}} \leq n$ always holds, it remains to bound $\abs{P_{t'}}$ for $t' \leq t$. We note that all the ``new'' pieces have a $Q$ where $\abs{Q} = t+1$. Hence, for $t' \leq t-1$, $\abs{P_{t'}}$ satisfies property (5) by the induction hypothesis. To bound $\abs{P_{t}}$, we observe that each new partition contains $d^{\abs{Q'}} = d^{t}$ chains. As $\cH^{(t)}$ has at most $n (3 \delta n)^t$ chains (see, e.g., the proof of \cref{lem:val-lb}), the bound on $\abs{P_t}$ follows.
 \end{proof} 

Every $d$-contiguously regular partition of $\cH^{(r)}$ naturally relates to a ``bipartite'' polynomial $\Psi$ (i.e., $\Psi$ has additional variables $y_{Q,p}$ corresponding to labels of the buckets in the partition in addition to the original variables $x \in \Fits^n$) such that $\val(\Psi)$ upper bounds $\val(\Phi)$. Our main technical argument will construct a spectral refutation to upper bound  $\val(\Psi)$ for a $d$-contiguously regular partition of $\cH^{(r)}$ for an appropriate choice of $d$. 

\begin{definition}[Bipartite XOR Formulas from a contiguously regular partition] \label{def:bip-Psi}
Fix $r,d \in \N$ and for the $r$-chain hypergraph $\cH^{(r)}$ built from $3$-matchings $H_1, H_2,\ldots, H_n$ on $[n]$ of size $\delta n$ each, let $\cH^{(r)} = \cup_{p \in P, Q \in ([n] \cup \{\star\})^r \times [n]} \cH^{(r)}_{Q,p}$ be a contiguously $d$-regular partition. 
For each nontrivial piece $\cH^{(r)}_{Q,p}$, we define $\Psi_{i,Q,p}$ as the following XOR formula with terms corresponding $(r+1)$-chains obtained by (1) taking $r$-chains from a single piece $\cH^{(r)}_{Q,p}$ with $x_Q$ ``modded out'' from the corresponding monomial and (2) \emph{joining} with a $1$-chain $(i, C_0, w_0)$. Namely,
\begin{flalign*}
\Psi_{i,Q,p}(x) =  \sum_{C_0, w_0 : C_0 \cup \{w_0\} \in H_i} \sum_{(w_0, C_1, w_1, C_2, w_2, \dots, C_r, w_r) \in \cH_{Q,p}^{(r)}} x_{C_0} x_{w_r\setminus Q_{r+1}} \prod_{h = 1}^r x_{C_h \setminus Q_h} \enspace.
\end{flalign*}
Here, we use the convention that if $Q_h = \star$, then $C_h \setminus Q_h \coloneqq C_h$. We note that because $w_r = Q_{r+1}$, we have $x_{w_r \setminus Q_{r+1}} = 1$.

For each $0 \leq t \leq r$,  let $\Psi^{(t)}(x,y) = \sum_{i=1}^k \sum_{(Q,p) \in P_t} b_i y_{Q,p}  \Psi_{i,Q,p}(x)$. Finally, we let $\Psi(x,y) = \sum_{0 \leq t \leq r} \Psi^{(t)}(x,y)$; here, for every piece $\cH^{(r)}_{Q,p}$ in the contiguously regular partition, we introduce a new variable $y_{Q,p}$.
\end{definition}

We next observe that $\Psi$ is satisfiable and thus has a large value for every $b \in \Fits^k$. Indeed, the observation is that we have replaced the monomial $x_{Q}$ in $\Phi$ with a new variable $y_{Q,p}$ for each $(Q,p)$.
\begin{lemma}
\label{lem:val-lb-labeled}
Let $\cH^{(r)} = \cup_{p \in P, Q} \cH^{(r)}_{Q,p}$ be a contiguously $d$-regular partition. Fix $b \in \Fits^k$ and $x \in \Fits^n$. Then, there is a $y$ such that $\Psi(x,y) = \Phi(x)$. In particular, setting $x = \Code(b)$, we have that $\val(\Psi(x,y)) \geq \val(\Phi(x)) \geq k (3 \delta n)^{r+1}$.
\end{lemma}
We note that the system of equations in $\Phi(x)$ is satisfiable, and so $\val(\Phi)$ is simply the number of constraints in the instance.
\begin{proof}
Set $y_{Q,p} = x_Q$ for every $(Q,p)$, where $x_Q \coloneqq \prod_{h : Q_h \ne \star} x_{Q_h}$. 
\end{proof}

For intuition, we observe that for random matchings $H_1, \dots, H_n$, the $O(\log n)$-contiguously regular partition is the trivial one.
\begin{lemma}[Trivial partition is regular for random matchings ] \label{lem:contiguous-regularity-is-pseudorandom}
Let $H_1, H_2, \ldots, H_n$ be uniformly random and independent $3$-matchings on $[n]$. Then, the trivial partition of the associated $r$-chain hypergraph $\cH^{(r)}$, where we set $\cH^{(r)}_{Q,p}$ to be the set of all chains with tail $w$ if $Q = (\star, \dots, \star, w)$ and $p = 1$, and empty otherwise, is $O(\log n)$-regular with probability at least $1-1/n$ over the draw of $H_i$'s. 
\end{lemma}
\begin{remark}
Eventually (in \cref{lem:rowpruning,sec:step5,sec:step6}), we will set the parameter $d$ to be constant. However, if the matchings $H_1, \dots, H_n$ are random, then with high probability the trivial partition will \emph{not} be $O(1)$-regular. However, if we run \cref{alg:decomp} to decompose the $r$-chains, then with high probability over the draw of $H_1, \dots, H_n$, only a $o(1)$-fraction of the $r$-chains will be placed in a ``non-trivial component'' of the decomposition, i.e., in a piece $\cH^{(r)}_{Q,p}$ where $\abs{Q} \geq 2$.
Phrased differently, if we discard a $o(1)$-fraction of hyperedges from the random matchings, then the trivial partition of the $r$-chain hypergraph of the remaining hyperedges will be $O(1)$-regular. 
This fact is somewhat analogous to the fact that sparse random graphs are not, e.g., triangle-free with high probability, but can be made triangle-free by removing a very small number of edges.
\end{remark}

\begin{proof}
We claim that the trivial partition of $\cH^{(r)}$ is $d$-regular for $d = O(\log n)$ with probability at least $1-1/n$. In the trivial refinement, as defined above, we partition $\cH^{(r)}$ by simply placing a chain in $\cH^{(r)}_{Q,p}$ if $Q = (\star, \dots, \star, w)$, $p = 1$, and the tail of the chain is $Q_{r+1} = w$.

Towards this, we first prove that for every pair ${u,v} \subseteq [n]$, the number of hyperedges in the multiset $\cup_{u \in [n]} H_i$ that contain $u$ and $v$ is at most $O(\log n)$. To see this, observe that the chance that there are some $u,v$ that co-occur in a hyperedge in at least $b$ different $H_i$'s is at most $n^2 {n \choose b} (3/n)^b \leq 3^b/b! \leq 1/n$ if $b = c \log_2 n$ for some large enough $c>0$. We will now set $d_2 = 2c \log_2 n$ and confirm $d$-regularity of the trivial refinement. 

Now take any contiguous $Q \in ([n] \cup \{\star\})^{r} \times [n]$ of size $|Q| = t+1$. Consider the chains $\vec{C'} = (w_0,C_1,w_1,\ldots, C_t,w_t) \in \cH^{(t)}$ that contain $Q$. We now iteratively choose \\
$(w_{t-1},C_t,w_t), (w_{t-2},C_{t-1},w_{t-1}), \ldots, (w_0,C_1,w_1)$. Assuming we have made the first $h$ choices in the list, we have determined $w_{t-h-1}$. There are at most $c \log_2 n$ choices for a hyperedge in any of $H_i$s that contains $Q_{t-h-1}$ and $w_{t-h-1}$ and given this choice, at most $2$ choices for $w_{t-h}$. So in total, we have at most $(2c \log_2 n)^{t}=d^t$ choices. 
\end{proof}

\section{Spectral Refutation via Kikuchi Matrices}
\label{sec:kikuchimethod}
In \cref{sec:regular-partition}, we defined polynomials $\Psi^{(t)}(x,y)$ such that $\E_b[\val(\Phi)] \leq \sum_{t = 0}^r \E_b[\val(\Psi^{(t)})]$. Thus, to prove \cref{mthm:main}, we need to upper bound $\E_b[\val(\Psi^{(t)})]$ for each $t$. In this section, we will define, for each $0 \leq t \leq r$, a Kikuchi matrix $A^{(t)}$ such that $\E_b[\val(\Psi^{(t)})^2] \leq \boolnorm{A^{(t)}}$. Then, in \cref{sec:rowpruning} we shall bound $\boolnorm{A^{(t)}}$ and finish the proof.

\subsection{Step 1: the Cauchy--Schwarz trick} First, we show that we can relate $\Psi^{(t)}(x,y)$ to a certain ``cross-term'' polynomial $f_M$ obtained via applying the Cauchy--Schwarz inequality.

\begin{lemma}[Cauchy--Schwarz trick] \label{lem:cauchy-schwarz}
Let $M$ be a maximum directed matching\footnote{A directed matching is a matching, only the edges are additionally directed}\footnote{This is a perfect matching if $k$ is even, and will leave one element of $[k]$ unmatched if $k$ is odd.} of $[k]$ and let $f_M$ be the cross-term polynomial defined as
\begin{equation*}
f_M(x) = f_M^{(t)} = \sum_{\{i ,  j\} \in M} b_i b_j \sum_{(Q,p) \in P_t} \Psi_{i,Q,p}(x) \Psi_{j,Q,p} (x) \mper
\end{equation*}
Then,
\begin{equation*}
\E_{b \gets \Fits^k} \val(\Psi^{(t)}) \leq  k  \Paren{|P_t| (3 \delta n)^{r- t} d^{t}}^2  + |P_t| 2 k \E_{M} \E_{b \gets \Fits^k} [\val(f_M^{(t)})] \mcom
\end{equation*}
where the expectation $\E_{M}$ is over a uniformly random maximum matching $M$.
\end{lemma}
\begin{proof}
We will first apply the Cauchy--Schwarz inequality to eliminate the $y$ variables:
\begin{align*}
\Psi^{(t)}(x,y)^2 &= \Paren{\sum_{(Q,p) \in P_t} y_{Q,p} \Paren{\sum_{i=1}^k b_i \Psi_{i,Q,p}}}^2\\
&\leq \Paren{\sum_{(Q,p) \in P_t} y_{Q,p}^2} \Paren{\sum_{(Q,p) \in P_t} \Paren{\sum_{i = 1}^k b_i \Psi_{i,Q,p}}^2}\\
&= |P_t| \Paren{\sum_{(Q,p) \in P_t} \sum_{i =1}^k \Psi^{2}_{i,Q,p} + \sum_{(Q,p) \in P_t} \sum_{i \neq j \in [k]} b_i b_j \Psi_{i,Q,p} (x) \Psi_{j,Q,p}}\enspace.
\end{align*}

Observe that $\abs{\Psi_{i,Q,p}(x)}$ is at most the number of $(w_0, C_1, w_1, \dots, C_r, w_r) \in \cH^{(r)}_{Q,p}$ and $C_0 \in {[n] \choose 2}$ such that $C_ \cup \{w_0\} \in H_i$. If $\abs{Q} = r + 1$, i.e., $t = r$, then we observe that by \cref{obs:pieceupperbound}, we have $\abs{\cH^{(r)}_{Q,p}} \leq d^{\abs{Q} - 1}$, and for each choice of $(w_0, C_1, w_1, \dots, C_r, w_r) \in \cH^{(r)}_{Q,p}$, there is at most one choice of $C_0$. If $\abs{Q} \leq r$, then we have at most $(3 \delta n)$ choices for $(C_0, w_0)$, and for each $w_0$, we have by \cref{obs:pieceupperbound}  that $\abs{\cH^{(r)}_{w_0, Q, p}} \leq (3\delta n)^{r-|Q|} d^{|Q|-1}$, giving us $(3\delta n)^{r-|Q| + 1} d^{|Q|-1}$ choices in total. We thus have that $\abs{\Psi_{i,Q,p}(x)} \leq (3\delta n)^{r-|Q| + 1} d^{|Q|-1}$, regardless of $\abs{Q} = t + 1$.

Thus, for $|Q|=t+1$, $\sum_{Q,p} \sum_{i =1}^k \Psi^{2}_{i,Q,p} \leq k |P_t| \Paren{(3\delta n)^{r-t} d^{t}}^2$. This gives us an upper bound of $k |P_t| \Paren{(3\delta n)^{r-t} d^{t}}^2$ on the first term. 

Let's now analyze the second term. Since a uniformly random maximum matching on $[k]$ contains a (directed) edge $(i,j)$ with probability exactly $\frac{1}{2(k-1)}$ if $k$ is even, and $\frac{1}{2k}$, if $k$ is odd, we have:
\begin{align*}
\sum_{i \neq  j \in [k]} b_i b_j \sum_{(Q,p) \in P_t} \Psi_{i,Q,p}(x) \Psi_{j,Q,p} (x) 
\leq 2k \E_M \sum_{(i,j) \in M} b_i b_j \sum_{(Q,p) \in P_t} \Psi_{i,Q,p}(x) \Psi_{j,Q,p} (x) = 2k \E_M [f_M^{(t)}] \mper
\end{align*} 
Using that $\val(\E_M[ f_M^{(t)}])\leq \E_M [ \val(f_M^{(t)})]$ completes the proof.
\end{proof}

\subsection{Step 2: defining the Kikuchi matrices} It thus remains to bound $\val(f_M)$ for an arbitrary directed maximum matching $M$. 

For $i \in [k]$ and $(Q,p)$, we let $\cH^{(r+1)}_{i, Q, p}$ denote the set of chains in $\cH^{(r+1)}$ of the form $(i, C_0, w_0, C_1, w_1, \dots, C_r, w_r)$ where $(w_0, C_1, w_1, \dots, C_r, w_r) \in \cH^{(r)}_{Q,p}$. We define the Kikuchi matrices that we consider below.
\begin{definition}[Kikuchi matrices for a fixed $t$]
\label{def:kikuchi}
Let $i,j \in [k]$.

Let $\vec{C} = (i,C_0, w_0, C_1, w_1, C_2, w_2, \dots, C_r, w_r) \in \cH^{(r+1)}_{i,Q,z}$ and $\vec{C'} = (j,C'_0, w'_0, C'_1, w'_1, C'_2, w'_2, \dots, C'_r, w'_r) \in \cH^{(r+1)}_{j,Q,z}$. For $Q$ of size $|Q|=t+1$, we let $A_{i,j,Q,p}^{(\vec{C}, \vec{C'})} \in \Bits^{{{[n]} \choose \ell}^{2r+2-t}}$ be the matrix with rows and columns by indexed by $(2r+2-t)$-tuples of sets $(S_0, \dots, S_{r-t}, S'_0, \dots, S'_{r-t}, R_1, \dots, R_{t})$  of size exactly $\ell$. Note that when $t=0$, we do not have any ``$R_h$'s'' in the row/column index tuples. 

We set $A_{i,j,Q,p}^{(\vec{C}, \vec{C'})}((S_0, \dots, S_{r-t}, S'_0, \dots, S'_{r-t}, R_1, \dots, R_{t}), (T_0,  \dots, T_{r-t}, T'_0, \dots T'_{r-t}, W_1, \dots, W_{t}))$ equal to $1$ if the following holds, and otherwise we set this entry to be $0$.
\begin{enumerate}
\item For $h = 0, \dots, r - t$, we have $S_h \oplus T_h = C_h$,
\item For $h = 0, \dots, r - t$, we have $S'_h \oplus T'_h = C'_h$,
\item For $h = 1, \dots, t$, we have $R_h = \{u\} \cup U$, $W_h = \{v\} \cup U$, where $C_{r - t + h} = \{u, Q_h\}$, $C'_{r - t + h} = \{v, Q_h\}$, and $U \subseteq [n]$ is a set of size $\ell - 1$ with $u,v \notin U$.
\end{enumerate}
We let $A_{i,j} = \sum_{(Q,p) \in P_t} \sum_{\vec{C} \in \in \cH^{(r+1)}_{i,Q,p} \vec{C'} \in \cH^{(r+1)}_{j,Q,p}} A_{i,j,Q,p}^{(\vec{C}, \vec{C'})}$, and for any matching $M$ on $[k]$, let $ A_M = \sum_{(i,j) \in M} b_i b_j A_{i,j}$.
\end{definition}

\subsection{Step 3: relating the ``Cauchy--Schwarzed'' polynomial $f_M$ and the Kikuchi matrix $A$}
The following lemma shows that we can express $f_M(x)$ as a (scaling of a) quadratic form on the matrix $A$.
\begin{lemma}
\label{lem:completeness}
Let $x \in \Fits^n$, and let $x' \in \Fits^{N}$, where $N = {n \choose \ell}^{2r + 2 - t}$, denote the vector where the $(S_0, S_1, \dots, S_{r-t}, S'_0, S'_1, \dots, S'_{r-t}, R_1, \dots, R_{t})$-th entry of $x'$ is $\prod_{h = 0}^{r-t} x_{S_h} x_{S'_h} \prod_{h = 1}^{t} x_{R_h}$. Then, ${x'}^{\top} A x' = D f_M(x)$, where $D = 2^{2r + 2 - 2t} {n - 2 \choose \ell - 1}^{2r + 2 - t}$. 
Note that $D/N = 2^{2r + 2 - 2t} \left(\frac{\ell(n - \ell)}{n(n-1)}\right)^{2r + 2- t}$. 
In particular, $\val(f_M) \leq \frac{1}{D} \boolnorm{A}$. 
\end{lemma}
\begin{proof}
Expanding definitions, we have
\begin{flalign*}
&x'^{\top} A x' = \sum_{(i,j) \in M} b_i b_j \sum_{(Q,p) \in P_t} \sum_{\vec{C} \in \cH^{(r+1)}_{i,Q,p}, \vec{C'} \in \cH^{(r+1)}_{j,Q,p}} x'^{\top} A_{i,j,Q,p}^{(\vec{C}, \vec{C'})} x' \enspace,\\
&f_M(x) = \sum_{(i,j) \in M}b_i b_j  \sum_{(Q,p) \in P_t}  \Psi_{i,Q,p}(x) \Psi_{j,Q,p}(x) \enspace,
\end{flalign*}
where we recall that
\begin{flalign*}
&\Psi_{i,Q,p}(x) = \sum_{C_0, w_0 : C_0 \cup \{w_0\} \in H_i} x_{C_0} \sum_{(w_0, C_1, w_1, C_2, w_2, \dots, C_r, w_r) \in \cH^{(r)}_{Q,p}} \prod_{h = 1}^{r - t} x_{C_h} \prod_{h = 1}^{t} x_{C_{r - t + h} \setminus Q_h} \\
&= \sum_{\vec{C} \in \cH^{(r+1)}_{i,Q,p}}  \prod_{h = 0}^{r - t} x_{C_h} \prod_{h = 1}^{t} x_{C_{r - t + h} \setminus Q_h} 
\end{flalign*}
Thus, it suffices to show that
\begin{flalign*}
x'^{\top} A_{i,j,Q,p}^{(\vec{C}, \vec{C'})} x' &= D \cdot  \prod_{h = 0}^{r - t} x_{C_h} \prod_{h = 1}^{t} x_{C_{r - t + h}\setminus Q_h} \prod_{h = 0}^{r - t} x_{C'_h} \prod_{h = 1}^{t} x_{C'_{r - t + h}\setminus Q_h} \enspace.
\end{flalign*}
Let  $\vec{S} = (S_0, S_1, \dots, S_{r-t}, S'_0, S'_1, \dots, S'_{r-t},R_1, \dots, R_t)$ and $\vec{T} = (T_0,  \dots, T_{r-t}, T'_0, \dots T'_{r-t},W_1, \dots, W_{t})$ be such that $A_{i,j,Q,p}^{(\vec{C},\vec{C'})}(\vec{S},\vec{T}) = 1$. Then, we have that
\begin{flalign*}
x'_{\vec{S}} x'_{\vec{T}} &= \prod_{h = 0}^{r-t} x_{S_h} x_{T_h} x_{S'_h} x_{T'_h} \prod_{h = 1}^{t} x_{R_h} x_{W_h} =  \prod_{h = 0}^{r-t} x_{S_h \oplus T_h} x_{S'_h \oplus T'_h} \prod_{h = 1}^{t} x_{R_h \oplus W_h} \\
&= \prod_{h = 0}^{r-t} x_{C_h} x_{C'_h} \prod_{h = 1}^{t} x_{C_{r - t+h} \setminus Q_h} x_{C'_{r - t+h} \setminus Q_h} \enspace,
\end{flalign*}
where we use that this entry of $A_{i,j,Q,p}^{(\vec{C},\vec{C'})}(\vec{S},\vec{T})$ is nonzero and $x \in \Fits^n$ so that $x_u^2 = 1$ for any $u \in [n]$.

Next, we prove that there are exactly $D$ pairs $(\vec{S},\vec{T})$ where $A_{i,j,Q,p}^{(\vec{C},\vec{C'})}(\vec{S},\vec{T})$ is nonzero. We observe that, for each $h = 0, \dots, r - t+1$, there are exactly $2 {n - 2 \choose \ell - 1}$ pairs $(S_h, T_h)$ such that $S_h \oplus T_h = C_h$. Indeed, this is because $C_h = \{u,v\}$ has size exactly $2$, so $S_h \oplus T_h = C_h$ implies that $S_h = \{u\} \cup U$ and $T_h = \{v\} \cup U$. There are $2$ choices for $u \in C_h$ to assign to $S_h$, and then afterward there are ${n - 2 \choose \ell - 1}$ choices for the set $U$, which is a set of size $\ell - 1$ not containing either of $u,v$. For $h = 1, \dots, t$, there are exactly ${ n - 2 \choose \ell - 1}$ pairs $(R_h, W_h)$ satisfying the condition, as this is the number of choices for $U$. Here, note that we do not have the additional factor of $2$ because the $u$ coming from $C_{r - t + h}$ must be in $R_h$. Combining, we see that $D = 2^{2(r - t + 1)} {n - 2 \choose \ell - 1}^{2(r - t + 1) + t}$, as required. The ``in particular'' follows by noting $\val(f_M) = \max_{x \in \Fits^n} f_M(x)$ and $\boolnorm{A} = \max_{x', y' \in \Fits^{N}} x'^{\top} A y'$.
\end{proof}

\subsection{Step 4: bounding the $\infty \to 1$-norm of $A$ via row pruning}
By \cref{lem:completeness}, in order to upper bound $\E_b[\val(f_M)]$, it suffices to bound $\E_b[\boolnorm{A}]$.

We always have $\boolnorm{A} \leq N \norm{A}_2$. It turns out that $\norm{A}_2$ is governed by the maximum degree (relative to the average) of any of $A_{i,j}$, where by degree we mean the number of nonzero entries in a row/column. However, the $A_{i,j}$'s can have rows of degree significantly larger the the average, and this prohibits the spectral norm from giving a good bound on $\val(f_M)$.

The key observation is that, for a certain choice of our parameters $\ell, r, d$, the fraction of these ``bad rows'' is very small and does not noticeably affect $\boolnorm{A}$. We can thus ``zero out'' these bad rows and then use the spectral certificate on the ``pruned matrix'' to bound $\val(f_M)$. Establishing this combinatorial fact is the crux of our proof and is captured in the following lemma. The next section (\cref{sec:rowpruning}) is dedicated to the proof of this lemma, and constitutes the key component of the proof. In fact, at a high level all the steps done so far in the proof are somewhat generic, and our key innovation is choosing the decomposition step carefully to ensure that this row pruning step succeeds.

\begin{restatable}[Row pruning]{lemma}{RowPruning}
\label{lem:rowpruning}
Fix $r \geq 1$, and let $\cH^{(r)}$ denote the $r$-chain hypergraph. Let $\cup_{Q,p} \cH^{(r)}_{Q,p}$ be a $d$-contiguously regular partition of $\cH^{(r)}$. Fix $0 \leq t \leq r$ and a maximum directed matching $M$ on $[k]$. Let $A$ be the Kikuchi matrix defined in \cref{def:kikuchi}, which depends on $r$, $t$, the pieces $\cup_{(Q,p) \in P_t} \cH^{(r)}_{Q,p}$ of the refinement, and the matching $M$.

Let $\Gamma > 0$ be a constant, and let $\Delta = 9 \cdot 2^{2r + 2 - 2t} (\ell/n)^{2r + 2 - t} d^t (3 \delta n)^{2r + 1 - t}$. Let $\cB$ denote the set of rows $\vec{S}$ such that there exists $i \ne j \in [k]$ where the $\vec{S}$-th row/column of $A_{i,j}$ has more than $\Delta$ nonzero entries. 

Suppose there is $\gamma \in (0,1)$ such that
\begin{enumerate}[(1)]
\item\label{item:smallgamma} $\gamma \leq \frac{1}{c \Gamma r^3 \log_2 n}$, where $c$ is a sufficiently large absolute constant;
\item\label{item:largeell} $\frac{2}{3 \delta \ell}  \leq \gamma$;
\item\label{item:largedensity} $\left(\frac{3 \gamma \delta \ell}{4}\right)^{r + 1}  \geq n$;
\item\label{item:couplingbound}$(2r + 2) \exp(-\ell/64r^2) \leq \ell^{-\Gamma r}$;
\item $d = 3 \delta \ell \gamma$.
\end{enumerate}
Then, the number of bad rows is $\abs{\cB} \leq 2 \ell^{-\Gamma r} N$.
\end{restatable}

\subsection{Step 5: finishing the proof}
\label{sec:step5}
Let $\Gamma$ be a sufficiently large constant, $r = O(\log n)$, $\gamma = 1/O(\log^4 n)$, and $\ell = O(\log^4 n/\delta)$ for a sufficiently large constant. This choice of parameters satisfies all the conditions in \cref{lem:rowpruning} and furthermore they satisfy 
\begin{equation*}
 (2 \ell^{-\Gamma r})  \cdot (2 \ell)^{2r + 2 - t} \leq 1/n^2 \enspace.
\end{equation*}
We note that by our choice of parameters, $d$ is constant.

Applying \cref{lem:decomp}, we can construct a contiguous $d$-regular refinement of $\cH^{(r)}$, given by $\cH^{(r)} = \cup_{Q,p} \cH^{(r)}_{Q,p}$ and polynomials $\Psi^{(t)}$ for $0 \leq t \leq r$ such that $ k (3 \delta n)^{r + 1} \leq \E_b[\val(\Phi)] \leq \sum_{t = 0}^r \E_b[\val(\Psi^{(t)})]$.

By \cref{lem:cauchy-schwarz}, we have
\begin{flalign*}
&(\E_b[\val(\Psi^{(t)})])^2 \leq \E_b[\val(\Psi^{(t)})^2] =  \E_b[\val((\Psi^{(t)})^2)] \leq k  \Paren{|P_t| (3 \delta n)^{r - t} d^{t}}^2  + |P_t| 2 k \E_{M} \E_{b \gets \Fits^k} [\val(f_M^{(t)})] \enspace.
\end{flalign*}

By \cref{lem:completeness}, we have for any maximum directed matching $M$, 
\begin{flalign*}
\E_{b \gets \Fits^k} [\val(f_M^{(t)})] \leq \frac{1}{D} \boolnorm{A} \enspace.
\end{flalign*}

Let $B_{i,j}$ denote the matrix $A_{i,j}$ after we zero out all rows and columns in $\cB$. We observe that $\norm{B_{i,j}}_2 \leq \Delta$ as every row and \emph{column} in $B_{i,j}$ has at most $\Delta$ nonzero entries; the fact about the columns of $B_{i,j}$ follows because $A_{i,j} = A_{j,i}^{\top}$, so the set $\cB$ contains all the bad columns as well. Thus, $\boolnorm{A_{i,j} - B_{i,j}} \leq |\cB| \cdot (2 \ell)^{2r + 2 - t} \leq 2\ell^{-\Gamma r} N \cdot (2 \ell)^{2r + 2 - t} \leq N/n^2$, by our choice of parameters. Here, we used that every row of $A_{i,j}$ can (crudely) have at most $(2\ell)^{2r + 2 - t}$ nonzero entries.

Now, let $B = \sum_{(i,j) \in M} b_i b_j B_{i,j}$. The random matrix $B$ is a Rademacher series with $k$ terms in $\R^{N \times N}$. By the Matrix Khintchine (\cref{fact:matrixkhintchine}) inequality, we have that $\E_{b}[\norm{B}_2] \leq O(\Delta \sqrt{k \log N}) = O(\Delta \sqrt{k r \ell \log n})$. Note that here we use that $M$ is a matching, so $b_i b_j$ and $b_{i'} b_{j'}$ are independent Rademacher random variables for distinct edges $(i,j)$ and $(i',j')$ in the matching. 

Hence, we have that
\begin{flalign*}
D \E_b[\val(f_M)] &\leq \E_b[\boolnorm{A}] \leq \E_b[\boolnorm{B} + \boolnorm{A - B}] \\
&\leq \E_b[N \norm{B}_2] + kN/n^2 \leq N O(\Delta \sqrt{k r \ell \log n}) + o(N) \mper
\end{flalign*}
Thus, $\E_b[\val(f_M)] \leq \frac{N}{D} O(\Delta \sqrt{k r \ell \log n})$.  Using the bound on $\Delta$ from \cref{lem:rowpruning}, we have that
\begin{flalign*}
&\frac{N}{D} \Delta \leq 2^{-2r - 2 + 2t} \left(\frac{n(n-1)}{\ell(n - \ell)}\right)^{2r + 2- t} \cdot 9 \cdot 2^{2r + 2 - 2t} (\ell/n)^{2r + 2 - t} d^t (3 \delta n )^{2r + 1 - t} \\
&= \left(\frac{n-1}{n - \ell}\right)^{2r + 2- t} \cdot 9 \cdot  d^t (3 \delta n )^{2r + 1 - t} \leq e^{\frac{O(\ell r)}{n}} \cdot 9 \cdot  d^t (3 \delta n )^{2r + 1 - t} \leq O(1) \cdot  d^t (3 \delta n )^{2r + 1 - t} \enspace,
\end{flalign*}
and so we conclude that
$\E_b[ \val(f_M)] \leq d^t (3 \delta n)^{2r + 1 - t} O(\sqrt{k r \ell \log n})$, where we use that $\ell r \leq n$.

We thus have
\begin{flalign*}
&\E_b[\val(\Psi^{(t)})] \leq \sqrt{ k  \Paren{|P_t| (3 \delta n)^{r - t} d^{t}}^2  + |P_t| 2 k \cdot d^t (3 \delta n)^{2r + 1 - t} O(\sqrt{k r \ell \log n})} \enspace.
\end{flalign*}
Next, we note that we have $\abs{P_t} d^t \leq \abs{\cH^{(t)}} = n (3 \delta n)^t$, and so
\begin{flalign*}
&\E_b[\val(\Psi^{(t)})] \leq \sqrt{ k  \Paren{n (3 \delta n)^t \cdot (3 \delta n)^{r - t} }^2  + n (3 \delta n)^t \cdot 2 k \cdot  (3 \delta n)^{2r + 1 - t} O(\sqrt{k r \ell \log n})} \\
&\leq O(1) \cdot \sqrt{ k}  \frac{1}{3 \delta}  (3 \delta n)^{r+1} + O(1) \cdot (3 \delta n)^{r + 1} \sqrt{\frac{k}{3 \delta}}  ({k r \ell \log n})^{1/4} \\
&\leq  O(1) \cdot \sqrt{k} \cdot (3 \delta n)^{r+1} \left(  \frac{1}{3 \delta} +  \sqrt{\frac{1}{3 \delta}}  ({k r \ell \log n})^{1/4} \right) \\
&\leq  O(1) \cdot \sqrt{k} \cdot (3 \delta n)^{r+1}  \sqrt{\frac{1}{3 \delta}}  ({k r \ell \log n})^{1/4} \enspace,
\end{flalign*}
assuming that $\delta^{-2} \leq O(k r \ell \log n) = O(k \log^{6} n/\delta)$.\footnote{When we optimize the $\log n$ factor in the next step, we will no longer need this assumption, which is why it does not appear in \cref{mthm:main}.}
Thus,
\begin{flalign*}
 &k (3 \delta n)^{r + 1} \leq \E_b[\val(\Phi)] \leq \sum_{t = 0}^r \E_b[\val(\Psi^{(t)})] \leq (r+1) \cdot O(1) \cdot \sqrt{k} \cdot (3 \delta n)^{r+1}  \sqrt{\frac{1}{3 \delta}}  ({k r \ell \log n})^{1/4} \\
 &\implies k \leq \frac{1}{9 \delta^2} \cdot O( r^5  \ell \log n) = O( \log^{10} n/\delta^3) \enspace,
\end{flalign*}
by our choice of $r, \ell$. 
We note that, up to the proof of \cref{lem:rowpruning}, this \emph{almost} finishes the proof of \cref{mthm:main}. The issue is that we have lost an additional $\log^2 n$-factor. In the next and final step, we shall save this factor by reformulating the above proof as a reduction to a $2$-LDC and applying a off-the-shelf bound on linear $2$-LDCs instead of a spectral refutation to finish.

\subsection{Step 6: optimizing the $\log n$ factor}
\label{sec:step6}
We shall now reformulate the arguments in \cref{sec:step5} to give us a reduction from the $3$-LCC $\Code$ to a $2$-LDC $\Code'$. 
Instead of bounding $\val(\Psi^{(t)})$ using the $\infty\to1$ norm of the Kikuchi matrices, we shall instead use the Kikuchi matrices to give a \emph{reduction} to a linear $2$-LDC, and then we apply the lower bound of~\cite{GoldreichKST06} (\cref{fact:2ldclb}). The difference between \cref{sec:step5} and this subsection is similar to the difference between the main proof and the proof in Appendix B in~\cite{AlrabiahGKM23}, which also saves some additional $\log n$ factors in the setting of $3$-LDC lower bounds.

The reason for the savings is that, in the case of $2$-query linear codes, \cref{fact:2ldclb} shows a lower bound of $2 \log_2 n \geq \delta k$, which saves a factor of $\delta$ over the lower bound from spectral refutation of $O(\log n) \geq \delta^2 k$ for general codes. In our reduction, we shall produce a $2$-LDC with $\delta' \sim \delta/(\log^2 n)$, so this optimization saves us a $O(\log^2 n)$ factor. As a result, we get a final lower bound of $k \leq O(\log^{8} n)$, as opposed to the lower bound of $k \leq O(\log^{10} n)$ that we obtained in \cref{sec:step5}.

We proceed similarly to~\cref{sec:step5}.
Let $\Gamma$ be a sufficiently large constant, $r = O(\log n)$, $\gamma = 1/O(\log^4 n)$, and $\ell = O(\log^4 n/\delta)$ for a sufficiently large constant. We note that this choice of parameters satisfies all the conditions in \cref{lem:rowpruning}, and furthermore they satisfy 
\begin{equation*}
 (2 \ell^{-\Gamma r})  \cdot (2 \ell)^{2r + 2 - t} \leq 1/n^2 \enspace.
\end{equation*}

Applying \cref{lem:decomp}, we can construct a contiguous $d$-regular refinement of $\cH^{(r)}$, given by $\cH^{(r)} = \cup_{Q,p} \cH^{(r)}_{Q,p}$ and polynomials $\Psi^{(t)}$ for $0 \leq t \leq r$ such that $ k (3 \delta n)^{r + 1} \leq \E_b[\val(\Phi)] \leq \sum_{t = 0}^r \E_b[\val(\Psi^{(t)})]$. 

Now, we observe that there exists $t \in \{0, \dots, r\}$ such that $k (3 \delta n)^{r + 1}/(r+1) \leq \E_b[\val(\Psi^{(t)})]$. In particular, $\Psi^{(t)}$ has at least $k (3 \delta n)^{r + 1}/(r+1)$ constraints.
For the remainder of the proof, we let $t$ be this particular value in $\{0, \dots, r\}$.

By \cref{lem:cauchy-schwarz}, we have
\begin{flalign*}
&(\E_b[\val(\Psi^{(t)})])^2 \leq \E_b[\val(\Psi^{(t)})^2] =  \E_b[\val((\Psi^{(t)})^2)] \leq k  \Paren{|P_t| (3 \delta n)^{r - t} d^{t}}^2  + |P_t| 2 k \E_{M} \E_{b \gets \Fits^k} [\val(f_M^{(t)})] \enspace.
\end{flalign*}
Therefore, there exists a maximum directed matching $M$ on $[k]$ such that 
\begin{equation*}
\frac{1}{k \abs{P_t}} (\E_b[\val(\Psi^{(t)})])^2 -  \abs{P_t} \Paren{ (3 \delta n)^{r - t} d^{t}}^2 \leq 2\E_{b \gets \Fits^k} [\val(f_M^{(t)})] \enspace.
\end{equation*}
For the remainder of the proof, we let $M$ be this particular directed matching.

Let $L = \{ i : (i,j) \in M\}$ denote the ``left halves'' of the edges in the matching $M$. We note that $k' \coloneqq \abs{L} \geq \frac{k-1}{2}$.
Let $\Code' \colon \Fits^L \to \Fits^{2N}$, where $N \coloneqq {n \choose \ell}^{2r + 2 - t}$, be the linear code defined from $\Code$ as follows. For each $b \in \Fits^L$, we first extend $b$ to be in $\Fits^k$ by setting $b_j = 1$ for all $j \notin L$ (for $b \in \Fits^L$, we shall abuse notation and think of $b$ as in $\Fits^k$ using this trivial extension). Then, we let $x = \Code(b)$, and finally we let $x' \coloneqq \Code'(b)$ be the vector with $2N$ coordinates, one for each row/column of $A$, where the $\vec{S}$-th entry (similarly $\vec{T}$-th entry) is given by $x'_{\vec{S}} = \prod_{h = 0}^{r-t} x_{S_h} x_{S'_h} \prod_{h = 1}^{t} x_{R_h}$.

We make the following observations. First, we note that $\Code'$ is clearly a linear map. Secondly, following \cref{lem:completeness}, we note that for every $b \in \Fits^L$, every $(i,j) \in M$ (which implies that $i \in L$ and $j \notin L$), and row $\vec{S}$ and column $\vec{T}$ where $A_{i,j}(\vec{S}, \vec{T}) = 1$, we have that $x' = \Code'(b)$ satisfies $x'_{\vec{S}} x'_{\vec{T}} = b_i$.

We now show that $\Code'$ is a $(2, \delta)$-LDC for $\delta' = \Omega(\delta/r^2)$. Formally, we shall show that for each $(i,j) \in M$, there exists a matching $G''_{i,j}$ on $[2N]$ such that for every $b \in \Fits^L$, each edge $(\vec{S}, \vec{T})$ in $G''_{i,j}$, we have $x'_{\vec{S}} x'_{\vec{T}} = b_i$ where $x' = \Code'(b)$, and furthermore $\frac{1}{k'} \sum_{(i,j) \in M} \abs{G''_{i,j}} \geq \delta' \cdot 2N$, where we recall that $k' \coloneqq \abs{L}$ is the dimension of $\Code'$.

We have already argued that for every edge $(\vec{S}, \vec{T})$ in the bipartite graph $G_{i,j}$ defined by the adjacency matrix $A_{i,j}$ (where the rows and columns form the left and right sets of vertices), we have $x'_{\vec{S}} x'_{\vec{T}} = b_i$. It thus remains to show that $G_{i,j}$ has a matching of size $\delta' N$.

As before, let $B_{i,j}$ denote the matrix $A_{i,j}$ after we zero out all rows and columns in $\cB$. We observe that every row and \emph{column} in $B_{i,j}$ has at most $\Delta$ nonzero entries; the fact about the columns of $B_{i,j}$ follows because $A_{i,j} = A_{j,i}^{\top}$, so the set $\cB$ contains all the bad columns as well. Thus, the bipartite graph $G'_{i,j}$ defined by the adjacency matrix $B_{i,j}$ has maximum (left or right) degree at most $\Delta$ and therefore has a matching $G''_{i,j}$ of size at least $\abs{G'_{i,j}}/\Delta$.

Now, the number of edges removed is at most $|\cB| \cdot (2 \ell)^{2r + 2 - t} \leq 2\ell^{-\Gamma r} N \cdot (2 \ell)^{2r + 2 - t} \leq N/n^2$, by our choice of parameters. Here, we used that every row of $A_{i,j}$ can (crudely) have at most $(2\ell)^{2r + 2 - t}$ nonzero entries. Thus, we have that 
$\abs{E(G'_{i,j})} \geq \abs{E(G_{i,j})} - N/n^2$.

In order to finish the reduction, we need to lower bound $\frac{1}{N} \sum_{(i,j) \in M} \abs{E(G''_{i,j})}$. We have that
\begin{flalign*}
&\frac{1}{N} \sum_{(i,j) \in M} \abs{E(G''_{i,j})} \geq \sum_{(i,j \in M)} \frac{1}{N \Delta}\abs{E(G'_{i,j})} \geq \sum_{(i,j \in M)} \frac{1}{N \Delta}\left(\abs{E(G_{i,j})} - \frac{N}{n^2}\right) \geq  \frac{1}{N \Delta}\left(D \cdot \E_{b \gets \Fits^k} [\val(f_M^{(t)})] - \frac{k N}{n^2}\right) \\
&\geq \frac{D}{2N\Delta}\left(\frac{1}{k \abs{P_t}} (\E_b[\val(\Psi^{(t)})])^2 -  \abs{P_t} \Paren{ (3 \delta n)^{r - t} d^{t}}^2 \right) - \frac{k}{\Delta n^2} \\
&\geq \frac{D}{2N \Delta}\left(\frac{1}{k \abs{P_t}} \cdot\frac{k^2 (3 \delta n)^{2r + 2}}{(r+1)^2} -  \abs{P_t}\Paren{ (3 \delta n)^{r - t} d^{t}}^2 \right) - \frac{k}{\Delta n^2} \enspace.
\end{flalign*}
Using the bound on $\Delta$ from \cref{lem:rowpruning}, we have that
\begin{flalign*}
&\frac{N \Delta}{D} \leq 2^{-2r - 2 + 2t} \left(\frac{n(n-1)}{\ell(n - \ell)}\right)^{2r + 2- t} \cdot 9 \cdot 2^{2r + 2 - 2t} (\ell/n)^{2r + 2 - t} d^t (3 \delta n )^{2r + 1 - t} \\
&= \left(\frac{n-1}{n - \ell}\right)^{2r + 2- t} \cdot 9 \cdot  d^t (3 \delta n )^{2r + 1 - t} \leq e^{\frac{O(\ell r)}{n}} \cdot 9 \cdot  d^t (3 \delta n )^{2r + 1 - t} \leq O(1) \cdot  d^t (3 \delta n )^{2r + 1 - t} \enspace,
\end{flalign*}
which implies that
\begin{flalign*}
&\frac{1}{N} \sum_{(i,j) \in M} \abs{E(G''_{i,j})} \geq \frac{1}{O(1) d^t (3 \delta n )^{2r + 1 - t}} \left(\frac{1}{k \abs{P_t}} \cdot\frac{k^2 (3 \delta n)^{2r + 2}}{(r+1)^2} -  \abs{P_t}\Paren{ (3 \delta n)^{r - t} d^{t}}^2 \right) - \frac{k}{\Delta n^2} \enspace.
\end{flalign*}
Next, we note that we have $\abs{P_t} d^t \leq \abs{\cH^{(t)}} = n (3 \delta n)^t$, and so
\begin{flalign*}
&\frac{1}{N} \sum_{(i,j) \in M} \abs{E(G''_{i,j})} \geq \frac{1}{O(1) d^t (3 \delta n )^{2r + 1 - t}} \left(\frac{1}{k \abs{P_t}} \cdot\frac{k^2 (3 \delta n)^{2r + 2}}{(r+1)^2} -  \abs{P_t}\Paren{ (3 \delta n)^{r - t} d^{t}}^2 \right) - \frac{k}{\Delta n^2} \\
&\geq  \frac{1}{O(1)} \left(\frac{k (3 \delta)}{(r+1)^2}- \frac{1}{3 \delta} \right)   - \frac{k}{\Delta n^2} \geq \frac{1}{O(1)} \frac{k (3 \delta)}{(r+1)^2} \enspace,
\end{flalign*}
if we assume that $k \geq \Omega(r^2/\delta^2) = \Omega(\log^2 n/\delta^2)$. Note that if not, then we have that $\delta^2 k \leq O(\log^2 n)$, which is a better lower bound than \cref{mthm:main}.

Therefore,
we have shown that 
\begin{equation*}
\frac{1}{k'} \sum_{(i,j) \in M} \abs{E(G''_{i,j})} \geq \delta' N \enspace,
\end{equation*}
where $k' = \abs{L} \geq \frac{k-1}{2}$ and $\delta' = \Omega(\delta/r^2)$. Hence, by \cref{fact:2ldclb}, it follows that 
\begin{flalign*}
&O(\ell r \log n) \geq 2 \log_2 N \geq \delta' k \geq \Omega(\delta k/r^2) \\
&\implies k \leq  O(\ell r^3 \log n/\delta)  \leq O(\log^8 n/\delta^2) \enspace,
\end{flalign*}
i.e., $2^{\Omega((\delta^2 k)^{1/8})} \leq n$. This finishes the proof of \cref{mthm:main} for the case of $\F = \F_2$, up to the proof of \cref{lem:rowpruning}.

\section{Row Pruning: Proof of \cref{lem:rowpruning}}
\label{sec:rowpruning}
In this section, we prove \cref{lem:rowpruning}, restated below, which is the main technical component in the proof of \cref{mthm:main}. 

\RowPruning*

For $i \ne j \in [k]$ and a row $\vec{S}$, let $\deg_{i,j}(\vec{S})$ denote the number of nonzero entries in the $\vec{S}$-th row.
The main idea of the proof is to observe that for any $i,j$, $\deg_{i,j}(\vec{S})$ is upper-bounded by a $(2r + 2 - t)$-partite polynomial $\Deg_{i,j}(s^{(1)}, s^{(2)}, \ldots, s^{(2r+2-t)})$ in $n(2r+2-t)$ variables $s^{(h)}_u$ for $1 \leq h \leq 2r + 2 - t$ and $u \in [n]$ that define $\vec{S}$, i.e., $s^{(h)}$ (in $\Bits^n$) represents the $0-1$ indicator vector of $S_h$ (or $S'_h$ or $R_h$, depending on the value of $h$).
The contiguous regularity property allows us to control the expected partial derivatives of $\Deg_{i,j}$ and thus apply the tail bounds for partite polynomials in \cref{lem:partitepolyconc}. 

Let us first set up the polynomial $\Deg_{i,j}$ formally.  

For $\vec{C} \in \cH^{(r+1)}_{i,Q,p}$ and $\vec{C'} \in \cH^{(r+1)}_{j,Q,p}$, let $\cT_{i,j,Q,p}^{(\vec{C}, \vec{C'})}$ denote the set of $(2r + 2 - t)$-tuples\\ $(u_0, \dots, u_{r - t}, u_{r - t+1}, \dots, u_{r}, v_{0}, \dots, v_{r - t})$ such that for $h = 0, \dots, r - t$, $u_h \in C_h$ and $v_h \in C'_h$, and for $h = 1, \dots, t$, we have $u_{r - t + h} \in C_h \setminus Q_h$. 
For a row $\vec{S}$ and a tuple $U = (u_0, \dots, u_{r - t}, u_{r - t+1}, \dots, u_{r}, v_{0}, \dots, v_{r - t})$, we write $U \in \vec{S}$ to mean that $u_h \in S_h, v_h \in S'_h$ for $h = 0, \dots, r - t$ and $u_{r - t + h} \in R_h$ for $h = 1, \dots, t$.

We next make an easy observation about the structure of the matrices $A_{i,j,Q,p}^{(\vec{C},\vec{C'})}$.

\begin{observation}
Every row of $A_{i,j,Q,p}^{(\vec{C},\vec{C'})}$ has at most $1$ non-zero entry. Further, for every non-zero row $\vec{S}$, there is a unique $(2r+2-t)$-tuple $U \in \cT_{i,j,Q,p}^{\vec{C},\vec{C'}}$ such that $U \in \vec{S}$. Finally, $U \in \vec{S}$ does not guarantee a non-zero entry. 
\end{observation}

Let $\cT_{i,j} = \bigcup_{(Q,p) \in P_t} \bigcup_{\vec{C} \in \cH^{(r+1)}_{i,Q,p}, \vec{C'} \in\cH^{(r+1)}_{j,Q,p}} \cT_{i,j,Q,p}$. Then, by the above observation, the number of nonzero entries in the $\vec{S}$-th row of $A_{i,j}$ is upper bounded by the number of tuples $U \in \cT_{i,j}$ with $U \in \vec{S}$. Define the following polynomial $\Deg_{i,j}$ that counts this latter quantity, as follows.
 
Let $s=(s^{(0)}, \dots, s^{(r - t)}, s^{(r - t + 1)}, \dots, s^{(r)}, s'^{(0)}, \dots, s'^{(r - t)})$ be a partitioned set of ${0,1}$-valued variables where each $s^{(h)},s'^{(h)}$ is an $n$-tuple $(s^{(h)}_u)_{u \in [n]}$. We view $s$ as the tuple of $0$-$1$ indicators for the sets appearing in a $\vec{S}$. Formally, we have
\begin{equation*}
\Deg_{i,j}(s) \coloneqq \sum_{U \in \cT_{i,j}} \prod_{h = 0}^{r} s^{(h)}_{u_h} \prod_{h = 0}^{r - t} s'^{(h)}_{v_h} \enspace.
\end{equation*}

Let $\cD$ denote the uniform distribution over the rows $\vec{S}$, i.e., each $s^{(h)}$ is drawn independently and uniformly at random from $\Bits^n$ conditioned on $\norm{s^{(h)}}_1 = \ell$. Thus, to bound the fraction of rows with a large number of nonzero entries, it suffices to prove bounds on the tail probability of $\Deg_{i,j}$ on $\cD$. As $\cD$ is not quite a product distribution, we cannot directly apply \cref{lem:partitepolyconc}. Nonetheless, we shall show that its tail bounds behave like those for a product distribution, via the following coupling lemma.

Let $\cD'$ denote the distribution where each $s^{(h)}_u,s'^{(h)}_v$ are chosen independently as a $p$-biased Bernoulli random variable where $p = (1 + \beta)\ell/n$ independently for $\beta = \frac{1}{4r}$. The following lemma relates tail bounds for $\Deg_{i,j}$ on $\cD'$ with those on $\cD$.
 
 \begin{lemma}[Coupling]
 \label{lem:coupling}
 We have $\Pr_{s \gets \cD}[\Deg_{i,j}(s) \geq \Delta] \leq \Pr_{s \gets \cD'}[\Deg_{i,j}(s) \geq \Delta] + (2r + 2) \exp(-\ell/64r^2)$.
 \end{lemma}
 \begin{proof}[Proof of \cref{lem:coupling}]
 To relate the two probabilities, we will couple $\cD'$ with $\cD$ as follows. 
 First, sample $s \gets \cD'$. Then, for each $h = 0, \dots, r$, set $s^{(h)}$ to be a uniformly random subset of $s^{(h)}$ (if one exists), and similarly for $h = 0, \dots, r - t$, set $s'^{(h)}$ to be a uniformly random subset of $s'^{(h)}$ of size $\ell$ also. If one of the sets $S_h$ or $S'_h$ has size $< \ell$, i.e., $\norm{s^{(h)}}_1 < \ell$ for some $h \in \{0, \dots, r\}$ or $\norm{s'^{(h)}}_1 < \ell$ for some $h \in \{0, \dots, r - t\}$, then the coupling fails and we abort. Let $\cJ$ be the joint distribution induced by this coupling.
 
Fix $h \in \{0, \dots, r\}$. By Chernoff bound, we have for every $\delta \in [0,1]$ and for any $h$, 
\begin{flalign*}
\Pr_{s \sim \cD'}[\norm{s^{(h)}}_1 < (1 - \delta)(1 + \beta) \ell] \leq \exp\Paren{\frac{\delta^2 \ell (1 + \beta)}{2}} \enspace.
\end{flalign*}
Setting $\delta = 1 - \frac{1}{1 + \beta}$ and noting that $\beta = \frac{1}{4r} < 1$, we see that $\Paren{\frac{\delta^2 \ell (1 + \beta)}{2}} = \frac{\ell}{2(4r + 16r^2)} \geq \frac{\ell}{64 r^2}$. Hence, the probability that $\cJ$ aborts is at most $(2r + 2) \exp(-\ell/64r^2) \leq \ell^{-\Gamma r}$. Here, we use \cref{item:couplingbound} in the assumptions of the parameters in \cref{lem:rowpruning}.

We also observe that $\Deg_{i,j}$ is monotone, that is, $\Deg_{i,j}(s) \geq \Deg_{i,j}(s')$ for \emph{any} Boolean variables $s,s'$ where $s' \leq s$ coordinate-wise. In particular, if we first sample $s \gets \cD'$ and it holds that $\Deg_{i,j}(s) \leq \Delta$, then it also holds that $\Deg_{i,j}(s') \leq \Delta$ also, regardless of the choice of $s'$ made by the coupling $\cJ$. We thus have
\begin{equation*}
\Pr_{s' \gets \cD}[\Deg_{i,j}(s') > \Delta] \leq \Pr_{(s,s')\sim \cJ}[\Deg_{i,j}(s) > \Delta \mid \text{ $\cJ$ does not abort}] \leq \Pr_{s \gets \cD'}[\Deg_{i,j}(s) > \Delta] + \ell^{-\Gamma r} \mper
\end{equation*}
This completes the proof.
 \end{proof}

We finally obtain a tail bound on $\Deg_{i,j}$ for the product distribution $\cD'$ to complete the proof.
 \begin{lemma}
 \label{lem:app-ss}
 For $\Delta = 9 \cdot 2^{2r + 2 - 2t} (\ell/n)^{2r + 2 - t} \cdot  d^t (3 \delta n)^{2r + 1 - t}$, we have $\Pr_{s \gets \cD'}[\Deg_{i,j}(s) \geq \Delta] \leq \ell^{-\Gamma r}$.
 \end{lemma}
  \begin{proof}[Proof of \cref{lem:app-ss}]
 We will apply \cref{lem:partitepolyconc} to bound $\Pr_{s \gets \cD'}[\Deg_{i,j}(s) \geq \Delta]$. Note that $\Deg_{i,j}$ is homogeneous, multilinear, $(2r+2-t)$-partite polynomial. To apply \cref{lem:partitepolyconc}, we will now bound the expected partial derivatives $\mu_Z$ of $\Deg_{i,j}$ for each tuple $Z \in \{[n] \cup \{\star\})^{2r + 2 - t}$. 

Let 
\begin{equation}
\mu = 3 \cdot 2^{2r + 2 - 2t} (\ell/n)^{2r + 2 - t} \cdot  d^t (3 \delta n)^{2r + 1 - t} = 3 \cdot 2^{2r+2-2t} \frac{\ell}{n} (3 \delta \ell)^{2r+1-t} d^t \mper
\end{equation} 

\begin{claim}[Bounding Expected Partials] \label{claim:bounding-partials}
Let $\gamma= \frac{c}{\Gamma r^3 \log_2 n}$, $\ell \geq \frac{4n^{1/r}}{3\delta \gamma^4}$, $d = 3 \gamma \delta \ell$. Then, for any $h = 0, \dots, 2r + 2 - t$, we have $\mu_Z \coloneqq \mu_Z(\Deg_{i,j}) \leq \mu \cdot \gamma^{\abs{Z}}$.
\end{claim}
We postpone the proof of \cref{claim:bounding-partials}, and now finish the proof of \cref{lem:rowpruning} by using \cref{lem:partitepolyconc}. Applying of \cref{lem:partitepolyconc} with $\beta = 1/(2r+2)$ and $\gamma$, we see that
\begin{flalign*}
 &\Pr_{x \gets \cD'}[P_{i,j}(x) \geq 3 \mu] \leq r(n+1)^r \exp\left(-\frac{\frac{1}{2 (2r + 2)^2}}{2 \gamma + \frac{1}{3 (2r + 2)} \gamma } \right)\\
 & \leq r(n+1)^r \exp\left(-\frac{1}{24 \gamma (r + 1)^2 } \right) \leq \ell^{-\Gamma r} \enspace,
 \end{flalign*}
as $3 \geq (1 + \frac{1}{2r+2})^{2r + 2} \geq (1 + \frac{1}{2r+2})^{2r + 2 - t}$ and $\gamma \leq \frac{1}{c \Gamma r^3 \log_2 n}$ by \cref{item:smallgamma} of our parameter assumptions.
\end{proof}

It thus remains to prove \cref{claim:bounding-partials}. 
\begin{proof}[Proof of \cref{claim:bounding-partials}] 
For $U \in \cT_{i,j}$, we say $Z \subseteq U$ if $Z$ and $U$ agree on all non-$\star$ entries of $Z$. We let $\deg_{i,j}(Z)$ denote the number of tuples $U \in \cT_{i,j}$ where $Z \subseteq U$. Note that $\mu_Z = p^{2r + 2 - t - \abs{Z}} \deg_{i,j}(Z)$. Let's now estimate $\deg_{i,j}(Z)$ -- which equals the number of triples $(U, \vec{C}, \vec{C'})$ where $U \in \cT_{i,j}^{(\vec{C}, \vec{C'})}$ and $Z \subseteq U$.

Fix a $Z$ and let $Z_1$ denote the first $r+1$ entries, and $Z_2$ denote the last $r + 1 - t$ entries. First, we argue that there are at most $2^{\abs{Z_1}} (3 \delta n)^{r + 1 - \abs{Z_1}}$ choices for $\vec{C} \in \cup_{(Q,p) \in P_t}\cH^{(r+1)}_{i, Q,p}$ for which $Z_1$ is contained in $\vec{C}$. To see why, consider choosing $\vec{C}$ iteratively. Given the first $h-1$ choices, let's now consider the $h$-th choice. If $(Z_1)_h$ is a $\star$, then there are $\delta n$ choices for the hyperedge $C_h \cup \{w_h\} \in H_{w_{h-1}}$, as we already know $w_{h-1}$ (when $h = 0$, $w_{-1} \coloneqq i$ is fixed). Then, there are $3$ choices for $w_h$ within this hyperedge. If $(Z_1)_h=u_h \neq \star$, then, there is a unique hyperedge in $H_{w_{h-1}}$ containing $u_h$. This hyperedge has two other vertices that could be chosen as $w_h$. Hence, we have $2^{\abs{Z_1}} (3 \delta n)^{r + 1 - \abs{Z_1}}$ choices in total. Observe that once we have chosen $\vec{C}$, we also know the index $(Q,p)$ of the partition in the refinement that $\vec{C}$ comes from. 

Next, let's count the number of partial tuples $U =(u_0, \dots, u_r)$ that we can produce from this $\vec{C}$. 
For each non-$\star$ entry of $Z_1$, we know $u_h=(Z_1)_h$. For each $h$ where $(Z_1)_h = \star$, if $h \in \{0, \dots, r - t\}$, then we only know $u \in C_h$, which gives us $\abs{C_h} = 2$ choices for $u_h$. If $h \in \{r - t + 1, \dots, r\}$, then we know that $u_h$ must equal $C_h \setminus Q_{h - (r - t)}$ -- a unique choice. We thus pay an additional $2^{r + 1 - t - \abs{Z'_1}}$, where $Z'_1$ is the partial tuple $((Z_1)_{0}, \dots, (Z_1)_{r - t})$, to determine $(u_0, \dots, u_r)$. 

We now have two cases.
\begin{enumerate}[(1)]
\item Case 1: $Z_2$ has no $\star$ entries, i.e., $\abs{Z_2} = r + 1 - t$. This implies that $(v_0, \dots, v_{r - t}) = ((Z_2)_0, \dots, (Z_2)_{r - t})$, and so we have uniquely determined $U$. By an argument similar to above, we also have at most $2^{r+1}$ choices for $\vec{C'} \in \cH^{(r+1)}_{j,Q,p}$ (recall that we already know $Q$, which determines $v_{r + 1 - t}, \dots, v_{r}$ up to $2^{t}$ choices). Hence, we have argued in this case that $\deg_{i,j}(Z) \leq 2^{2r + 2}  (3 \delta n)^{r + 1 - \abs{Z_1}}$, where we use that $\abs{Z_1} \leq t + \abs{Z'_1}$. 

It then follows that 
\begin{flalign*}
&\mu_Z = p^{2r + 2 - \abs{Z}} \deg_{i,j}(Z) \leq p^{2r + 2 - t - (r + 1 - t) - \abs{Z_1}} 2^{2r + 2}  (3 \delta n)^{r + 1 - \abs{Z_1}} \\
&\leq (1 + \beta)^{2r + 2} \cdot (\ell/n)^{r + 1 - \abs{Z_1}} 2^{2r+2} (3 \delta n)^{r + 1 - \abs{Z_1}} \\
&\leq \left(1 + \frac{1}{4r}\right)^{4r} 2^{2r + 2} (3 \delta \ell)^{r + 1 - \abs{Z_1}} \\
&\leq 3 \cdot 2^{2r + 2} \cdot (3 \delta \ell)^{r + 1 - \abs{Z_1}} \enspace.
\end{flalign*}
Now, we observe that since $\mu = 3 \cdot 2^{2r+2-2t} \frac{\ell}{n} (3 \delta \ell)^{2r+1-t} d^t$ and $d = 3 \delta \ell \gamma$, we have that
\begin{flalign*}
&\frac{\mu_Z}{\gamma^{\abs{Z}} \mu} = \frac{ 3 \cdot 2^{2r + 2} \cdot (3 \delta \ell)^{r + 1 - \abs{Z_1}}}{ \gamma^{\abs{Z_1} + r + 1 - t} \cdot 3 \cdot 2^{2r+2-2t} \frac{\ell}{n} (3 \delta \ell)^{2r+1-t} d^t} = \frac{ 2^{2t} n}{ \gamma^{\abs{Z_1} + r + 1} \cdot \ell (3 \delta \ell)^{r + \abs{Z_1}}}  \\
&\leq \frac{ 2^{2(r+1)} n}{ \gamma^{\abs{Z_1} + r + 1} \cdot (3 \delta \ell)^{r + 1 + \abs{Z_1}}} \leq 1 \enspace,
\end{flalign*}
provided that $3 \delta \ell \geq 1$ and $(3 \delta \ell \gamma/4)^{r+1} \geq n$, which hold by \cref{item:largedensity} of the parameter assumptions.

\item Case 2: $Z_2$ has at least one $\star$ entry. In this case, let us write $Z_2 = (Z'_2, \star, Z''_2)$, where $Z''_2$ does not contain any $\star$. Note that $Z''_2$ may be empty, i.e., have length $0$. 

We observe that there are at most $2^{\abs{Z'_2}} (3 \delta n)^{r + 1 - t - (\abs{Z''_2} + 1) - \abs{Z'_2}}$ choices for the partial chain\\ $(j,C'_0, w_0, \dots, C'_{h'}, w'_{h'})$ where $h' = r - t - (\abs{Z''_2} + 1)$ (i.e., the number of entries in $Z'_2$). As in Case 1, we argue inductively and consider the step when we have chosen $j,C'_0,w_0,\ldots, C'_h, w'_h$ for some $0\leq h < h'$. If $(Z'_2)_h= \star$, then, there are $\delta n$ choices for choosing the next hyperedge and $3$ choices for deciding the $w'_{h+1}$ within it giving a total of $3\delta n$ choices. If $(Z'_2)_h \neq \star$, then there is at most one hyperedge (so no choice to be made) in $H_{w'_h}$ that could appear as the next link and, given the hyperedge, there are $2$ choices for the $w'_{h+1}$. 

Given the first $h'$ links in the partial chain, we have at most $2^{r + 1 - t - (\abs{Z''_2} + 1) - \abs{Z'_2}}$ choices for the partial tuple $(v_0,\ldots, v_{h'})$. So in total, we have $2^{r + 1 - t - (\abs{Z''_2} + 1)}(3 \delta n)^{r + 1 - t - (\abs{Z''_2} + 1) - \abs{Z'_2}}$ choices for the partial chain $(j,C'_0, w_0, \dots, C'_{h'}, w'_{h'})$ and the partial tuple $(v_0,\ldots, v_{h'})$.

To count the number of ways to complete the chain, we break our analysis into two subcases. 

\begin{enumerate}[(a)]
\item Subcase 1: $Z''_2$ is empty, and so $\abs{Z_2} = \abs{Z'_2}$. In this case, $h' = r - t - 1$. Since we have already chosen $(Q,p)$, the number of different choices for the partial chain $(w'_{r-t},C'_{r - t +1}, w'_{r - t + 1}, \dots, C'_{r}, w'_{r})$ must be at most $d^t$, by Item (4) in \cref{def:decomp} with $Q' = Q$. Given this choice, $w'_{r-t}$ is fixed so there is at most one choice for a hyperedge in $H_{w'_{r-t-1}}$ that contains $w'_{r-t}$ and given that choice, there are two possible ways to choose $v_{r-t}$. In total, we have made at most $2 d^t$ choices.

In the case that $t = 0$, the partial chain is the ``$0$-chain'' given by $Q_{r+1} = w'_r$, and as we have $w'_r = Q_{r+1} = w_r$, this gives a unique choice for the ``chain'', i.e., $d^{0} = 1$ choices.
\item Subcase 2: $Z''_2$ is nonempty. We observe that for $h'= r-t-(\abs{Z''_2}+1)$, the partial chain $(w'_{h+1},C'_{h' + 2}, w'_{h' + 2}, \dots, C'_{r}, w'_{r})$ must contain the \emph{complete} tuple $Z''_2 \| Q^{(t)}_z$, where $\cdot \| \cdot$ denotes concatenation. Thus, by $d$-regularity, there are at most $d^{t + \abs{Z''_2}}$ choices of such tuples. Given the choice of this partial chain, there are $2^{r - t - h'}$ choices for $(v_{h'+1}, \dots, v_{r - t})$. Hence, in total we have made $2^{\abs{Z''_2} + 1} d^{t + \abs{Z''_2}}$ choices.
\end{enumerate}
We note that in either subcase, we make at most $2^{\abs{Z''_2} + 1} d^{t + \abs{Z''_2}}$ choices to pick $\vec{C'} \in \cH^{(r+1)}_{j,Q,p}$ and $(v_{h'+1}, \dots, v_{r - t})$, where we  can have $\abs{Z''_2} = 0$. Thus, the total number of choices of $\vec{C'}$ and $(v_0,\ldots,v_{r-1})$ is at most
\begin{equation*}
2^{r + 1 - t - (\abs{Z''_2} + 1)}(3 \delta n)^{r + 1 - t - (\abs{Z''_2} + 1) - \abs{Z'_2}} 2^{\abs{Z''_2} + 1} d^{t + \abs{Z''_2}}  =  2^{r + 1 - t }   (3 \delta n)^{r  - t - \abs{Z_2}} \cdot  d^{t + \abs{Z''_2}} \mcom
\end{equation*}
and thus, the total number of triples $(U,\vec{C},\vec{C'})$ that contribute to $\deg_{i,j}(Z)$ is at most\\ $2^{r + 1 - t - \abs{Z'_1} + \abs{Z_1} } (3 \delta n)^{r + 1 - \abs{Z_1}}  2^{r + 1 - t }   (3 \delta n)^{r  - t - \abs{Z_2}} \cdot  d^{t + \abs{Z''_2}}$.

Thus,
\begin{flalign*}
&\mu_Z \leq p^{2r + 2 - t - \abs{Z_1} - \abs{Z_2}} 2^{r + 1 - t - \abs{Z'_1} + \abs{Z_1} } (3 \delta n)^{r + 1 - \abs{Z_1}}  2^{r + 1 - t }   (3 \delta n)^{r  - t - \abs{Z_2}} \cdot  d^{t + \abs{Z''_2}} \\
&\leq  (1 + \beta)^{2r + 2} 2^{2r + 2 - 2t + \abs{Z_1}} (\ell/n)^{2r + 2 - t - \abs{Z_1} - \abs{Z_2}}  (3 \delta n)^{2r + 1 - t - \abs{Z_1} - \abs{Z_2}} \cdot  d^{t + \abs{Z''_2}} \\
&\leq  (1 + \beta)^{2r + 2} 2^{2r + 2 - 2t + \abs{Z_1}} (\ell/n)  (3 \delta \ell)^{2r + 1 - t - \abs{Z_1} - \abs{Z_2}} \cdot  d^{t + \abs{Z''_2}} \\
&\leq  \mu \frac{1}{3} (1 + \beta)^{2r + 2} 2^{\abs{Z_1}} (3 \delta \ell)^{- \abs{Z_1} - \abs{Z_2}} \cdot  d^{ \abs{Z''_2}} \\
&\leq \mu  \left(\frac{2}{3 \delta \ell}\right)^{\abs{Z_1}} (3 \delta \ell)^{-\abs{Z'_2}} \left(\frac{d}{3 \delta \ell}\right)^{\abs{Z''_2}} \\
&\leq \mu  \left(\frac{2}{3 \delta \ell}\right)^{\abs{Z_1}} (3 \delta \ell)^{-\abs{Z'_2}} \gamma^{\abs{Z''_2}} \enspace,
\end{flalign*}
using that $(1+\beta)^{2r+2}\leq 3$ and $d =3 \gamma \delta \ell$. We thus have
\begin{equation*}
\mu_Z \leq \mu \max\{\frac{2}{3 \delta \ell}, \gamma\}^{|Z|} \leq \gamma^{|Z|} \mu\mcom 
\end{equation*}
where we use \cref{item:largeell} in the parameter assumptions.

\end{enumerate}
This finishes the proof of \cref{claim:bounding-partials}.
\end{proof}

\section{Discussion}
\label{sec:discussion}
We conclude with some remarks on the proof of \cref{mthm:main}, possible strengthenings, and extensions. 

\begin{enumerate}[(1)]
\item \textbf{Non-linear codes.} The lower bound in \cref{mthm:main} applies only to linear codes. However, we note that we only use linearity of the code to argue a lower bound on $\val(\Phi_b)$, the XOR instance polynomial for $(r+1)$-chains. For the natural XOR instances (i.e., when $r=0$), a lower bound on $\val(\Phi_b)$ easily follows even for non-linear codes. This is the reason why the $3$-LDC lower bounds in~\cite{AlrabiahGKM23} apply to non-linear codes. The issue (that nevertheless appears surmountable) that prevents us from obtaining a similar lower bound on $\val(\Phi_b)$ for XOR instances with chains of length $>1$ is the following: for non-linear codes, we are only guaranteed that each constraint is satisfied for a non-trivial constant fraction of codewords. That is, $\E_{x \gets \Code}[x_C x_u] \geq \eps$ for some constant $\eps>0$ (for linear codes, we instead obtain $x_C x_u = 1$ for all $x \in \Code$). In particular, it is not clear that $\E_{b}[\Phi_b(\Code(b))]$ is non-trivially lower-bounded. 

\item \textbf{LCCs with more queries.} While our approach can likely improve the lower bounds (beyond those known for LDCs) even for $q>3$, the improvements based on natural generalizations of our approach are likely to only yield a polynomial factor improvement. Our explanation is rooted in the heuristic calculation based on the density of the Kikuchi matrices explained earlier in \cref{sec:chainheuristic}. For larger $q$, the number of length $(r+1)$-chains with head $i \in [k]$ is still $k (3 \delta n)^{r+1}$. The arity of the derived constraints, however, is now $(q-1)(r+1) + 1$. This means that the density (i.e., average degree of the natural Kikuchi matrix) at level $\ell$ is $(3 \delta n)^{r+1} (\ell/n)^{\frac{(q-1)(r+1) + 1}{2}} \to \left(n (\ell/n)^{\frac{q-1}{2}}\right)^{r+1}$ for large $r$. Thus, the optimal $\ell$ turns out to be $n^{1 - \frac{2}{q-1}}$, and so we can only hope to achieve a lower bound of $k \leq \tilde{O}(n^{1 - \frac{2}{q-1}})$. This nevertheless would yield an improvement on the current best-known lower bound of $k \leq \tilde{O}(n^{1-\frac{2}{q}})$, inherited from $q$-LDCs, by a polynomial factor via long chains.

\item \textbf{Optimality of Reed--Muller codes?} Our main result \cref{mthm:main} comes close to showing that Reed--Muller codes, which achieve a blocklength of $n = 2^{O(\sqrt{k})}$, are optimal linear $3$-LCCs --- a longstanding goal in understanding LCCs. Closing the gap between our result and the blocklength of Reed--Muller codes relates to optimizing the $\polylog(n)$ factors in our analysis. Let us now explain each of the $\log n$ factors that we ``lose'' with an eye for the losses that appear naturally surmountable and ones that appear rather inherent. 

First, we note that we must take the chain length $r$ to be $\geq O(\log n)$ and $\ell \geq 1/\delta$ for the heuristic calculation in \cref{sec:chainheuristic} to work. Second, we note that the application of matrix Khintchine (\cref{fact:matrixkhintchine}) loses a $\sqrt{\log N} = \sqrt{\ell r \log n}$ factor. Thus, in the ideal case,  our method could potentially yield that $k \leq O(\log N)$ where $\ell \sim 1/\delta$ and $r = O(\log n)$. This would yield a bound of $k \leq O(\log^2 n)$, or in other words $n \geq 2^{\Omega(\sqrt{k})}$, matching the blocklength of Reed--Muller codes up to constant factors in the exponent.

However, our argument currently loses additional $\log n$ factors that appear improvable. First, the hypergraph decomposition step loses a factor of $r$ in the ``density'' because we need to refute at least one of the $\sim r$ subinstances produced each of which may only have $1/(r+1)$-fraction of all the $(r+1)$-chains. This loses us $O(r^2)$ factor in the density once we use the Cauchy--Schwarz trick. Second, we cannot take $\ell$ to be as small as $1/\delta$, i.e., a constant. Currently, we need to take $\ell\geq O(\log^4 n)$ for the tail bounds used in the proof of \cref{lem:rowpruning} to be effective. 

These additional $\log n$ factors that we lose not appear to be inherent to our approach. To save these losses would require a sharper chain decomposition method (that does not lose a factor $r$ in the density) and a sharper concentration bound than \cref{lem:partitepolyconc}. While these appear technically challenging, it does look plausible that one remove these additional $\log n$ factors and obtain a lower bound that matches the blocklength of Reed--Muller codes up to absolute constant factors in the exponent. 
\end{enumerate}

\section*{Acknowledgements}
We thank Venkatesan Guruswami for detailed feedback on an earlier version of this manuscript. We thank Zeev Dvir and Hans Yu for suggesting related works and helpful discussions. 


\bibliographystyle{alpha}
{\small
\bibliography{references}
}

\appendix


\section{Linear $3$-LCC Lower Bounds over Larger Fields}
 \label{sec:largeralphabet}
In this section, we prove \cref{mthm:main} in the case where the finite field $\F$ is not $\F_2$. The proof will be nearly identical to the proof in \cref{sec:lcctoxor,sec:regular-partition,sec:kikuchimethod,sec:rowpruning} for the case of $\F = \F_2$, and so we shall only give a proof sketch and mainly focus on the parts of the proof where modifications are required.

To begin, we recall that by \cref{fact:normalform}, there exist $3$-uniform hypergraph matchings $H_1, \dots, H_n$, each of size at least $\delta n$, such that for each $u \in [n]$ and $C  = \{v_1, v_2, v_3\} \in H_u$, there exists $\alpha_1, \alpha_2, \alpha_3 \in \F \setminus \{0\}$ such that for every $x \in \Code$, it holds that $\alpha_{1} x_{v_1} + \alpha_{2} x_{v_2} + \alpha_{3} x_{v_3} = x_u$. Furthermore, without loss of generality we can assume that the code is systematic, i.e., for any $b \in \F^k$, $x = \Code(b)$ satisfies $x_i = b_i$ for all $i \in [k]$.

Next, let us define a code $\Code' \colon \Fits^k \to \Fits^{n(\abs{\F} - 1)}$ where, for each $u \in [n]$ and $\alpha \in \F \setminus \{0\}$, we set $\Code'(b)_{(u, \alpha)} = \alpha \Code(b)_u$. Let $n' = n(\abs{\F} - 1)$, and associate $[n']$ with the set $[n] \times (\F \setminus \{0\})$. We now observe that $\Code'$ is a $3$-LCC in normal form with the additional property that the coefficients of all constraints can be taken to be $1$ without loss of generality. Formally, there exist $3$-uniform hypergraph matchings $H_1, \dots, H_{n'}$ such that (1) each $H_u$ has $\abs{H_u} \geq \delta n'/(\abs{\F} - 1)$, and (2) for each $u \in [n']$ and each $C = \{v_1, v_2, v_3\} \in H_{(u,\alpha)}$, every $x \in \Code$ satisfies $x_{u} = x_{v_1} + x_{v_2} + x_{v_3}$.

Moreover, there is now a group action of $(\F \setminus \{0\}, \times)$ on the elements of $[n']$, namely for any $\alpha \in \F \setminus \{0\}$, this action maps $u \mapsto \alpha u$. We note that this action respects the constraints. Namely, for $C = \{v_1, v_2, v_3\} \in {[n'] \choose 3}$, if we define $\alpha C = \{\alpha v_1, \alpha v_2, \alpha v_3\}$, then we have that $H_{\alpha u} = \alpha H_u = \{\alpha C : C \in H_u\}$. For the proof, we will be using the fact that there is a negation action for $\alpha = -1$; this is because this transformation has made all coefficients in the constraints be equal to $1$, so to cancel a variable $x_u$ we shall only need $x_{-u}$.

We shall now abuse notation and redefine $n'$ to be $n$, and we now simply assume that we have this group action on $[n]$. We have thus added this additional property to the code, and in doing so we have only decreased $\delta$ by a factor of $\abs{\F} - 1$.

We now turn to the main part of the proof. Following \cref{sec:lcctoxor}, we define $t$-chains. The definition of $t$-chains now requires a small modification because in the original definition we formed longer chains by canceling a variable $x_w$ via the operation $x_w + x_w = 0$, which was specific to the field $\F_2$. Now, we use the negation action on $[n]$ to cancel a variable.

\begin{definition}[$t$-chain hypergraph $\cH^{(t)}$]
Let $t \geq 1$ be an integer. For any $u \in [n]$, let $\cH_u^{(t)}$ denote the set of tuples of the form $(u,C_1, w_1, C_2, w_2, \dots, C_t, w_t)$, where each $C_h \in {[n] \choose 2}$, $w_h \in [n]$, and it holds that for all $1 \leq h \leq t$, $C_h \cup \{w_h\} \in H_{-w_{h - 1}}$ where we set $w_0 \coloneqq u$.

Given any $t$-chain $(u,C_1, w_1, C_2, w_2, \dots, C_t, w_t)$, we let the negation of the chain, denoted by $-(u,C_1, w_1, C_2, w_2, \dots, C_t, w_t)$, be the chain $(-u,-C_1, -w_1, -C_2, -w_2, \dots, -C_t, -w_t)$.
\end{definition}
As before, we note that the linear equation defined by a $t$-chain or its negation is satisfied by any $x \in \Code$.

In \cref{sec:lcctoxor}, we defined an instance polynomial $\Phi_b$ related to the system of linear constraints. This was natural over $\F_2$ as there is a group isomorphism between $(\F_2, +)$ and $\Fits \in (\R, \times)$. Here, we can make a similar definition by using a nontrivial group homomorphism $\pi$ from $(\F,+)$ to $(\C, \times)$ where the image of $\pi$ is contained in the unit circle $\{z \in \C : \abs{z} = 1\}$. However, the instance polynomial $\Phi_b$ (and the ``decomposed polynomials'' $\Psi_{i,Q,p}$ defined later) were only formally needed to discuss sets of linear constraints that are satisfied by the subspace $\Code$. Thus, to avoid using the group homomorphism $\pi$, here we shall simply use these polynomials to refer to the underlying sets of constraints.

We now perform the hypergraph decomposition step as in \cref{sec:regular-partition}, which is unchanged (once we use the updated definition of chain).\footnote{We note that the naive application of the decomposition step will produce partitions $\cH^{(r)}_{Q,p}$ where $\cH^{(r)}_{-Q, p}$ is not necessarily equal to $-\cH^{(r)}_{Q,p}$. This turns out to not matter in the proof; as it turns out, we merely need that both decompositions $\cup_{Q, p} -\cH^{(r)}_{Q,p}$ and $\cup_{Q, p} \cH^{(r)}_{Q,p}$ are both contiguously regular partitions of $\cH^{(r)}$, which obviously holds. Nonetheless, we note that one could also easily modify the decomposition step to respect this negation action.} This produces the subinstances $\Psi^{(t)}(x,y)$, as before.

We now finish the proof following~\cref{sec:step6} in \cref{sec:kikuchimethod}. We let $t$ denote the value $0 \leq t \leq r$ such that $\Psi^{(t)}$ contains at least $k (3 \delta n)^{r+1}/(r+1)$ constraints. Applying the Cauchy--Schwarz trick, we then have that there exists a maximum directed matching $M$ on $[k]$ such that 
\begin{equation*}
\frac{1}{2k \abs{P_t}} \left(\frac{k (3 \delta n)^{r+1}}{r+1}\right)^2 -  \frac{\abs{P_t}}{2} \Paren{ (3 \delta n)^{r - t} d^{t}}^2
\end{equation*}
is a lower bound on number of constraints in the system of linear equations given by:
\begin{flalign*}
b_i - b_{j} = \sum_{h = 0}^{r - t} x_{C_h} + \sum_{h = 1}^{t} x_{C_{r - t + h} \setminus Q_h} + \sum_{h = 0}^{r - t} x_{-C'_h} + \sum_{h = 1}^{t} x_{-C'_{r - t + h} \setminus -Q_h}
\end{flalign*}
for every $(Q,p) \in P_t$, $(i, C_0, w_0, C_1, w_1, C_2, w_2, \dots, C_r, w_r) \in \cH^{(r)}_{i, Q,p}$, $(j, C'_0, w'_0, C'_1, w'_1, C'_2, w'_2, \dots, C'_r, w'_r) \in \cH^{(r)}_{j, Q,p}$. Here, we let $x_{C_h} \coloneqq \sum_{v \in C_h} x_v$.

As before, the definition of the Kikuchi matrices \cref{def:kikuchi} is nearly identical: we merely swap $b_j$ with $-b_{j}$. Because of this, the key technical part of the argument, namely the row pruning step \cref{lem:rowpruning}, holds without any changes.

Now, we define the code $\Code' \colon \F^k \to \F^{2N}$ identically as before. We let $L = \{ i : (i,j) \in M\}$ denote the ``left halves'' of the edges in the matching $M$, and we define $\Code' \colon \F^L \to \F^{2N}$ to be the same map as before; we simply replace sums with products, as we have not used the homomorphism $\pi$ to embed $\F$ into $\C$. Namely, for $x' = \Code'(b)$, the $\vec{S}$-th entry of $x'$ is $x'_{\vec{S}} = \sum_{h = 0}^{r-t} (x_{S_h}+ x_{S'_h}) + \sum_{h = 1}^{t} x_{R_h}$, and similarly for the $\vec{T}$-th entry, where $x = \Code(b)$.

Now, the same calculation as before shows that $\Code'$ is a $(2, \delta)$-LDC for $\delta' = \Omega(\delta/r^2)$ provided that $\delta^2 k \leq O(\log^2 n)$. Namely, there are matchings $G''_{i,j}$ on $[2N]$ such that (1) for any $(\vec{S}, \vec{T}) \in E(G''_{i,j})$, it holds that $x'_{\vec{S}} - x'_{\vec{T}} = b_i - b_j = b_i$ (as $b_j = 0$ for $j \notin L$), and (2) $\frac{1}{k'} \sum_{(i,j) \in M} \abs{E(G''_{i,j})} \geq \delta' \cdot 2N$, where $k' = \abs{L} \geq \frac{k-1}{2}$.

As before, we now apply \cref{fact:2ldclb}. It follows that
\begin{flalign*}
&O(\ell r \log n) \geq 2 \log_2 N \geq \delta' k \geq \Omega(\delta k/r^2) \\
&\implies k \leq  O(\ell r^3 \log n/\delta)  \leq O(\log^8 n/\delta^2) \enspace.
\end{flalign*}
Recall now that we had redefined $n$ to be $n (\abs{\F} - 1)$ and $\delta$ to be $\delta = \delta/(\abs{\F} - 1)$. Thus, we have that for the original code, $\frac{k\delta^2}{(\abs{\F} - 1)^2} \leq O(\log^8 n)$ provided that $\abs{\F} \leq n$. Note that if $\abs{\F} \geq k$, then \cref{mthm:main} becomes trivial, and so we can assume that $\abs{\F} \leq k \leq n$ (as we always have $k \leq n$). Finally, we note that we have assumed (when we substitute back the original values of $\delta$ and $n$) that $\frac{\delta^2 k}{(\abs{\F} - 1)^2} \leq O(\log^2 (n \abs{\F}))$, which implies that $\frac{\delta^2 k}{(\abs{\F} - 1)^2} \leq O(\log^2 n)$, as we may again assume that $\abs{\F} \leq n$. This is a stronger lower bound than \cref{mthm:main}, and so this finishes the proof of \cref{mthm:main} for larger fields.

\end{document}